\documentclass{amsart}
\usepackage{fullpage,graphicx,mathpazo,color}
\begin{document}

\newtheorem{rhp}{Riemann-Hilbert Problem}
\newtheorem{theorem}{Theorem}
\newtheorem{lemma}{Lemma}
\newtheorem{proposition}{Proposition}
\newtheorem{corollary}{Corollary}
\newtheorem{conjecture}{Conjecture}

\theoremstyle{definition}
\newtheorem{mydef}{Definition}

\newenvironment{remark}{\medskip$\triangleleft$\textit{Remark:}}{$\triangleright$\medskip}

\newcommand{\ii}{\mathrm{i}}
\newcommand{\ee}{\mathrm{e}}
\newcommand{\dd}{\mathrm{d}}

\title{Extreme Superposition:  Rogue Waves of Infinite Order and the Painlev\'e-III Hierarchy}
\author{Deniz Bilman}
\address{Department of Mathematics, University of Michigan, East Hall, 530 Church St., Ann Arbor, MI 48109}
\email{bilman@umich.edu}
\author{Liming Ling}
\address{Department of Mathematics, University of Michigan, East Hall, 530 Church St., Ann Arbor, MI 48109\footnote{Permanent address: Department of Mathematics, South China University of Technology, Guangzhou, China 510641.}}
\email{limingl@umich.edu}
\author{Peter D. Miller}
\address{Department of Mathematics, University of Michigan, East Hall, 530 Church St., Ann Arbor, MI 48109}
\email{millerpd@umich.edu}
\date{\today}

\begin{abstract}
We study the fundamental rogue wave solutions of the focusing nonlinear Schr\"odinger equation in the limit of large order.  Using a recently-proposed Riemann-Hilbert representation of the rogue wave solution of arbitrary order $k$, we establish the existence of a limiting profile of the rogue wave in the large-$k$ limit when the solution is viewed in appropriate rescaled variables capturing the near-field region where the solution has the largest amplitude.  The limiting profile is a new particular solution of the focusing nonlinear Schr\"odinger equation in the rescaled variables --- the rogue wave of infinite order --- which also satisfies ordinary differential equations with respect to space and time.  The spatial differential equations are identified with certain members of the Painlev\'e-III hierarchy.  We compute the far-field asymptotic behavior of the near-field limit solution and compare the asymptotic formul\ae\ with the exact solution with the help of numerical methods for solving Riemann-Hilbert problems.  In a certain transitional region for the asymptotics the near field limit function is described by a specific globally-defined tritronqu\'ee solution of the Painlev\'e-II equation.  These properties lead us to regard the rogue wave of infinite order as a new special function.  
\end{abstract}

\maketitle

\section{Introduction}
The focusing nonlinear Schr\"odinger equation in the form:
\begin{equation}
\ii\frac{\partial\psi}{\partial t} +\frac{1}{2}\frac{\partial^2\psi}{\partial x^2} + (|\psi|^2-1)\psi=0,\quad (x,t)\in\mathbb{R}^2
\label{eq:NLS}
\end{equation}
and subject to the boundary conditions $\psi(x,t)\to 1$ as $|x|\to\infty$ is a model for the study of spatially-localized perturbations of Stokes waves, i.e., uniform periodic wavetrains, in diverse physical systems where \eqref{eq:NLS} arises as a weakly-nonlinear complex amplitude equation.  The exact solution $\psi=\psi_0(x,t)\equiv 1$ consistent with these boundary conditions is called the \emph{background}, and it represents the unperturbed Stokes wave.  One exact solution representing a nontrivial perturbation of the background is the \emph{Peregrine solution} \cite{Peregrine83}
\begin{equation}
\psi=\psi_1(x,t):=1-4\frac{1+2\ii t}{1+4x^2+4t^2},
\end{equation}
which represents a disturbance localized near the origin in both space $x$ and time $t$.  The maximum amplitude of $\psi_1(x,t)$ occurs at the origin $(x,t)=(0,0)$ and has a value of three times the unit background amplitude. As such, Peregrine's solution $\psi_1(x,t)$ is a model for \emph{rogue waves}, i.e., large-amplitude spatio-temporally localized disturbances of a uniform background state.  In general rogue waves are of great interest because they are known to have caused damage to ships and they represent one of the basic modes of nonlinear saturation of the well-known modulational instability of the background $\psi_0(x,t)$.  The latter instability is sometimes called the Benjamin-Feir instability in the context of water waves \cite{KharifP03}.  

The focusing nonlinear Schr\"odinger equation \eqref{eq:NLS} is an integrable nonlinear equation, and it therefore comes with a nonlinear analogue of a linear superposition principle known as a \emph{B\"acklund transformation}.  B\"acklund transformations of solutions can be iterated, especially when the transformation is implemented at the level of the Lax pair eigenfunctions underlying the complete integrability via a so-called \emph{Darboux transformation}.  Iterated B\"acklund/Darboux transformations can produce a zoo of increasingly-complicated solutions of \eqref{eq:NLS}; in particular via a limiting technique known as a \emph{generalized Darboux transformation} \cite{GuoLL12} it is possible to iterate the transformation at the distinguished value of the spectral parameter that produces the Peregrine solution $\psi_1(x,t)$ from the background $\psi_0(x,t)$ producing ``higher-order'' rogue wave solutions of \eqref{eq:NLS}.  Such solutions can resemble multiple copies of the Peregrine solution centered at distant space-time points, but it is also possible to choose the auxiliary parameters introduced at each iteration to concentrate the disturbance near the origin (say).  Thus one arrives at a sequence of ``fundamental'' higher-order rogue wave solutions of \eqref{eq:NLS}, $\psi_k(x,t)$, $k=0,1,2,3,\dots$, in which the effect of nonlinear superposition is maximized in a sense.  These solutions are especially interesting in applications because the spatio-temporal concentration turns out to coincide with large amplitude.

Iterated Darboux transformations of a simple solution such as the background $\psi_0(x,t)$ have both an analytic character and an algebraic character, and the latter is especially popular because it leads to closed-form formul\ae\ in which $\psi_k(x,t)$ is expressed, say, in terms of determinants of matrices with simple entries.  For instance, the following algebraic characterization of $\psi_k(x,t)$  can be found in \cite{GuoLL12}.
Let quantities $F_\ell(x,t)$ and $G_\ell(x,t)$, $\ell\in\mathbb{Z}_{\ge 0}$, be defined by entire generating functions as follows:
\begin{equation}
\begin{split}
(1-\ii\lambda)\frac{\sin((x+\lambda t)\sqrt{\lambda^2+1})}{\sqrt{\lambda^2+1}}&=\sum_{\ell=0}^\infty
\left(\frac{1}{2}\ii\right)^\ell F_\ell(x,t)(\lambda-\ii)^\ell\\
\cos((x+\lambda t)\sqrt{\lambda^2+1})&=\sum_{\ell=0}^\infty
\left(\frac{1}{2}\ii\right)^\ell G_\ell(x,t)(\lambda-\ii)^\ell.
\end{split}
\label{eq:series-expansions}
\end{equation}
It is easy to see that the coefficients $F_\ell(x,t)$ and $G_\ell(x,t)$ are polynomials in $(x,t)$.  Define a $k\times k$ matrix $\mathbf{K}^{(k)}(x,t)$ by
\begin{equation}
K^{(k)}_{pq}(x,t):=\sum_{\mu=0}^{p-1}\sum_{\nu=0}^{q-1}\binom{\mu+\nu}{\mu}\left(F_{q-\nu-1}(x,t)^*F_{p-\mu-1}(x,t)+G_{q-\nu-1}(x,t)^*G_{p-\mu-1}(x,t)\right),\quad 1\le p,q\le k,
\end{equation}
and a $k\times k$ rank-one perturbation $\mathbf{H}^{(k)}(x,t)$ by 
\begin{equation}
H^{(k)}_{pq}(x,t):=-2\left(F_{p-1}(x,t)+G_{p-1}(x,t)\right)\left(F_{q-1}(x,t)^*-G_{q-1}(x,t)^*\right),\quad 1\le p,q\le k.
\end{equation}
We take the the following as a definition.
\begin{mydef}[Fundamental rogue waves]
The fundamental rogue wave solution of \eqref{eq:NLS} of order $k$ is 
\begin{equation}
\psi_k(x,t):=(-1)^k\frac{\det(\mathbf{K}^{(k)}(x,t)+\mathbf{H}^{(k)}(x,t))}{\det(\mathbf{K}^{(k)}(x,t))}.
\label{eq:psi-k-determinants}
\end{equation}
\label{def:rogue-wave}
\end{mydef}
In the Appendix, we show that $\det(\mathbf{K}^{(k)}(x,t))\neq 0$ for all $(x,t)\in\mathbb{R}^2$, so $\psi_k(x,t)$ is well-defined.  The square modulus also has a compact representation as
\begin{equation}
|\psi_k(x,t)|^2 = 1+\frac{\partial^2}{\partial x^2}\ln\det(\mathbf{K}^{(k)}(x,t)).
\end{equation}
The latter equation shows that $\det(\mathbf{K}^{(k)}(x,t))$ is a ``$\tau$-function'' for the fundamental rogue wave solutions.  
We now describe the same solution $\psi_k(x,t)$ from a more analytical perspective.
Let $\Sigma_\mathrm{c}$ denote the vertical line segment connecting the points $\pm\ii$, with upward orientation.  Let $\rho(\lambda)$ be the function analytic for $\lambda\in\mathbb{C}\setminus\Sigma_\mathrm{c}$ satisfying $\rho(\lambda)^2=\lambda^2+1$ and $\rho(\lambda)=\lambda+O(\lambda^{-1})$ as $\lambda\to\infty$.  Let $f(\lambda)$ be the function analytic for $\lambda\in\mathbb{C}\setminus\Sigma_\mathrm{c}$ that satisfies $f(\lambda)^2=(\lambda+\rho(\lambda))/(2\rho(\lambda))$ and $f(\lambda)\to 1$ as $\lambda\to\infty$.  Let $\mathbf{E}(\lambda)$ denote the matrix function defined for $\lambda\in\mathbb{C}\setminus\Sigma_\mathrm{c}$ by
\begin{equation}
\mathbf{E}(\lambda):= f(\lambda)\begin{bmatrix}1 & \ii (\lambda-\rho(\lambda))\\\ii(\lambda-\rho(\lambda)) & 1\end{bmatrix},\quad\lambda\in\mathbb{C}\setminus\Sigma_\mathrm{c}.
\end{equation}
This matrix is analytic in its domain of definition and has unit determinant.  We define the constant orthogonal matrix $\mathbf{Q}$ by
\begin{equation}
\mathbf{Q}:=\frac{1}{\sqrt{2}}\begin{bmatrix}1 & -1\\1 & 1\end{bmatrix},\quad\mathbf{Q}^{-1}=\mathbf{Q}^\top,\quad\det(\mathbf{Q})=1.
\label{eq:Q-define}
\end{equation}
Finally, let $\Sigma_\circ$ denote a clockwise-oriented circular contour centered at the origin and having radius greater than $1$.  
In \cite{BilmanM17} the following Riemann-Hilbert problem was proposed as an alternative characterization of the rogue wave solution of order $k$.  Here and below, we use subscripts $+$/$-$ to refer to boundary values taken on an oriented jump contour from the left/right.  We also make frequent use of the Pauli spin matrices:
\begin{equation}
\sigma_1:=\begin{bmatrix}0&1\\1&0\end{bmatrix},\quad\sigma_2:=\begin{bmatrix}0 & -\ii\\\ii&0\end{bmatrix},\quad\text{and}\quad
\sigma_3:=\begin{bmatrix}1&0\\0&-1\end{bmatrix}.
\end{equation}
\begin{rhp}[Rogue wave of order $k$]
Let $(x,t)\in\mathbb{R}^2$ be arbitrary parameters, and let $k\in\mathbb{Z}_{\ge 0}$.  Find a $2\times 2$ matrix $\mathbf{M}^{(k)}(\lambda;x,t)$ with the following properties:
\begin{itemize}
\item[]\textbf{Analyticity:}  $\mathbf{M}^{(k)}(\lambda;x,t)$ is analytic in $\lambda$ for $\lambda\in\mathbb{C}\setminus(\Sigma_\circ\cup\Sigma_\mathrm{c})$, and it takes continuous boundary values on $\Sigma_\circ\cup\Sigma_\mathrm{c}$.
\item[]\textbf{Jump conditions:}  The boundary values on the jump contour $\Sigma_\circ\cup\Sigma_\mathrm{c}$ are related as follows:
\begin{equation}
\mathbf{M}_+^{(k)}(\lambda;x,t)=\mathbf{M}_-^{(k)}(\lambda;x,t)\ee^{2\ii\rho_+(\lambda)(x+\lambda t)\sigma_3},\quad \lambda\in\Sigma_\mathrm{c},
\label{eq:jump-cut}
\end{equation}
and if $k=2n$, $n\in\mathbb{Z}_{\ge 0}$,
\begin{equation}
\mathbf{M}_+^{(k)}(\lambda;x,t)=\mathbf{M}_-^{(k)}(\lambda;x,t)\ee^{-\ii\rho(\lambda)(x+\lambda t)\sigma_3}\mathbf{Q}\left(\frac{\lambda-\ii}{\lambda+\ii}\right)^{n\sigma_3}\mathbf{Q}^{-1}\mathbf{E}(\lambda)\ee^{\ii\rho(\lambda)(x+\lambda t)\sigma_3},\quad
\lambda\in\Sigma_\circ
\end{equation}
while if instead $k=2n-1$, $n\in\mathbb{Z}_{>0}$,
\begin{equation}
\mathbf{M}_+^{(k)}(\lambda;x,t)=\mathbf{M}_-^{(k)}(\lambda;x,t)\ee^{-\ii\rho(\lambda)(x+\lambda t)\sigma_3}\mathbf{Q}
\left(\frac{\lambda+\ii}{\lambda-\ii}\right)^{n\sigma_3}\mathbf{Q}^{-1}\mathbf{E}(\lambda)\ee^{\ii\rho(\lambda)(x+\lambda t)\sigma_3},\quad\lambda\in\Sigma_\circ.
\end{equation}
\item[]\textbf{Normalization:}  $\mathbf{M}^{(k)}(\lambda;x,t)\to\mathbb{I}$ as $\lambda\to\infty$. 
\end{itemize}
\label{rhp:rogue-wave}
\end{rhp}
It turns out (cf., Proposition~\ref{prop:Equivalence} below) that the rogue wave solution of order $k$ is given in terms of the solution of this problem by the formula
\begin{equation}
\psi(x,t)=\psi_{k}(x,t):=2\ii\lim_{\lambda\to\infty}\lambda M^{(k)}_{12}(\lambda;x,t),\quad k\in\mathbb{Z}_{\ge 0}.
\label{eq:rogue-wave-recover}
\end{equation}
The rogue wave of order $k=0$ coincides with the background solution.  Indeed, if $k=0$, then the solution of Riemann-Hilbert Problem~\ref{rhp:rogue-wave} is 
\begin{equation}
\mathbf{M}^{(0)}(\lambda;x,t)=\begin{cases}
\mathbf{E}(\lambda),&\quad\text{$\lambda$ exterior to $\Sigma_\circ$}\\
\mathbf{E}(\lambda)\ee^{-\ii\rho(\lambda)(x+\lambda t)\sigma_3}\mathbf{E}(\lambda)^{-1}\ee^{\ii\rho(\lambda)(x+\lambda t)\sigma_3},
&\quad\text{$\lambda$ in the interior of $\Sigma_\circ$}.
\end{cases}
\label{eq:M-zero}
\end{equation}
In verifying the jump condition \eqref{eq:jump-cut} one should make use of the fact that the first three factors appearing on the second line of the right-hand side in \eqref{eq:M-zero} combine, perhaps despite appearances, to form an entire function $\mathbf{U}(\lambda;x,t)$ of $\lambda$:
\begin{equation}
\mathbf{U}(\lambda;x,t):=\mathbf{E}(\lambda)\ee^{-\ii\rho(\lambda)(x+\lambda t)\sigma_3}\mathbf{E}(\lambda)^{-1}=(x+\lambda t)\frac{\sin(\theta)}{\theta}
\begin{bmatrix}-\ii\lambda & 1\\-1 & \ii\lambda\end{bmatrix} + \cos(\theta)\mathbb{I},\quad\theta:=\rho(\lambda)(x+\lambda t),
\label{eq:entire}
\end{equation}
noting that analyticity follows because $\sin(\theta)/\theta$ and $\cos(\theta)$ are even in $\theta$ and hence entire functions of $\theta^2=(\lambda^2+1)(x+\lambda t)^2$.  Applying the formula \eqref{eq:rogue-wave-recover} for $k=0$ then gives
\begin{equation}
\psi_0(x,t)=2\ii\lim_{\lambda\to\infty}\lambda M^{(0)}_{12}(\lambda;x,t)=2\ii\lim_{\lambda\to\infty}\lambda E_{12}(\lambda) = 1.
\label{eq:background-recover}
\end{equation}
In \cite{BilmanM17}, the conditions of Riemann-Hilbert Problem~\ref{rhp:rogue-wave} were translated into a finite-dimensional linear algebra problem via a suitable rational ansatz for the matrix $\mathbf{M}^{(k)}(\lambda;x,t)\mathbf{E}(\lambda)^{-1}$ in the exterior domain that builds in poles of order $n$ at $\lambda=\pm\ii$ (only visible upon analytic continuation into the interior domain through $\Sigma_\circ$).  The coefficients in the partial-fraction expansion of this rational ansatz are determined so that the jump condition produces a matrix in the interior domain that is consistent with the required analyticity and continuity at $\lambda=\pm\ii$.  It turns out that the Taylor coefficients of the entire function \eqref{eq:entire} at $\lambda=\pm\ii$ appear when these conditions are implemented, and in fact we can recognize these coefficients in the quantities $F_\ell(x,t)$ and $G_\ell(x,t)$ defined by \eqref{eq:series-expansions}.  Thus it is possible to show the following.
\begin{proposition}
The function $\psi_k(x,t)$ obtained from the solution of Riemann-Hilbert Problem~\ref{rhp:rogue-wave} by \eqref{eq:rogue-wave-recover} coincides with the determinantal formula \eqref{eq:psi-k-determinants}.
\label{prop:Equivalence}
\end{proposition}
We give the proof in the Appendix.

\subsection{Qualitative properties of high-order fundamental rogue waves}
Using the determinantal formula \eqref{eq:psi-k-determinants}, it is easy to make plots that reveal certain qualitative features of fundamental rogue waves.  Figure~\ref{fig:surface-plots} shows surface plots of the modulus $|\psi_k(x,t)|$ over the $(x,t)$-plane for $k=1,2,3,4$.  
\begin{figure}[h]
\begin{center}
\includegraphics{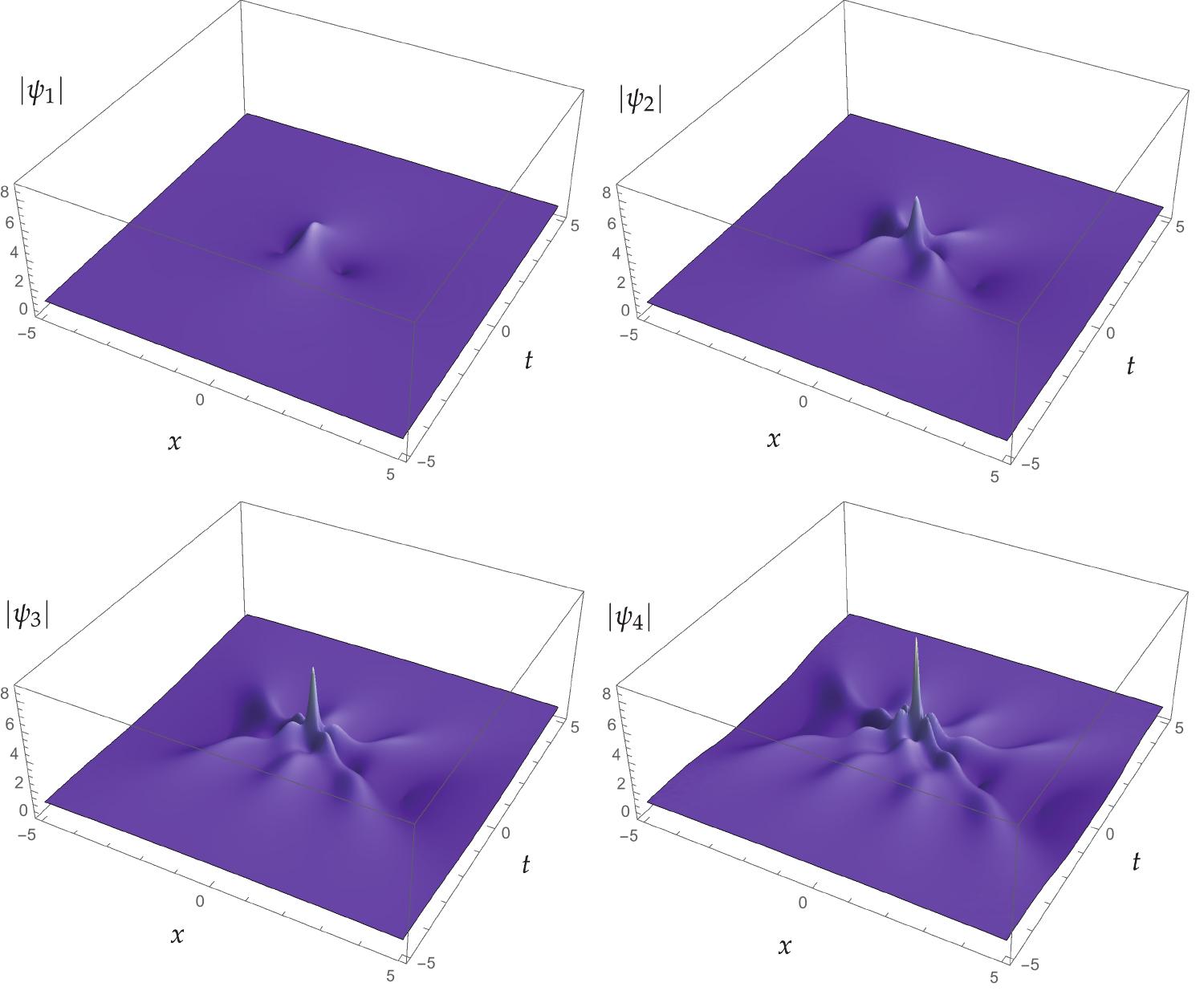}
\end{center}
\caption{The modulus $|\psi_k(x,t)|$ as a function of $(x,t)\in\mathbb{R}^2$ for $k=1,2,3,4$.}
\label{fig:surface-plots}
\end{figure}
These plots display the key characteristic that the amplitude of the fundamental rogue wave of order $k$ increases with $k$, and also shows that the extreme amplitude is achieved at a central peak that also concentrates as $k$ increases.  However, it is also clear that the solution becomes more complex as $k$ increases, with the formation of more and more subordinate peaks in amplitude.  One can also see that the rogue wave of order $k$ is not very symmetrical with respect to the roles of the coordinates $(x,t)$; indeed the amplitude seems to form a double ``shelf'' in the $t$-direction and a double ``channel'' in the $x$-direction.  

Features such as the space-time distribution of maxima on the shelves can more easily be seen in two-dimensional plots in which the amplitude is indicated with a grayscale.  Such plots are shown in Figure~\ref{fig:density-plots}.
\begin{figure}[h]
\begin{center}
\includegraphics{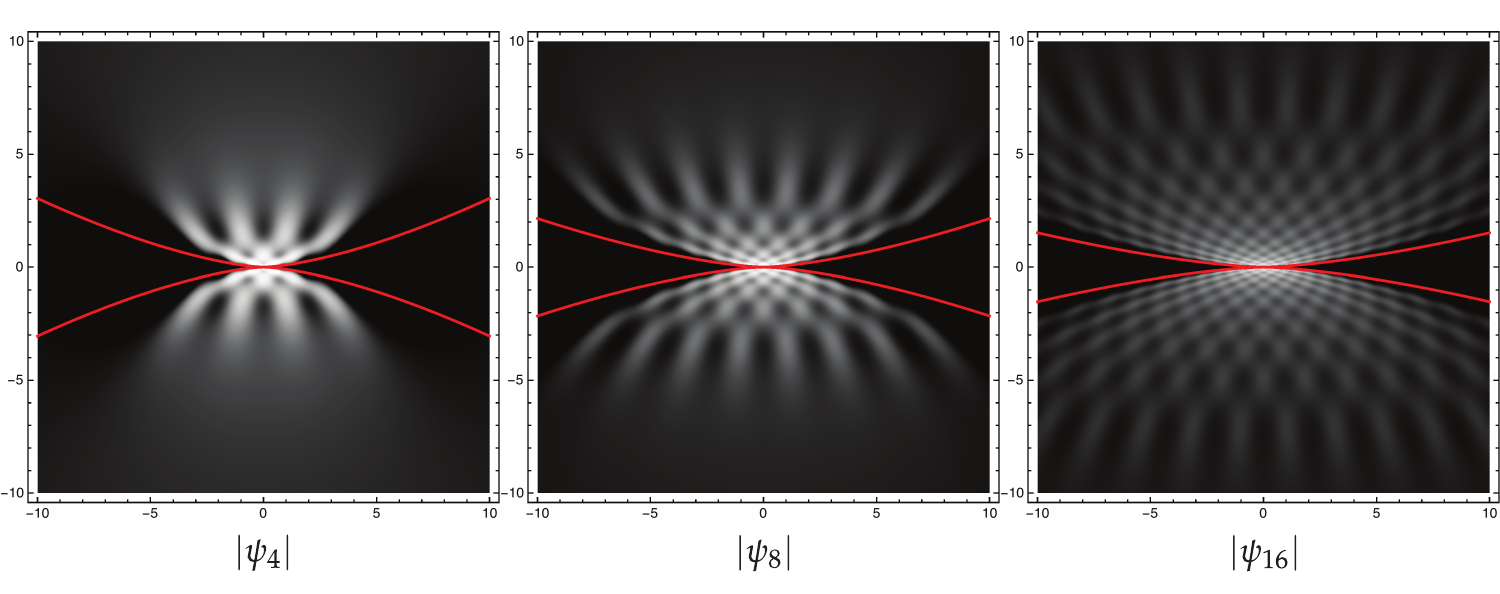}
\end{center}
\caption{$|\psi_k(x,t)|$ plotted over the $x$ (horizontal coordinate) and $t$ (vertical coordinate) plane for $k=4,8,16$ (or $n=2,4,8$).  Black means $|\psi_k|=0$ and lighter color corresponds to higher peaks of amplitude.  Superimposed in red are transitional curves $t=\pm|x|^{3/2}/\sqrt{54n}$ for the near-field asymptotics (cf., Section~\ref{sec:Painleve}).}
\label{fig:density-plots}
\end{figure}
These plots clearly show that the ``shelves'' in the amplitude $|\psi_k|$ that form before and after the amplitude peak at the origin have a boundary that apparently becomes more sharply-defined the larger the order $k$.  The shelves develop a regular crystalline pattern of local maxima, and meanwhile the ``channels'' near the $x$-axis become more clearly defined.  

The channels appear featureless in these plots by comparison with the shelves, but the rogue wave actually displays remarkable structure in these regions, as can be seen in one-dimensional plots of the restriction of the rogue wave to the $x$-axis.  Such plots are shown in Figures~\ref{fig:Odds-tzero} and \ref{fig:Evens-tzero}.
\begin{figure}[h]
\begin{center}
\includegraphics{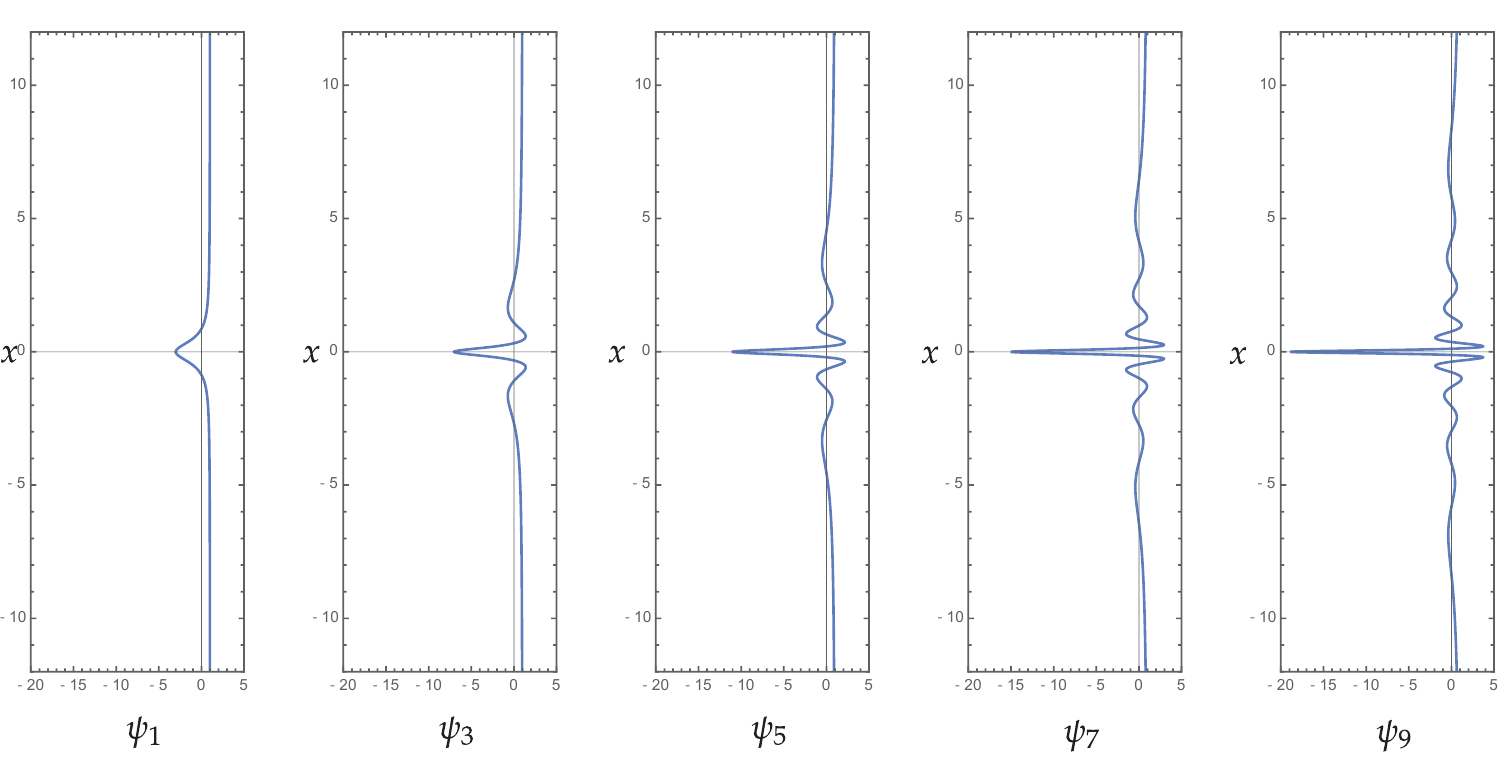}
\end{center}
\caption{Plots of $\psi_k(x,0)$ (real-valued) versus $x$ for $k=1,3,5,7,9$.}
\label{fig:Odds-tzero}
\end{figure}
\begin{figure}[h]
\begin{center}
\includegraphics{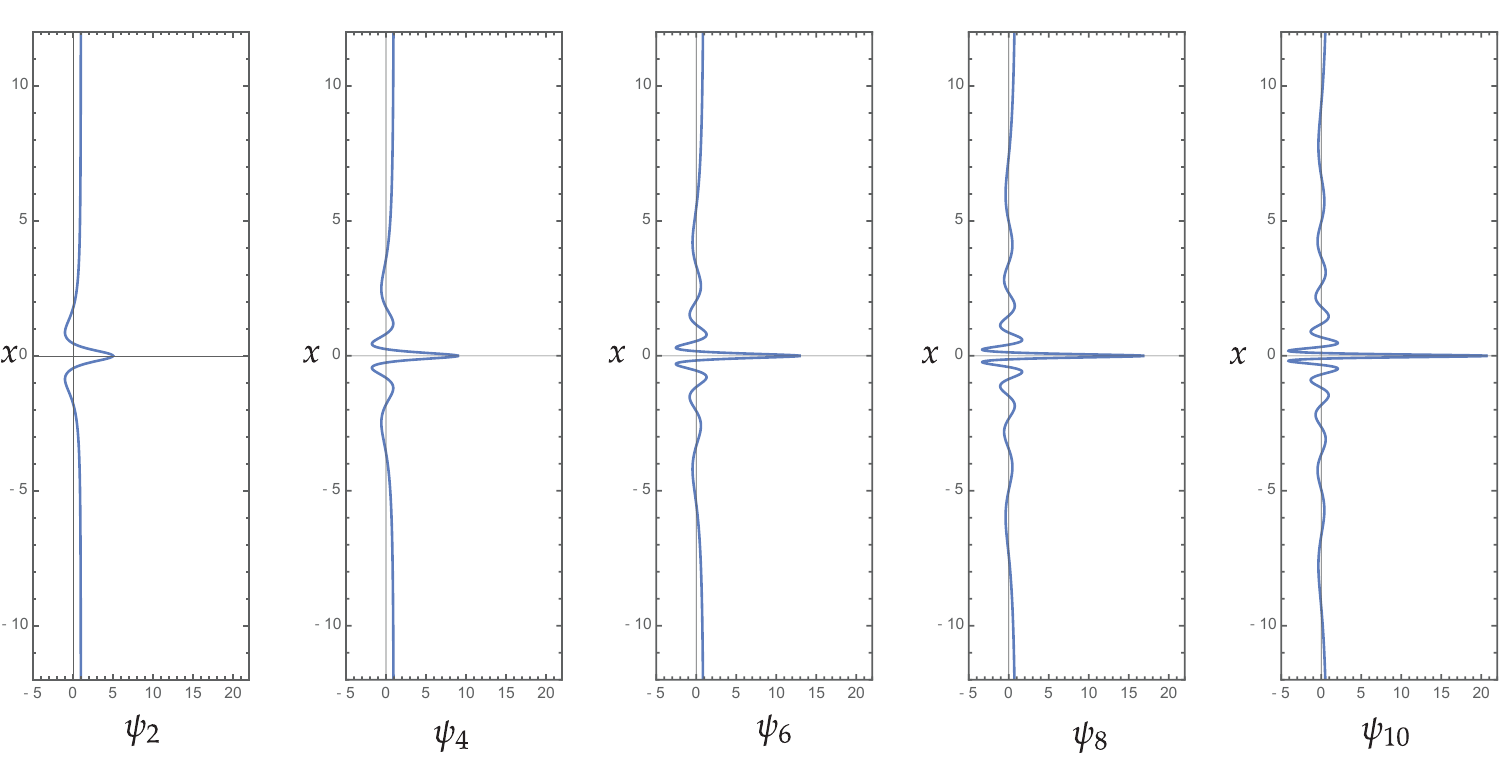}
\end{center}
\caption{Plots of $\psi_k(x,0)$ (real-valued) versus $x$ for $k=2,4,6,8,10$.}
\label{fig:Evens-tzero}
\end{figure}
These figures show that the rogue wave is highly oscillatory in the channels near the $x$-axis, with a number of zeros increasing with $k$.  In fact, there appear to be $2k$ zeros, and the largest zero appears to occur at approximately $x=\pm k$, beyond which the solution tends to the background value of $\psi=1$.  On the other hand, we will show in this paper that the zeros are by no means asymptotically equally spaced; the zeros near the origin in fact have spacing proportional to $k^{-1}$.
Similar plots of $\psi_k(x,t)$ restricted to the $t$-axis are shown in Figures~\ref{fig:Odds-xzero} and \ref{fig:Evens-xzero}.
\begin{figure}[h]
\begin{center}
\includegraphics{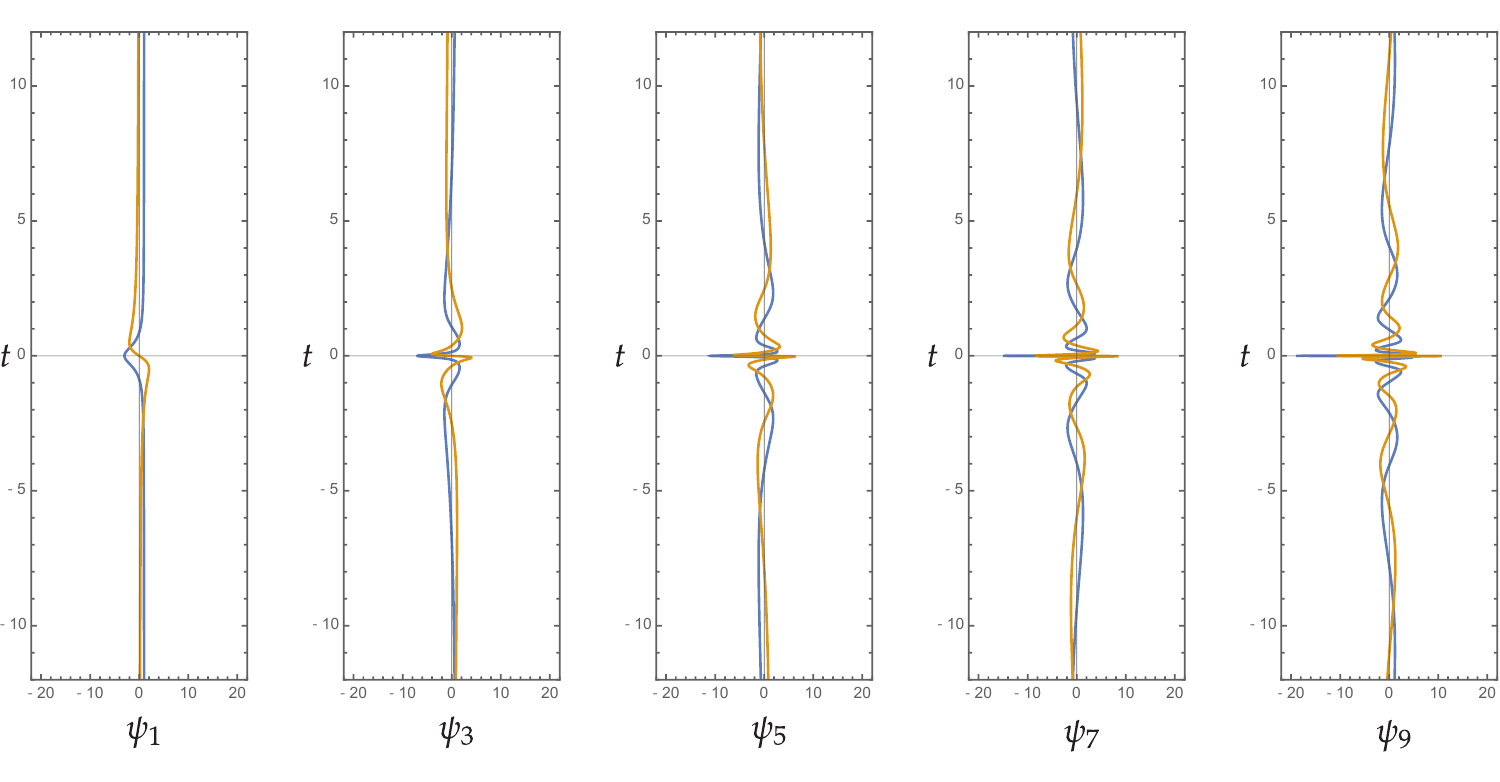}
\end{center}
\caption{Plots of $\mathrm{Re}(\psi_k(0,t))$ (blue) and $\mathrm{Im}(\psi_k(0,t))$ (maize) versus $t$ for $k=1,3,5,7,9$.}
\label{fig:Odds-xzero}
\end{figure}
\begin{figure}[h]
\begin{center}
\includegraphics{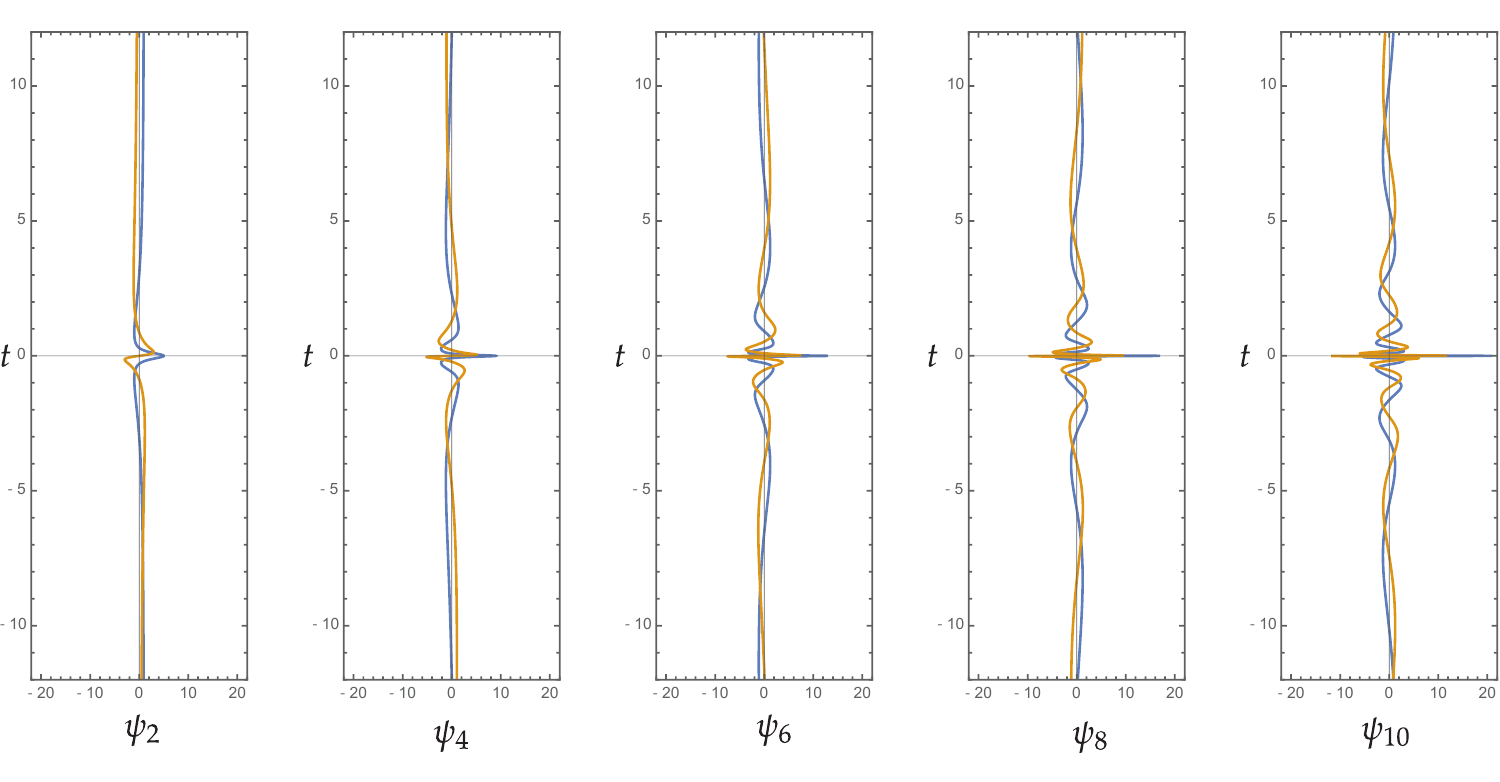}
\end{center}
\caption{Plots of $\mathrm{Re}(\psi_k(0,t))$ (blue) and $\mathrm{Im}(\psi_k(0,t))$ (maize) versus $t$ for $k=2,4,6,8,10$.}
\label{fig:Evens-xzero}
\end{figure}
These figures show that the rogue waves are also highly oscillatory in the $t$-direction when $k$ is large, and one can clearly observe that the frequency of the oscillations is greater near the origin than in the plots shown in Figures~\ref{fig:Odds-tzero} and \ref{fig:Evens-tzero}.  We will show in this paper that the time frequency of the rogue wave near $t=0$ scales like $k^2$.

The fundamental rogue wave of order $k$ clearly displays remarkable complexity when $k$ is large, and yet it also clearly demonstrates many of the hallmark features of a multiscale structure.  Such features are very difficult to extract from the determinantal formula \eqref{eq:psi-k-determinants} because the natural limit $k\to\infty$ involves computing determinants of larger and larger dimension.  On the other hand, the representation of $\psi_k(x,t)$ via Riemann-Hilbert Problem~\ref{rhp:rogue-wave} turns out to be a more fruitful avenue for large-$k$ asymptotic analysis of the fundamental rogue wave of order $k$.  In this paper, we take the first steps in such analysis by giving an asymptotic description of $\psi_k(x,t)$ in the \emph{near-field limit}, i.e., for $(x,t)$ in a small neighborhood (shrinking in size as $k\to\infty$) of the origin $(0,0)$.  This analysis reveals something nontrivial, namely a particular pair of opposite transcendental solutions of the focusing nonlinear Schr\"odinger equation that we call the \emph{rogue waves of infinite order}.  This paper is devoted to the proof of this result and the detailed description of these special limiting solutions.

\subsection{Removing the branch cut}
An equivalent Riemann-Hilbert problem is easily formulated in which the unknown has no jump across $\Sigma_\mathrm{c}$, the branch cut for $\rho$ and $f$.  To this end, we use the matrix $\mathbf{M}^{(0)}(\lambda;x,t)$ as a parametrix for $\mathbf{M}^{(k)}(\lambda;x,t)$ and hence consider the matrix
\begin{equation}
\mathbf{N}^{(k)}(\lambda;x,t):=\mathbf{M}^{(k)}(\lambda;x,t)\mathbf{M}^{(0)}(\lambda;x,t)^{-1}.
\label{eq:N-M-relation}
\end{equation}
It is easy to check that since the jump condition \eqref{eq:jump-cut} is independent of $k$, $\mathbf{N}^{(k)}_+(\lambda;x,t)=\mathbf{N}^{(k)}_-(\lambda;x,t)$ for all $\lambda\in\Sigma_\mathrm{c}$.  Since the boundary values taken on $\Sigma_\mathrm{c}$ are continuous, a Morera argument shows that $\mathbf{N}^{(k)}(\lambda;x,t)$ can be defined on $\Sigma_\mathrm{c}$ in such a way that $\mathbf{N}^{(k)}(\lambda;x,t)$ becomes analytic for $\lambda\in\mathbb{C}\setminus\Sigma_\circ$.  Similarly, since $\mathbf{M}^{(k)}(\lambda;x,t)\to\mathbb{I}$ as $\lambda\to\infty$ independent of $k$, it follows that $\mathbf{N}^{(k)}(\lambda;x,t)\to\mathbb{I}$ as $\lambda\to\infty$ also.  It only remains to compute the jump condition satisfied by $\mathbf{N}^{(k)}(\lambda;x,t)$ across $\Sigma_\circ$ to formulate the following equivalent problem.
\begin{rhp}[Rogue wave of order $k$ --- Reformulation]
Let $(x,t)\in\mathbb{R}^2$ be arbitrary parameters, and let $k\in\mathbb{Z}_{\ge 0}$.  Find a $2\times 2$ matrix $\mathbf{N}^{(k)}(\lambda;x,t)$ with the following properties:
\begin{itemize}
\item[]\textbf{Analyticity:}  $\mathbf{N}^{(k)}(\lambda;x,t)$ is analytic in $\lambda$ for $\lambda\in\mathbb{C}\setminus\Sigma_\circ$, and it takes continuous boundary values on $\Sigma_\circ$ from the interior and exterior.
\item[]\textbf{Jump condition:}  The boundary values on $\Sigma_\circ$ (recall clockwise orientation) are related as follows.  If $k=2n$, $n\in\mathbb{Z}_{\ge 0}$,
\begin{equation}
\mathbf{N}^{(k)}_+(\lambda;x,t)=\mathbf{N}^{(k)}_-(\lambda;x,t)\mathbf{U}(\lambda;x,t)\mathbf{Q}\left(\frac{\lambda-\ii}{\lambda+\ii}\right)^{n\sigma_3}\mathbf{Q}^{-1}\mathbf{U}(\lambda;x,t)^{-1},\quad\lambda\in\Sigma_\circ,
\label{eq:N-jump-even}
\end{equation}
while if instead $k=2n-1$, $n\in\mathbb{Z}_{>0}$,
\begin{equation}
\mathbf{N}^{(k)}_+(\lambda;x,t)=\mathbf{N}^{(k)}_-(\lambda;x,t)\mathbf{U}(\lambda;x,t)\mathbf{Q}\left(\frac{\lambda+\ii}{\lambda-\ii}\right)^{n\sigma_3}\mathbf{Q}^{-1}\mathbf{U}(\lambda;x,t)^{-1},\quad\lambda\in\Sigma_\circ,
\label{eq:N-jump-odd}
\end{equation}
where the entire unit-determinant matrix $\mathbf{U}(\lambda;x,t)$ is defined in \eqref{eq:entire}.
\item[]\textbf{Normalization:}  $\mathbf{N}^{(k)}(\lambda;x,t)\to\mathbb{I}$ as $\lambda\to\infty$.
\end{itemize}
\label{rhp:renormalized}
\end{rhp}
Clearly, if $k=0$, then the jump condition on $\Sigma_\circ$ simply reads $\mathbf{N}^{(0)}_+(\lambda;x,t)=\mathbf{N}^{(0)}_-(\lambda;x,t)$ so the solution of the problem is simply $\mathbf{N}^{(0)}(\lambda;x,t)\equiv\mathbb{I}$.  Using \eqref{eq:rogue-wave-recover} and \eqref{eq:background-recover} shows that
\begin{equation}
\psi_k(x,t)=1+2\ii\lim_{\lambda\to\infty}\lambda N^{(k)}_{12}(\lambda;x,t),\quad k\in\mathbb{Z}_{\ge 0}.
\label{eq:rogue-wave-recover-2}
\end{equation}
This formulation immediately gives a new and very simple proof of a recent result \cite{WangYWH17} characterizing the maximum amplitude of the rogue wave of order $k$, which turns out to be achieved at the origin $(x,t)=(0,0)$.
\begin{proposition}
$\psi_k(0,0)=(-1)^k(2k+1)$.
\label{prop:origin-value}
\end{proposition}
\begin{proof}
Set $(x,t)=(0,0)$ in Riemann-Hilbert Problem~\ref{rhp:renormalized}.  Since $\mathbf{U}(\lambda;0,0)=\mathbb{I}$, the jump condition then becomes simply
\begin{equation}
\mathbf{N}^{(k)}_+(\lambda;0,0)=\mathbf{N}^{(k)}_-(\lambda;0,0)\mathbf{Q}\left(\frac{\lambda-\ii}{\lambda+\ii}\right)^{n\sigma_3}\mathbf{Q}^{-1},\quad \lambda\in\Sigma_\circ,\quad k=2n,
\end{equation}
or
\begin{equation}
\mathbf{N}^{(k)}_+(\lambda;0,0)=\mathbf{N}^{(k)}_-(\lambda;0,0)\mathbf{Q}\left(\frac{\lambda+\ii}{\lambda-\ii}\right)^{n\sigma_3}\mathbf{Q}^{-1},\quad \lambda\in\Sigma_\circ,\quad k=2n-1,
\end{equation}
depending on whether $k$ is even or odd.  Either way, it is clear that the jump is diagonalized by a constant conjugation, which also preserves the normalization at $\lambda=\infty$:  $\mathbf{N}^{(k)}(\lambda;0,0) = \mathbf{Q}\mathbf{D}^{(k)}(\lambda)\mathbf{Q}^{-1}$.  Then one solves the resulting diagonal problem for $\mathbf{D}^{(k)}(\lambda)$ explicitly by setting $\mathbf{D}^{(k)}(\lambda)\equiv\mathbb{I}$ in the interior of $\Sigma_\circ$ and 
\begin{equation}
\mathbf{D}^{(k)}(\lambda)=\left(\frac{\lambda-\ii}{\lambda+\ii}\right)^{n\sigma_3},\quad \text{$\lambda$ exterior to $\Sigma_\circ$},\quad k=2n,
\end{equation}
or
\begin{equation}
\mathbf{D}^{(k)}(\lambda)=\left(\frac{\lambda+\ii}{\lambda-\ii}\right)^{n\sigma_3},\quad \text{$\lambda$ exterior to $\Sigma_\circ$},\quad k=2n-1.
\end{equation}
Since in the limit $\lambda\to\infty$, 
\begin{equation}
\mathbf{D}^{(k)}(\lambda)=\begin{cases}
\mathbb{I}-2\ii n\sigma_3\lambda^{-1} + O(\lambda^{-2}),&\quad k=2n\\
\mathbb{I}+2\ii n\sigma_3\lambda^{-1}+O(\lambda^{-2}),&\quad k=2n-1,
\end{cases}
\end{equation}
conjugating by $\mathbf{Q}$ and using $\mathbf{Q}\sigma_3\mathbf{Q}^{-1}=\sigma_1$ gives
\begin{equation}
\mathbf{N}^{(k)}(\lambda;0,0)=\begin{cases}\mathbb{I}-2\ii n\sigma_1\lambda^{-1}+O(\lambda^{-2}),&\quad k=2n\\
\mathbb{I}+2\ii n\sigma_1\lambda^{-1}+O(\lambda^{-2}),&\quad k=2n-1
\end{cases}
\end{equation}
as $\lambda\to\infty$.
Applying the formula \eqref{eq:rogue-wave-recover-2} finishes the proof.
\end{proof}
It is also true that $\psi_k(x,t)\to 1$ as $x^2+t^2\to\infty$, although the shortest proof of this that we know so far comes from the algebraic representation \eqref{eq:psi-k-determinants} and is not very enlightening in the present context. 

\subsection{Summary of results}
The main result of our paper is Theorem~\ref{theorem:main}, which is formulated and proved in Section~\ref{sec:near-field} with the help of the Riemann-Hilbert representation of $\psi_k(x,t)$.  This result asserts that, when examined on spatial scales $x=O(k^{-1})$ and temporal scales $t=O(k^{-2})$, a suitable rescaling of $\psi_k(x,t)$ actually has a nontrivial limit as $k\to\infty$ along subsequences of even and odd $k$.  The two ``near-field'' limits are functions $\Psi^\pm(X,T)$ of rescaled space and time variables that are well-defined transcendental solutions of the focusing nonlinear Schr\"odinger equation in the rescaled variables.  They are \emph{rogue waves of infinite order}, and they have a natural Riemann-Hilbert characterization (cf., Riemann-Hilbert Problem~\ref{rhp:limit}).  Heuristically, the near-field limit is capturing the central peak of the rogue wave $\psi_k(x,t)$ and an arbitrary finite number of neighboring peaks; all of this interesting behavior is occurring just within the bright spot near the origin in the plots in Figure~\ref{fig:density-plots}!

In Section~\ref{sec:Exact-Properties}, we establish several important exact properties of the functions $\Psi^\pm(X,T)$.  First, in Section~\ref{sec:symmetries} we show that $\Psi^-(X,T)=-\Psi^+(X,T)$ (Corollary~\ref{cor:Psi-pm}), that $\Psi^\pm(-X,T)=\Psi^\pm(X,T)$ (Corollary~\ref{cor:Psi-even}), that $\Psi^\pm(X,-T)=\Psi^\pm(X,T)^*$ (Corollary~\ref{cor:Psi-real}), and that $\Psi^\pm(0,0)=\pm 4$ (Proposition~\ref{prop:Psi-peak}).  Then, in Section~\ref{sec:ODE} we show that not only do the functions $\Psi^\pm(X,T)$ satisfy the focusing nonlinear Schr\"odinger equation, but they also satisfy simple ordinary differential equations with respect to $X$ for fixed $T$ (Theorem~\ref{theorem:X-ODEs}) and with respect to $T$ for fixed $X$ (Theorem~\ref{theorem:Psi-ODEs-T}).  We identify the differential equations with respect to $X$ as belonging to the Painlev\'e-III hierarchy in the sense of Sakka \cite{Sakka09}.  In particular, when $T=0$, the latter reduces to a special case of the classical Painlev\'e-III equation in which the formal monodromy parameters both vanish:  $\Theta_0=\Theta_\infty=0$; see Corollary~\ref{corollary:PIII}.

Then, in Section~\ref{sec:Psi-asymptotic}, we specify the rogue waves of infinite order $\Psi^\pm(X,T)$ more precisely by determining their asymptotic behavior as $X,T\to\infty$.  Such asymptotic formul\ae\ would perhaps describe the rogue wave of order $k$ when $k$ is large in a certain overlap domain\footnote{See Conjecture~\ref{conjecture:numerical} in Section~\ref{sec:numerics} which concerns such overlap domains.} where the near-field asymptotic of Theorem~\ref{theorem:main} gives way to a far-field description that is the subject of ongoing research \cite{BilmanLMT18}.  It turns out that the large $(X,T)$ behavior of $\Psi^\pm(X,T)$ depends on whether $(X,T)$ tends to infinity primarily in the $T$-direction (thus matching onto the ``shelves'' visible in Figures~\ref{fig:surface-plots} and \ref{fig:density-plots}) or primarily in the $X$-direction (matching onto the ``channels'').  The large-$X$ asymptotic regime is described in Theorem~\ref{theorem:large-X} which is formulated and proved in Section~\ref{sec:large-X}.  The large-$T$ asymptotic regime is described in Theorem~\ref{theorem:large-T} which is formulated and proved in Section~\ref{sec:large-T}.  The latter results become even more explicit if $T=0$ (Corollary~\ref{corollary:large-X-T0}) or $X=0$ (Corollary~\ref{corollary:large-T-X0}) respectively.  The two regimes meet along curves $T=\pm 54^{-1/2}|X|^{3/2}$, and in a neighborhood of these curves neither asymptotic result is valid.  In Section~\ref{sec:Painleve} we therefore consider the asymptotic regime of large $(X,T)$ with $T\approx \pm 54^{-1/2}|X|^{3/2}$ and we formulate and prove Theorem~\ref{theorem:PII} where we show that the transitional asymptotics are described by a certain \emph{tritronqu\'ee} solution of the Painlev\'e-II equation.  All of the results in Section~\ref{sec:Psi-asymptotic} are obtained by applying elements of the Deift-Zhou steepest descent method \cite{DeiftZ93} to Riemann-Hilbert Problem~\ref{rhp:limit-simpler}, which is equivalent to Riemann-Hilbert Problem~\ref{rhp:limit} and characterizes uniquely the rogue waves of infinite order.  These results lead us to regard $\Psi^\pm(X,T)$ as new special functions.

In Section~\ref{sec:numerics}, we apply numerical methods for Riemann-Hilbert problems to reliably compute these new special functions.  We first produce accurate plots of rogue waves of infinite order.  We then compare these solutions with finite-order rogue waves and also with large-$X$ asymptotic formul\ae\ for $\Psi^\pm(X,T)$ obtained in Section~\ref{sec:Psi-asymptotic}.  We also use numerics to formulate a conjecture generalizing our main convergence result, asserting its validity on larger sets than predicted by Theorem~\ref{theorem:main}.

In an appendix, we give a proof of Proposition~\ref{prop:Equivalence}.

\subsection*{Acknowledgements}  The work of D. Bilman was supported by a travel grant from the Simons Foundation.  The work of L. Ling was supported by the National Natural Science Foundation of China (Contact Nos. 11771151, 11401221), Guangdong Natural Science Foundation (Contact No.\@ 2017A030313008), China Scholarship Council under Grant 201706155005, Guangzhou Science and Technology Program (No.\@ 201707010040).  The work of P. D. Miller was supported by the National Science Foundation under grant DMS-1513054.

\section{Near-field asymptotic behavior of fundamental rogue waves}
\label{sec:near-field}
Writing $k=2n$ for $k$ even and $k=2n-1$ for $k$ odd, consider the following substitutions in Riemann-Hilbert Problem~\ref{rhp:renormalized}:
\begin{equation}
x=\frac{X}{n},\quad t=\frac{T}{n^2},\quad \lambda=n\Lambda.
\end{equation}
We choose the contour $\Sigma_\circ$ to be the circle of radius $n$.  Observe the following asymptotic behavior of the jump matrix:
\begin{multline}
\left.\mathbf{E}(\lambda)\ee^{-\ii\rho(\lambda)(x+\lambda t)\sigma_3}\mathbf{E}(\lambda)^{-1}\mathbf{Q}\left(\frac{\lambda-\ii}{\lambda+\ii}\right)^{\pm n\sigma_3}\mathbf{Q}^{-1}\mathbf{E}(\lambda)\ee^{\ii\rho(\lambda)(x+\lambda t)\sigma_3}\mathbf{E}(\lambda)^{-1}\right|_{\lambda=n\Lambda,x=n^{-1}X,t=n^{-2}T}\\
= (\mathbb{I}+O(n^{-1}))\ee^{-\ii(\Lambda X+\Lambda^2T)\sigma_3}\mathbf{Q}\ee^{\mp 2\ii \Lambda^{-1}\sigma_3}\mathbf{Q}^{-1}\ee^{\ii (\Lambda X+\Lambda^2T)\sigma_3}
\label{eq:jump-approximation}
\end{multline}
which holds uniformly for $|\Lambda|=1$ and $(X,T)$ in compact subsets of $\mathbb{R}^2$.  Considering $k$ and hence $n$ large, and neglecting the error term results in the following model Riemann-Hilbert problem.
\begin{rhp}[Rogue waves of infinite order]
Let $(X,T)\in\mathbb{R}^2$ be fixed.  Find a $2\times 2$ matrix $\mathbf{P}^\pm(\Lambda;X,T)$ with the following properties:
\begin{itemize}
\item[]\textbf{Analyticity:}  $\mathbf{P}^\pm(\Lambda;X,T)$ is analytic in $\Lambda$ for $|\Lambda|\neq 1$, and it takes continuous boundary values on the unit circle from the interior and exterior.
\item[]\textbf{Jump condition:}  Assuming clockwise orientation of the unit circle $|\Lambda|=1$, the boundary values are connected by the following formula:
\begin{equation}
\mathbf{P}^\pm_+(\Lambda;X,T)=\mathbf{P}^\pm_-(\Lambda;X,T)\ee^{-\ii (\Lambda X+\Lambda^2T)\sigma_3}\mathbf{Q}\ee^{\mp 2\ii \Lambda^{-1}\sigma_3}\mathbf{Q}^{-1}\ee^{\ii (\Lambda X+\Lambda^2 T)\sigma_3},\quad |\Lambda|=1.
\label{eq:P-jump}
\end{equation}
\item[]\textbf{Normalization:}  $\mathbf{P}^\pm(\Lambda;X,T)\to\mathbb{I}$ as $\Lambda\to\infty$.
\end{itemize}
\label{rhp:limit}
\end{rhp}
The matrix $\mathbf{P}^+(\Lambda;X,T)$ will correspond to the large-$k$ asymptotics of rogue waves of even order $k=2n$, while $\mathbf{P}^-(\Lambda;X,T)$ will correspond to the large-$k$ asymptotics of rogue waves of odd order $k=2n-1$.  In fact, these two matrices are explicitly related, as we will show below.  The basic properties of Riemann-Hilbert Problem~\ref{rhp:limit} are summarized in the following proposition.
\begin{proposition}
Riemann-Hilbert Problem~\ref{rhp:limit} has a unique solution for each choice of sign $\pm$ and for each $(X,T)\in\mathbb{R}^2$.  The solution satisfies $\det(\mathbf{P}^\pm(\Lambda;X,T))=1$, and
for every compact subset $K\subset\mathbb{R}^2$,
\begin{equation}
\sup_{|\Lambda|\neq 1,(X,T)\in K}\|\mathbf{P}^\pm(\Lambda;X,T)\| = C_K<\infty.
\label{eq:P-bound}
\end{equation}
The function $\Psi(X,T)=\Psi^\pm(X,T)$ defined from $\mathbf{P}^\pm(\Lambda;X,T)$ by the limit
\begin{equation}
\Psi^\pm(X,T):=2\ii\lim_{\Lambda\to\infty} \Lambda P^\pm_{12}(\Lambda;X,T)
\label{eq:Psi-pm-define}
\end{equation}
is a global solution of the focusing nonlinear Schr\"odinger equation in the form
\begin{equation}
\ii\frac{\partial\Psi}{\partial T} + \frac{1}{2}\frac{\partial^2\Psi}{\partial X^2} + |\Psi|^2\Psi=0.
\label{eq:Psi-NLS}
\end{equation}
\label{prop:basic}
\end{proposition}
\begin{proof}
To prove unique solvability, we will show that the jump conditions and the jump matrices in Riemann-Hilbert Problem~\ref{rhp:limit} satisfy the hypotheses of Zhou's \emph{Vanishing Lemma} \cite[Theorem 9.3]{Zhou89}. To this end, we reorient the jump contour $|\Lambda|=1$ to have clockwise orientation in the upper half plane and counter-clockwise orientation in the lower half plane.  This makes the reoriented jump contour invariant, including orientation, under Schwarz reflection symmetry in the real axis.  Reversing the orientation on the lower semicircle means exchanging the boundary values or equivalently replacing the jump matrix there with its inverse; hence the jump matrix in \eqref{eq:P-jump} when defined on the reoriented jump contour becomes:
\begin{equation}
\mathbf{V}^{\pm}(\Lambda;X,T):= \begin{cases}\ee^{-\ii(\Lambda X + \Lambda^2 T)\sigma_3}\mathbf{Q}\ee^{\mp 2\ii\Lambda^{-1}\sigma_3}\mathbf{Q}^{-1} \ee^{\ii(\Lambda X + \Lambda^2 T)\sigma_3},&|\Lambda|=1,\quad\mathrm{Im}(\Lambda)>0\\
\ee^{-\ii(\Lambda X + \Lambda^2 T)\sigma_3}\mathbf{Q}\ee^{\pm 2\ii\Lambda^{-1}\sigma_3}\mathbf{Q}^{-1} \ee^{\ii(\Lambda X + \Lambda^2 T)\sigma_3},&|\Lambda|=1,\quad\mathrm{Im}(\Lambda)<0.
 \end{cases}
\end{equation}
For $|\Lambda|=1$ with $\mathrm{Im}(\Lambda)>0$, using the fact that $\mathbf{Q}$ is a real orthogonal matrix, we have
\begin{equation}
\begin{aligned}
\mathbf{V}^{\pm}(\Lambda^*;X,T) &=\left[ \ee^{\ii(\Lambda X +\Lambda^2 T)\sigma_3} \mathbf{Q} \ee^{\mp 2\ii \Lambda^{-1}} \mathbf{Q}^{\top} \ee^{-\ii(\Lambda X + \Lambda^2 T)\sigma_3}\right]^*\\
&=\left[ \ee^{-\ii(\Lambda X +\Lambda^2 T)\sigma_3} \mathbf{Q} \ee^{\mp 2\ii \Lambda^{-1}} \mathbf{Q}^{\top} \ee^{\ii(\Lambda X + \Lambda^2 T)\sigma_3}\right]^\dagger\\
&= \mathbf{V}^{\pm}(\Lambda;X,T)^\dagger,
\end{aligned}
\label{eq:limit-jump-Schwarz}
\end{equation}
where the superscript ``$^\dagger$" denotes the conjugate transpose of the matrix. Thus, whenever $(X,T)\in\mathbb{R}^2$, the identity $\mathbf{V}^{\pm}(\Lambda^*;X,T)=\mathbf{V}^{\pm}(\Lambda;X,T)^\dagger$ holds on the reoriented Schwarz-symmetric jump contour $|\Lambda|=1$.
Taking into account the normalization condition $\mathbf{P}^{\pm}(\Lambda;X,T)\to \mathbb{I}$ as $\Lambda\to\infty$, we have confirmed all the hypotheses of the vanishing lemma. Consequently, Riemann-Hilbert Problem~\ref{rhp:limit} is uniquely solvable for all $(X,T)\in\mathbb{R}^2$.  

Because as a polynomial in analytic matrix entries $\det(\mathbf{P}^\pm(\Lambda;X,T))$ is analytic for $|\Lambda|\neq 1$, and since the jump matrix is unimodular, Morera's Theorem shows that $\det(\mathbf{P}^\pm(\Lambda;X,T))$ can be extended to $|\Lambda|=1$ as an entire function.  Applying the normalization condition and invoking Liouville's Theorem then shows that $\det(\mathbf{P}^{\pm}(\Lambda;X,T))\equiv 1$ holds for $|\Lambda|\neq 1$ and for all $(X,T)\in\mathbb{R}^2$.

Moreover, since the jump contour is compact and the jump matrix depends analytically on $X$ and $T$, it follows from analytic Fredholm theory applied to the system of singular integral equations equivalent to Riemann-Hilbert Problem~\ref{rhp:limit} that the solution $\mathbf{P}^\pm(\Lambda;X,T)$ is real-analytic in $(X,T)$; in particular it is continuous and hence bounded on compact sets $K$ in the $(X,T)$-plane.  This fact, together with the continuous manner in which the boundary values of $\mathbf{P}^\pm$ are achieved on the unit circle in the $\Lambda$-plane (actually, the boundary values can easily be seen to extend analytically through the jump contour from both directions) proves the estimate \eqref{eq:P-bound}.  Being analytic in $\Lambda$ outside of the unit circle, the matrix $\mathbf{P}^\pm(\Lambda;X,T)$ admits a convergent Laurent expansion of the form
\begin{equation}
\mathbf{P}^\pm(\Lambda;X,T) = \mathbb{I}+\sum_{j=1}^\infty\mathbf{P}^{\pm[j]}(X,T)\Lambda^{-j},\quad |\Lambda|>1,
\label{eq:P-Laurent}
\end{equation}
and analytic Fredholm theory implies that each coefficient $\mathbf{P}^{\pm[j]}(X,T)$ is real-analytic on $\mathbb{R}^2$ and that the series \eqref{eq:P-Laurent} is differentiable term-by-term with respect to $X$ and/or $T$.  In particular, the function $\Psi^\pm(X,T)$ obtained from $\mathbf{P}^\pm(\Lambda;X,T)$ via the limit \eqref{eq:Psi-pm-define} is simply $\Psi^\pm(X,T)=2\ii P^{\pm[1]}_{12}(X,T)$, which is a real-analytic function on $\mathbb{R}^2$.

We will now use a ``dressing'' argument to show that $\Psi^{\pm}(X,T)=2\ii P^{\pm[1]}_{12}(X,T)$ is a solution of the focusing nonlinear Schr\"odinger equation in the form \eqref{eq:Psi-NLS}. To this end, we define 
\begin{equation}
\mathbf{W}^{\pm}(\Lambda;X,T):= \mathbf{P}^{\pm}(\Lambda;X,T)\ee^{-\ii(\Lambda X+\Lambda^2 T)\sigma_3},
\end{equation}
and observe that $\mathbf{W}^{\pm}(\Lambda;X,T)$ is analytic for $|\Lambda|\neq 1$, satisfying a jump condition across the unit circle with jump matrix $\mathbf{Q}\ee^{\mp 2\ii\Lambda^{-1}\sigma_3}\mathbf{Q}^{-1}$ (assuming clockwise orientation) that is independent of $(X,T)\in\mathbb{R}^2$. The  partial derivatives $\mathbf{W}^{\pm}_X(\Lambda;X,T)$ and $\mathbf{W}^{\pm}_T(\Lambda;X,T)$ are both analytic in the same domain and, by differentiation of the jump condition for $\mathbf{W}^\pm$ with respect to $X$ and $T$, they satisfy the same jump condition as $\mathbf{W}^{\pm}(\Lambda;X,T)$ does. It then follows that the matrices
\begin{equation}
\mathbf{A}^{\pm}(\Lambda;X,T):= \mathbf{W}^{\pm}_X(\Lambda;X,T)\mathbf{W}^{\pm}(\Lambda;X,T)^{-1}\quad\text{and}\quad\mathbf{B}^{\pm}(\Lambda;X,T):=\mathbf{W}^{\pm}_T(\Lambda;X,T)\mathbf{W}^{\pm}(\Lambda;X,T)^{-1}
\label{eq:def-A-B}
\end{equation}
can be defined by continuity for $|\Lambda|=1$ so that they become entire functions of $\Lambda$. 
Since the series \eqref{eq:P-Laurent} is differentiable term-by-term with respect to $X$ and $T$, we obtain
\begin{equation}
\begin{aligned}
\mathbf{A}^{\pm}(\Lambda;X,T) &= -\ii\Lambda\sigma_3 + \ii[\sigma_3,\mathbf{P}^{\pm[1]}(X,T)] + O(\Lambda^{-1}),\quad \Lambda\to\infty \\
&=-\ii\Lambda\sigma_3 + \ii[\sigma_3,\mathbf{P}^{\pm[1]}(X,T)]
\end{aligned}
\label{eq:dress-X}
\end{equation}
and
\begin{equation}
\begin{aligned}
\mathbf{B}^{\pm}(\Lambda;X,T) &= -\ii \Lambda^2 \sigma_3 + \ii\Lambda[\sigma_3, \mathbf{P}^{\pm[1]}(X,T)]+\ii[\mathbf{P}^{\pm[1]}(X,T), \sigma_3\mathbf{P}^{\pm[1]}(X,T)] + \ii[\sigma_3,\mathbf{P}^{\pm[2]}(X,T)] \\
&\quad\quad\quad\quad{}+ O(\Lambda^{-1}), \quad \Lambda\to\infty\\
&=-\ii \Lambda^2 \sigma_3 + \ii\Lambda[\sigma_3, \mathbf{P}^{\pm[1]}(X,T)]+\ii[\mathbf{P}^{\pm[1]}(X,T), \sigma_3\mathbf{P}^{\pm[1]}(X,T)] + \ii[\sigma_3,\mathbf{P}^{\pm[2]}(X,T)] 
\end{aligned}
\end{equation}
where the last equality in each case is a consequence of Liouville's Theorem. The dependence on the matrix $\mathbf{P}^{\pm[2]}(X,T)$ can be removed because the coefficient of $\Lambda^{-1}$ in the $O(\Lambda^{-1})$ error term in \eqref{eq:dress-X} is
\begin{equation}
\ii[\sigma_3, \mathbf{P}^{\pm[2]}(X,T)] + \ii[\mathbf{P}^{\pm[1]}(X,T),\sigma_3 \mathbf{P}^{\pm[1]}(X,T)] + \mathbf{P}^{\pm[1]}_{X}(X,T)
\label{eq:Lax-X-error}
\end{equation}
which must vanish again by Liouville's Theorem. Therefore, setting to zero the off-diagonal terms in \eqref{eq:Lax-X-error} allows $\mathbf{B}^{\pm}(\Lambda;X,T)$ to be expressed as the following quadratic polynomial in $\Lambda$:
\begin{equation}
\mathbf{B}^{\pm}(\Lambda;X,T) = -\ii \Lambda^2 \sigma_3 + \ii\Lambda[\sigma_3, \mathbf{P}^{\pm[1]}(X,T)] - \mathbf{P}^{\pm[1]}_{X}(X,T).
\end{equation}
Similarly, setting to zero the diagonal part of \eqref{eq:Lax-X-error} gives the differential identities
\begin{equation}
P^{\pm[1]}_{11,X}(X,T)=2\ii P^{\pm[1]}_{12}(X,T)P^{\pm[1]}_{21}(X,T) \quad\text{and}\quad P^{\pm[1]}_{22,X}(X,T)=- 2\ii P^{\pm[1]}_{12}(X,T)P^{\pm[1]}_{21}(X,T).
\label{eq:psi1X-diagonal}
\end{equation}
Because $\mathbf{Q}$ is invariant under conjugation by $\sigma_2$, $\mathbf{P}^{\pm}(\Lambda;X,T)$ and $\sigma_2 \mathbf{P}^{\pm}(\lambda^*;X,T)^*\sigma_2$ satisfy the same jump condition on $|\Lambda|=1$ and they enjoy the same analyticity properties and normalization as $\Lambda\to\infty$. Thus, by uniqueness $\sigma_2 \mathbf{P}^{\pm}(\Lambda^*;X,T)^*\sigma_2 = \mathbf{P}^{\pm}(\Lambda;X,T)$, which together with \eqref{eq:Psi-pm-define} implies $\Psi(X,T)^*=2\ii P^{\pm[1]}_{21}(X,T)$ and consequently the identities \eqref{eq:psi1X-diagonal} take the form
\begin{equation}
P^{\pm[1]}_{11,X}(X,T)=-\frac{\ii}{2}|\Psi^{\pm}(X,T)|^2 \quad\text{and}\quad P^{\pm[1]}_{22,X}(X,T)=\frac{\ii}{2}|\Psi^{\pm}(X,T)|^2.
\label{eq:P-diagonal}
\end{equation}
Finally, substituting $\mathbf{P}^{\pm[1]}(X,T)$ in \eqref{eq:def-A-B} we see that $\mathbf{W}^{\pm}(\Lambda;X,T)$ is for $|\Lambda|\neq 1$ a simultaneous fundamental solution matrix for the following system of first order linear differential equations
\begin{align}
\mathbf{w}_X &=\mathbf{A}^\pm\mathbf{w}=\begin{bmatrix}-\ii \Lambda & \Psi^{\pm} \\ -\Psi^{\pm} & \ii\Lambda  \end{bmatrix}\mathbf{w}\label{eq:Lax-X}\\
\mathbf{w}_T&=\mathbf{B}^\pm\mathbf{w}=\begin{bmatrix}-\ii \Lambda^2 +\ii \frac{1}{2}|\Psi^{\pm}|^2& \Lambda\Psi^{\pm} + \ii\frac{1}{2}\Psi^{\pm}_X \\ -\Lambda\Psi^{\pm*} + \ii\frac{1}{2}\Psi^{\pm*}_X & \ii \Lambda^2 -\ii \frac{1}{2}|\Psi^{\pm}|^2  \end{bmatrix}\mathbf{w}\label{eq:Lax-T}
\end{align}
which constitute the Lax pair for the nonlinear Schr\"odinger equation. The simultaneous solvability of the Lax pair implies that the matrices $\mathbf{A}^\pm$ and $\mathbf{B}^\pm$ satisfy the (zero-curvature) compatibility condition $\mathbf{A}^\pm_T-\mathbf{B}^\pm_X + [\mathbf{A}^\pm,\mathbf{B}^\pm]=\mathbf{0}$, which
is precisely the partial differential equation \eqref{eq:Psi-NLS} for $\Psi=\Psi^\pm(X,T)$.
\end{proof}
\begin{remark}
The jump matrix in Riemann-Hilbert Problem~\ref{rhp:limit} has an essential singularity at the origin, which although not on the jump contour is a point in the continuous spectrum for the associated Zakharov-Shabat scattering problem.  This suggests that $\Psi^\pm(X,T)$ might be related to solutions of the focusing nonlinear Schr\"odinger equation \eqref{eq:Psi-NLS} that generate spectral singularities of the particularly severe sort described by Zhou \cite{Zhou89a}.   On the other hand, the slow decay of $\Psi^\pm(X,T)$ as $|X|\to\infty$ that we will establish in Section~\ref{sec:Psi-asymptotic} precludes the proper definition of scattering data for the Zakharov-Shabat problem with zero boundary conditions as considered in \cite{Zhou89a}.
\end{remark}

The main result of our paper is then the following.
\begin{theorem}[Rogue waves of infinite order --- near-field limit]
Let $\psi_k(x,t)$ denote the fundamental rogue wave of order $k$ (cf., Definition~\ref{def:rogue-wave}).  Then if $k=2n$,
\begin{equation}
n^{-1}\psi_{2n}(n^{-1}X,n^{-2}T) = \Psi^+(X,T) + O(n^{-1}),\quad n\to\infty,
\end{equation}
while if instead $k=2n-1$, 
\begin{equation}
n^{-1}\psi_{2n-1}(n^{-1}X,n^{-2}T)=\Psi^-(X,T)+O(n^{-1}),\quad n\to\infty
\end{equation}
uniformly for $(X,T)$ in compact subsets of $\mathbb{R}^2$.
\label{theorem:main}
\end{theorem}
\begin{proof}
Consider the matrix $\mathbf{F}(\Lambda;X,T):=\mathbf{N}^{(k)}(n\Lambda;n^{-1}X,n^{-2}T)\mathbf{P}^\pm(\Lambda;X,T)^{-1}$, where if $k=2n$ we choose the $+$ sign and if $k=2n-1$ we choose the $-$ sign.  This matrix is analytic for $|\Lambda|\neq 1$ and tends to $\mathbb{I}$ as $\Lambda\to\infty$.  On the unit circle, according to \eqref{eq:jump-approximation} we have the jump condition 
\begin{equation}
\mathbf{F}_+(\Lambda;X,T)=\mathbf{F}_-(\Lambda;X,T)\mathbf{P}^\pm_-(\Lambda;X,T)(\mathbb{I}+O(n^{-1}))\mathbf{P}^\pm_-(\Lambda;X,T)^{-1},\quad |\Lambda|=1.
\end{equation}
Selecting a compact $K\subset\mathbb{R}^2$ and applying $\det(\mathbf{P}_-^\pm(\Lambda;X,T))\equiv 1$ along with \eqref{eq:P-bound} shows that  $\mathbf{F}_+(\Lambda;X,T)=\mathbf{F}_-(\Lambda;X,T)(\mathbb{I}+O(n^{-1}))$ holds uniformly for $(X,T)\in K$ and $|\Lambda|=1$.  Therefore $\mathbf{F}$ satisfies the conditions of a small-norm Riemann-Hilbert problem, and from standard theory it follows that $\mathbf{F}(\Lambda;X,T)=\mathbb{I}+O(n^{-1})$ holds uniformly for $(X,T)\in K$ and $\Lambda\in\mathbb{C}\setminus S^1$.  Moreover, every coefficient in the convergent Laurent series $\mathbb{I}+\mathbf{F}^{[1]}(X,T)\Lambda^{-1} + \cdots$ of $\mathbf{F}(\Lambda;X,T)$ about $\Lambda=\infty$ is also $O(n^{-1})$ uniformly for $(X,T)\in K$.  Therefore, from \eqref{eq:rogue-wave-recover-2},
\begin{equation}
\begin{split}
n^{-1}\psi_k(n^{-1}X,n^{-2}T)&=n^{-1}+2\ii n^{-1}\lim_{\lambda\to\infty}\lambda N^{(k)}_{12}(\lambda,n^{-1}X,n^{-2}T)\\ & = 
n^{-1}+2\ii n^{-1}\lim_{\Lambda\to\infty}n\Lambda\left[F_{11}(\Lambda;X,T)P^\pm_{12}(\Lambda;X,T) + F_{12}(\Lambda;X,T)P^\pm_{22}(\Lambda;X,T)\right]\\
&= n^{-1}+2\ii\lim_{\Lambda\to\infty}\Lambda P^\pm_{12}(\Lambda;X,T) + 2\ii F^{[1]}_{12}(X,T)\\
&= \Psi^\pm(X,T) + O(n^{-1})
\end{split}
\end{equation}
holds uniformly for $(X,T)\in K$, which completes the proof.
\end{proof}

Theorem~\ref{theorem:main} justifies calling the special solutions $\Psi(X,T)=\Psi^\pm(X,T)$ of the focusing nonlinear Schr\"odinger equation in the form \eqref{eq:Psi-NLS} the \emph{rogue waves of infinite order}, with the sign ``$+$'' referring to infinite even order and the sign ``$-$'' referring to infinite odd order.  Some plots of rogue waves of infinite order obtained by numerically solving Riemann-Hilbert Problem~\ref{rhp:limit} can be found in Section~\ref{sec:plots-of-Psi-plus}, and 
a computational comparison between finite-order rogue waves and the corresponding rogue wave of infinite order can be found in Section~\ref{sec:plots-of-finite-vs-infinite-order}.

\section{Exact Properties of the Near-Field Limit}
\label{sec:Exact-Properties}
To study $\Psi^\pm(X,T)$ further, it is helpful to reformulate Riemann-Hilbert Problem~\ref{rhp:limit}.  To this end, consider the matrix $\mathbf{R}^\pm(\Lambda;X,T)$ related to $\mathbf{P}^\pm(\Lambda;X,T)$ by the following explicit formula:
\begin{equation}
\mathbf{R}^\pm(\Lambda;X,T):=\begin{cases}
\mathbf{P}^\pm(\Lambda;X,T)\ee^{-\ii (\Lambda X+\Lambda^2 T)\sigma_3}\mathbf{Q}\ee^{\ii (\Lambda X+\Lambda^2T)\sigma_3},&\quad |\Lambda|<1\\
\mathbf{P}^\pm(\Lambda;X,T)\ee^{\pm 2\ii \Lambda^{-1}\sigma_3},&\quad |\Lambda|>1.
\end{cases}
\end{equation}
Noting that the matrix factors above are analytic in their respective domains and that $\ee^{\pm 2\ii \Lambda^{-1}\sigma_3}\to\mathbb{I}$ as $\Lambda\to\infty$, we see that $\mathbf{R}^\pm(\Lambda;X,T)$ satisfies the following Riemann-Hilbert problem.
\begin{rhp}[Rogue waves of infinite order --- Reformulation]
Let $(X,T)\in\mathbb{R}^2$ be arbitrary parameters.  Find a $2\times 2$ matrix $\mathbf{R}^\pm(\Lambda;X,T)$ with the following properties:
\begin{itemize}
\item[]\textbf{Analyticity:}  $\mathbf{R}^\pm(\Lambda;X,T)$ is analytic in $\Lambda$ for $|\Lambda|\neq 1$, and takes continuous boundary values on the unit circle from the interior and exterior.
\item[]\textbf{Jump condition:}  Assuming clockwise orientation of the unit circle $|\Lambda|=1$, the boundary values are related by
\begin{equation}
\mathbf{R}^\pm_+(\Lambda;X,T)=\mathbf{R}^\pm_-(\Lambda;X,T)\ee^{-\ii (\Lambda X+\Lambda^2T\pm 2\Lambda^{-1})\sigma_3}\mathbf{Q}^{-1}\ee^{\ii (\Lambda X+\Lambda ^2 T\pm 2\Lambda^{-1})\sigma_3},\quad |\Lambda|=1.
\label{eq:R-jump}
\end{equation}
\item[]\textbf{Normalization:}  $\mathbf{R}^\pm(\Lambda;X,T)\to\mathbb{I}$ as $\Lambda\to\infty$.
\end{itemize}
\label{rhp:limit-simpler}
\end{rhp}
Comparing with \eqref{eq:Psi-pm-define}, we may recover $\Psi^\pm(X,T)$ from the solution of this problem by a similar formula:
\begin{equation}
\Psi^\pm(X,T)=2\ii\lim_{\Lambda\to\infty} \Lambda R^\pm_{12}(\Lambda;X,T).
\label{eq:Psi-R}
\end{equation}
\subsection{Basic symmetries}
\label{sec:symmetries}
The formulation of Riemann-Hilbert Problem~\ref{rhp:limit-simpler} makes it easy to relate explicitly $\mathbf{R}^+(\Lambda;X,T)$ and $\mathbf{R}^-(\Lambda;X,T)$.
\begin{proposition}
We have the identity
\begin{equation}
\mathbf{R}^\mp(\Lambda;X,T)=\begin{cases}
\sigma_3\mathbf{R}^\pm(\Lambda;X,T)\ee^{\mp 4\ii \Lambda^{-1}\sigma_3}\sigma_3,&\quad |\Lambda|>1\\
\sigma_3\mathbf{R}^\pm(\Lambda;X,T)\ee^{-2\ii (\Lambda X+\Lambda^2T)\sigma_3}(\ii\sigma_2)\sigma_3,&\quad |\Lambda|<1.
\end{cases}
\label{eq:involution}
\end{equation}
\end{proposition}
\begin{proof}
The right-hand side of \eqref{eq:involution} is analytic for $|\Lambda|\neq 1$ and tends to the identity as $\Lambda\to\infty$.  It remains only to check the jump condition for $\mathbf{R}^\pm$ using \eqref{eq:R-jump} for $\mathbf{R}^\pm$:
\begin{equation}
\begin{split}
\mathbf{R}^\mp_+(\Lambda;X,T)&=\sigma_3\mathbf{R}^\pm_+(\Lambda;X,T)\ee^{\mp 4\ii \Lambda^{-1}\sigma_3}\sigma_3\\
&=\sigma_3\mathbf{R}^\pm_-(\Lambda;X,T)\ee^{-\ii(\Lambda X+\Lambda^2T\pm 2\Lambda^{-1})\sigma_3}\mathbf{Q}^{-1}\ee^{\ii (\Lambda X+\Lambda^2T\pm 2\Lambda^{-1})\sigma_3}\ee^{\mp 4\ii \Lambda^{-1}\sigma_3}\sigma_3\\
&=\sigma_3\mathbf{R}^\pm_-(\Lambda;X,T)\ee^{-\ii(\Lambda X+\Lambda^2T\pm 2\Lambda^{-1})\sigma_3}\mathbf{Q}^{-1}\ee^{\ii (\Lambda X+\Lambda^2T\mp 2\Lambda^{-1})\sigma_3}\sigma_3\\
&=\mathbf{R}^\mp_-(\Lambda;X,T)\sigma_3(-\ii\sigma_2)\ee^{2\ii (\Lambda X+\Lambda^2T)\sigma_3}
\ee^{-\ii(\Lambda X+\Lambda^2T\pm 2\Lambda^{-1})\sigma_3}\mathbf{Q}^{-1}\ee^{\ii (\Lambda X+\Lambda^2T\mp 2\Lambda^{-1})\sigma_3}\sigma_3\\
&=\mathbf{R}^\mp_-(\Lambda;X,T)\sigma_3(-\ii\sigma_2)\ee^{\ii(\Lambda X+\Lambda^2T\mp 2\Lambda^{-1})\sigma_3}\mathbf{Q}^{-1}\ee^{\ii (\Lambda X+\Lambda^2T\mp 2\Lambda^{-1})\sigma_3}\sigma_3\\
&=\mathbf{R}^\mp_-(\Lambda;X,T)\ee^{-\ii(\Lambda X+\Lambda^2T\mp 2\Lambda^{-1})\sigma_3}\sigma_3(-\ii\sigma_2)\mathbf{Q}^{-1}\sigma_3\ee^{\ii(\Lambda X+\Lambda^2T\mp 2\Lambda^{-1})\sigma_3}\\
&=\mathbf{R}^\mp_-(\Lambda;X,T)\ee^{-\ii (\Lambda X+\Lambda^2T\mp 2\Lambda^{-1})\sigma_3}\mathbf{Q}^{-1}\ee^{\ii(\Lambda X+\Lambda^2T\mp 2\Lambda^{-1})\sigma_3}
\end{split}
\end{equation}
because $\sigma_3(-\ii \sigma_2)\mathbf{Q}^{-1}\sigma_3=\mathbf{Q}^{-1}$, so this indeed matches \eqref{eq:R-jump} for $\mathbf{R}^\mp$.
\end{proof}
\begin{corollary}
$\Psi^-(X,T)=-\Psi^+(X,T)$ holds for all $(X,T)\in\mathbb{R}^2$.
\label{cor:Psi-pm}
\end{corollary}
\begin{proof}
The conjugation by $\sigma_3$ changes the signs of the off-diagonal entries.  The diagonal factors mediating between $\mathbf{P}^\pm(\Lambda;X,T)$ and $\mathbf{R}^\pm(\Lambda;X,T)$ and between $\mathbf{R}^+(\Lambda;X,T)$ and $\mathbf{R}^-(\Lambda;X,T)$ for $|\Lambda|>1$ have no effect on the leading off-diagonal entries as $\Lambda\to\infty$.
\end{proof}
An even easier result stems from considering a change of spectral parameter $\Lambda\mapsto -\Lambda$:
\begin{proposition}
$\mathbf{R}^\mp(-\Lambda;-X,T)=\mathbf{R}^\pm(\Lambda;X,T)$.
\label{prop:minus-z}
\end{proposition}
\begin{corollary}
$\Psi^\pm(-X,T)=\Psi^\pm(X,T)$.
\label{cor:Psi-even}
\end{corollary}
\begin{proof}
Using \eqref{eq:Psi-R} and Proposition~\ref{prop:minus-z},
\begin{equation}
\Psi^\pm(-X,T)=2\ii\lim_{\Lambda\to\infty}\Lambda R^\pm_{12}(\Lambda;-X,T)=2\ii\lim_{\Lambda\to\infty}\Lambda R^\mp_{12}(-\Lambda;X,T).
\end{equation}
Now replacing $\Lambda$ with $-\Lambda$,
\begin{equation}
\Psi^\pm(-X,T)=-2\ii\lim_{\Lambda\to\infty}\Lambda R^\mp_{12}(\Lambda;X,T) = -\Psi^\mp(X,T)
\end{equation}
from which the result follows by Corollary~\ref{cor:Psi-pm}.
\end{proof}
A related symmetry arises from $\Lambda\mapsto-\Lambda^*$:
\begin{proposition}
$\mathbf{R}^\pm(-\Lambda^*;X,-T)^*=\mathbf{R}^\pm(\Lambda;X,T)$ holds for all $(X,T)\in\mathbb{R}^2$.
\label{prop:minus-conjugate-z}
\end{proposition}
\begin{corollary}
$\Psi^\pm(X,-T)=\Psi^\pm(X,T)^*$.  In particular, $\Psi^\pm(X,0)$ is real-valued.
\label{cor:Psi-real}
\end{corollary}
\begin{proof}
Using \eqref{eq:Psi-R} and Proposition~\ref{prop:minus-conjugate-z},
\begin{equation}
\Psi^\pm(X,-T)=2\ii\lim_{\Lambda\to\infty}\Lambda R_{12}^\pm(\Lambda;X,-T)=2\ii\lim_{\Lambda\to\infty}
\Lambda R_{12}^\pm(-\Lambda^*;X,T)^*=-\left[2\ii\lim_{\Lambda\to\infty}\Lambda^*R_{12}^\pm(-\Lambda^*;X,T)\right]^*.
\end{equation}
Now replacing $\Lambda$ with $-\Lambda^*$,
\begin{equation}
\Psi^\pm(X,-T)=\left[2\ii\lim_{\Lambda\to\infty}\Lambda R_{12}^\pm(\Lambda;X,T)\right]^*=\Psi^\pm(X,T)^*
\end{equation}
according to \eqref{eq:Psi-R}.
\end{proof}
\begin{proposition}
$\Psi^\pm(0,0)=\pm 4$.
\label{prop:Psi-peak}
\end{proposition}
\begin{proof}
This result can be deduced by combining Proposition~\ref{prop:origin-value} with Theorem~\ref{theorem:main}, but we can also give the following independent proof.
If $X=T=0$, it is easy to see that the solution of Riemann-Hilbert Problem~\ref{rhp:limit-simpler} is simply
\begin{equation}
\mathbf{R}^\pm(\Lambda;0,0)=\begin{cases}
\mathbf{Q},&\quad |\Lambda|<1\\
\mathbf{Q}\ee^{\mp 2\ii \Lambda^{-1}\sigma_3}\mathbf{Q}^{-1}\ee^{\pm 2\ii \Lambda^{-1}\sigma_3},&\quad |\Lambda|>1.
\end{cases}
\end{equation}
We observe that 
\begin{equation}
\begin{split}
\mathbf{Q}\ee^{\mp 2\ii \Lambda^{-1}\sigma_3}\mathbf{Q}^{-1}\ee^{\pm 2\ii \Lambda^{-1}\sigma_3}& = \mathbf{Q}(\mathbb{I}\mp 2\ii\sigma_3\Lambda^{-1}+O(\Lambda^{-2}))\mathbf{Q}^{-1}(\mathbb{I}\pm 2\ii\sigma_3\Lambda^{-1}+O(\Lambda^{-2}))\\
& = \mathbb{I} \pm 2\ii(\sigma_3 - \mathbf{Q}\sigma_3\mathbf{Q}^{-1})\Lambda^{-1}+O(\Lambda^{-1})\\
&=\mathbb{I}\pm 2\ii (\sigma_3-\sigma_1)\Lambda^{-1}+O(\Lambda^{-2}),\quad \Lambda\to\infty.
\end{split}
\end{equation}
Therefore,
\begin{equation}
\Psi^\pm(0,0)=2\ii\lim_{\Lambda\to\infty}\Lambda R_{12}^\pm(\Lambda;0,0)=\pm 4
\end{equation}
which completes the proof.
\end{proof}

\subsection{Differential equations}
\label{sec:ODE}
Proposition~\ref{prop:basic} showed that the functions $\Psi^\pm(X,T)$ satisfy the partial differential equation \eqref{eq:Psi-NLS}.  The goal of this section is to show that these special solutions of the focusing nonlinear Schr\"odinger equation also satisfy certain ordinary differential equations in $X$ for each fixed $T$ as well as certain other ordinary differential equations in $T$ for each fixed $X$.
According to Corollary~\ref{cor:Psi-pm}, which explicitly relates $\Psi^-(X,T)$ to $\Psi^+(X,T)$, it suffices to consider the function $\Psi(X,T):=\Psi^+(X,T)$, and we will do so for the rest of this section. 
\subsubsection{Lax systems related to Riemann-Hilbert Problem~\ref{rhp:limit-simpler}}
The function $\Psi(X,T)=\Psi^+(X,T)$ is encoded in the solution $\mathbf{R}(\Lambda;X,T):=\mathbf{R}^+(\Lambda;X,T)$ of Riemann-Hilbert Problem~\ref{rhp:limit-simpler} in the ``$+$'' case.
The latter problem has a jump matrix in which all dependence on $\Lambda$ as well as $(X,T)\in\mathbb{R}^2$ appears only in conjugating exponential factors.  Therefore, as in the proof of Proposition~\ref{prop:basic}, we begin by setting
\begin{equation}
\mathbf{W}(\Lambda;X,T):=\mathbf{R}(\Lambda;X,T)\ee^{-\ii(\Lambda X + \Lambda^2 T + 2\Lambda^{-1})\sigma_3}.
\label{eq:W-def}
\end{equation}
This transformation removes the dependence on all three variables $(\Lambda;X,T)$ from the jump condition, and hence $\mathbf{W}(\Lambda;X,T)$ is analytic for $|\Lambda|\neq 1$ and satisfies the simple jump condition
\begin{equation}
\mathbf{W}_+(\Lambda;X,T) = \mathbf{W}_-(\Lambda;X,T) \mathbf{Q}^{-1},\quad |\Lambda|=1.
\label{eq:W-simple-jump}
\end{equation}
Exactly as in the proof of Proposition~\ref{prop:basic} it follows immediately that $\mathbf{W}(\Lambda;X,T)$ satisfies the Lax pair equations
\begin{equation}
\frac{\partial\mathbf{W}}{\partial X}(\Lambda;X,T)=\mathbf{A}(\Lambda;X,T)\mathbf{W}(\Lambda;X,T)
\label{eq:Lax-X-ODEs}
\end{equation}
and
\begin{equation}
\frac{\partial\mathbf{W}}{\partial T}(\Lambda;X,T)=\mathbf{B}(\Lambda;X,T)\mathbf{W}(\Lambda;X,T)
\label{eq:Lax-T-ODEs}
\end{equation}
where
\begin{equation}
\mathbf{A}(\Lambda;X,T):=\begin{bmatrix}-\ii\Lambda & \Psi(X,T)\\-\Psi(X,T)^* & \ii\Lambda\end{bmatrix}
\end{equation}
and
\begin{equation}
\mathbf{B}(\Lambda;X,T):=\begin{bmatrix}-\ii\Lambda^2 + \tfrac{1}{2}\ii|\Psi(X,T)|^2 & 
\Lambda\Psi(X,T)+\Phi(X,T)\\-\Lambda\Psi(X,T)^* -\Phi(X,T)^* & \ii\Lambda^2-\tfrac{1}{2}\ii |\Psi(X,T)|^2
\end{bmatrix},
\end{equation}
in which the potentials $\Psi(X,T)$ and $\Phi(X,T)$ can be found from the coefficients in the convergent Laurent series
\begin{equation}
\mathbf{R}(\Lambda;X,T)=\mathbb{I}+\sum_{j=1}^\infty \mathbf{R}^{[j]}(X,T)\Lambda^{-j},\quad|\Lambda|>1
\label{eq:R-series-infinity}
\end{equation}
by the formul\ae\
\begin{equation}
\Psi(X,T)=2\ii R^{[1]}_{12}(X,T)\quad\text{and}\quad\Psi(X,T)^*=2\ii R^{[1]}_{21}(X,T),
\end{equation}
\begin{equation}
\begin{split}
\Phi(X,T)&=2\ii\left(R^{[2]}_{12}(X,T)-R^{[1]}_{12}(X,T)R^{[1]}_{22}(X,T)\right)\\
\Phi(X,T)^*&=2\ii\left(R^{[2]}_{21}(X,T)-R^{[1]}_{21}(X,T)R^{[1]}_{11}(X,T)\right).
\end{split}
\end{equation}
In fact, 
\begin{equation}
\mathbf{A}(\Lambda;X,T)=-\ii\Lambda\sigma_3 + \mathbf{A}^{[0]}(X,T),
\end{equation}
where
\begin{equation}
\mathbf{A}^{[0]}(X,T):=\ii[\sigma_3,\mathbf{R}^{[1]}(X,T)]=\begin{bmatrix}0 & \Psi(X,T)\\-\Psi(X,T)^* & 0\end{bmatrix}
\end{equation}
and
\begin{equation}
\mathbf{B}(\Lambda;X,T)=-\ii\Lambda^2\sigma_3 +\mathbf{B}^{[1]}(X,T)\Lambda +\mathbf{B}^{[0]}(X,T),
\end{equation}
where
\begin{equation}
\begin{split}
\mathbf{B}^{[1]}(X,T)&:=\mathbf{A}^{[0]}(X,T)\\
\mathbf{B}^{[0]}(X,T)&:=\ii[\mathbf{R}^{[1]}(X,T),\sigma_3\mathbf{R}^{[1]}(X,T)] + \ii[\sigma_3,\mathbf{R}^{[2]}(X,T)]=
\begin{bmatrix}\tfrac{1}{2}\ii |\Psi(X,T)|^2 & \Phi(X,T)\\-\Phi(X,T)^* & -\tfrac{1}{2}\ii |\Psi(X,T)|^2\end{bmatrix}.
\end{split}
\end{equation}

Since according to \eqref{eq:W-simple-jump} the jump matrix for $\mathbf{W}(\Lambda;X,T)$ is also independent of $\Lambda$, it is possible to obtain an additional Lax equation by differentiating with respect to $\Lambda$.  Thus we find that the matrix $\mathbf{L}(\Lambda;X,T)$ defined by
\begin{equation}
\mathbf{L}(\Lambda;X,T):=\mathbf{W}'(\Lambda;X,T)\mathbf{W}(\Lambda;X,T)^{-1},\quad \prime:=\frac{\dd}{\dd\Lambda}
\label{eq:L-def}
\end{equation}
has no jump across the unit circle and hence may be considered to be analytic in the whole complex $\Lambda$-plane, with the possible exception only of an isolated singularity at $\Lambda=0$ arising from differentiation of the exponential factor $\ee^{-2\ii\Lambda^{-1}\sigma_3}$.

We will now determine $\mathbf{L}(\Lambda;X,T)$. By definition
\begin{equation}
\mathbf{L}(\Lambda;X,T) = \mathbf{R}'(\Lambda;X,T)\mathbf{R}(\Lambda;X,T)^{-1} -\ii(X+2T\Lambda -2 \Lambda^{-2})\mathbf{R}(\Lambda;X,T)\sigma_3\mathbf{R}(\Lambda;X,T)^{-1}
\end{equation}
and from the series \eqref{eq:R-series-infinity}, we obtain the expansion
\begin{equation}
\mathbf{L}(\Lambda;X,T)=\mathbf{L}^{[1]}(X,T)\Lambda +\mathbf{L}^{[0]}(X,T) + \mathbf{L}^{[-1]}(X,T)\Lambda^{-1}+
\mathbf{L}^{[-2]}(X,T)\Lambda^{-2}+O(\Lambda^{-3}),\quad\Lambda\to\infty,
\label{eq:L-infinity}
\end{equation}
where
\begin{equation}
\begin{split}
\mathbf{L}^{[1]}(X,T)&:=-2\ii T\sigma_3\\
\mathbf{L}^{[0]}(X,T)&:= -\ii X\sigma_3+2T\mathbf{B}^{[1]}(X,T)\\
\mathbf{L}^{[-1]}(X,T)&:= X\mathbf{A}^{[0]}(X,T)+2 T\mathbf{B}^{[0]}(X,T)\\
\mathbf{L}^{[-2]}(X,T)&:=2\ii\sigma_3-\mathbf{R}^{[1]}(X,T) + X\mathbf{B}^{[0]}(X,T)\\
&\quad\quad{}+2\ii T\left([\mathbf{R}^{[1]}(X,T),\sigma_3\mathbf{R}^{[2]}(X,T)]+[\mathbf{R}^{[2]}(X,T),\sigma_3\mathbf{R}^{[1]}(X,T)]
\right.\\
&\quad\quad\quad\quad\quad\quad\left.{}+[\sigma_3\mathbf{R}^{[1]}(X,T)^2,\mathbf{R}^{[1]}(X,T)]+[\sigma_3,\mathbf{R}^{[3]}(X,T)]\right).
\end{split}
\end{equation}
Likewise, from the Taylor expansions at the origin
\begin{equation}
\begin{split}
\mathbf{R}(\Lambda;X,T)&=\mathbf{R}(0;X,T) + \mathbf{R}'(0;X,T)\Lambda + O(\Lambda^2),\quad\Lambda\to 0\\
\mathbf{R}(\Lambda;X,T)^{-1}&=\mathbf{R}(0;X,T)^{-1}-\mathbf{R}(0;X,T)^{-1}\mathbf{R}'(0;X,T)\mathbf{R}(0;X,T)^{-1}\Lambda + O(\Lambda^2),\quad\Lambda\to 0,
\end{split}
\end{equation}
we get 
\begin{multline}
\mathbf{L}(\Lambda;X,T)=2\ii\mathbf{R}(0;X,T)\sigma_3\mathbf{R}(0;X,T)^{-1}\Lambda^{-2} \\+ 
2\ii[\mathbf{R}'(0;X,T)\mathbf{R}(0;X,T)^{-1},\mathbf{R}(0;X,T)\sigma_3\mathbf{R}(0;X,T)^{-1}]\Lambda^{-1}+O(1),\quad\Lambda\to 0.
\label{eq:L-zero}
\end{multline}
Comparing \eqref{eq:L-infinity} with \eqref{eq:L-zero} shows that $\mathbf{L}(\Lambda;X,T)$ is the Laurent polynomial
\begin{equation}
\mathbf{L}(\Lambda;X,T) = \mathbf{L}^{[1]}(X,T)\Lambda + \mathbf{L}^{[0]}(X,T) + \mathbf{L}^{[-1]}(X,T)\Lambda^{-1} + \mathbf{L}^{[-2]}(X,T)\Lambda^{-2}.
\label{eq:def-L}
\end{equation}
Moreover, we obtain an equivalent representation for $\mathbf{L}^{[-2]}(X,T)$, namely
\begin{equation}
\mathbf{L}^{[-2]}(X,T)=2\ii\mathbf{R}(0;X,T)\sigma_3\mathbf{R}(0;X,T)^{-1}
\end{equation}
which shows that $\mathrm{tr}(\mathbf{L}^{[-2]}(X,T))=0$ and $\det(\mathbf{L}^{[-2]}(X,T))=4$.  Taking into account the Schwarz symmetry satisfied by $\mathbf{R}(\Lambda;X,T)$:
\begin{equation}
\mathbf{R}(\Lambda;X,T)=\sigma_2\mathbf{R}(\Lambda^*;X,T)^*\sigma_2,\quad |\Lambda|\neq 1,\quad (X,T)\in\mathbb{R}^2,
\end{equation}
it follows that $\mathbf{L}^{[-2]}(X,T)$ is a matrix with the form
\begin{equation}
\mathbf{L}^{[-2]}(X,T)=\begin{bmatrix}\ii a(X,T) & \ii b(X,T)\\\ii b(X,T)^* & -\ii a(X,T)\end{bmatrix},\quad \text{where}\quad a:\mathbb{R}^2\to\mathbb{R}\quad\text{and}\quad a(X,T)^2 + |b(X,T)|^2 = 4.
\label{eq:cons-law}
\end{equation}
With $\mathbf{L}(\Lambda;X,T)$ defined in this way, we reinterpret the definition \eqref{eq:L-def}
as the Lax system
\begin{equation}
\mathbf{W}'(\Lambda;X,T)=\mathbf{L}(\Lambda;X,T)\mathbf{W}(\Lambda;X,T),
\label{eq:Lax-Lambda-ODEs}
\end{equation}

\subsubsection{Ordinary differential equations in $X$}
Since $\mathbf{W}(\Lambda;X,T)$ is simultaneously a fundamental solution matrix of the first-order linear Lax systems \eqref{eq:Lax-X-ODEs} and \eqref{eq:Lax-Lambda-ODEs}, the coefficient matrices $\mathbf{A}(\Lambda;X,T)$ and $\mathbf{L}(\Lambda;X,T)$ necessarily satisfy the zero-curvature condition $\mathbf{L}_X-\mathbf{A}_\Lambda + [\mathbf{L},\mathbf{A}]=\mathbf{0}$.  The left-hand side is a Laurent polynomial in $\Lambda$ with powers ranging from $\Lambda^2$ through $\Lambda^{-2}$, and therefore its coefficients must all vanish.  The equation arising from terms proportional to $\Lambda^2$ reads $[\mathbf{L}^{[1]},-\ii\sigma_3]=\mathbf{0}$, which holds automatically.  Similarly, the terms proportional to $\Lambda$ give the equation $\mathbf{L}^{[1]}_X+[\mathbf{L}^{[1]},\mathbf{A}^{[0]}]+[\mathbf{L}^{[0]},-\ii\sigma_3]=\mathbf{0}$, which holds automatically because $\mathbf{B}^{[1]}(X,T)\equiv\mathbf{A}^{[0]}(X,T)$.  The first nontrivial information comes from the terms proportional to $\Lambda^0$, giving the equation $\mathbf{L}^{[0]}_X+\ii\sigma_3 + [\mathbf{L}^{[0]},\mathbf{A}^{[0]}]+[\mathbf{L}^{[-1]},-\ii\sigma_3]=\mathbf{0}$, the diagonal elements of which give no information, but the off-diagonal elements read
\begin{equation}
\begin{split}
T\Psi_X+2\ii T\Phi&=0\\ T\Psi_X^* -2\ii T\Phi^*&=0.
\end{split}
\label{eq:X-eqns-1}
\end{equation}
The terms proportional to $\Lambda^{-1}$ give the equation $\mathbf{L}^{[-1]}_X+[\mathbf{L}^{[-1]},\mathbf{A}^{[0]}]+[\mathbf{L}^{[-2]},-\ii\sigma_3]=\mathbf{0}$.  Here one can easily confirm that the diagonal terms reproduce again the same conditions \eqref{eq:X-eqns-1}, while the off-diagonal terms read
\begin{equation}
\begin{split}
(X\Psi)_X -2b + 2T\Phi_X + 2\ii T|\Psi|^2\Psi&=0\\
(X\Psi^*)_X-2b^*+2T\Phi^*_X-2\ii T|\Psi|^2\Psi^*&=0.
\end{split}
\label{eq:X-eqns-2}
\end{equation}
Finally, the terms proportional to $\Lambda^{-2}$ give the equation $\mathbf{L}^{[-2]}_X + [\mathbf{L}^{[-2]},\mathbf{A}^{[0]}]=\mathbf{0}$, which is equivalent to
\begin{equation}
\begin{split}
a_X-b\Psi^*-b^*\Psi&=0\\ 
b_X+2a\Psi&=0\\ 
b^*_X+2a\Psi^*&=0.
\end{split}
\label{eq:X-eqns-3}
\end{equation}

Our primary interest is in the function $\Psi(X,T)$, so we first observe that the product $T\Phi$ can be explicitly eliminated using \eqref{eq:X-eqns-1}, so that \eqref{eq:X-eqns-2} can be replaced with
\begin{equation}
\begin{split}
(X\Psi)_X-2b+\ii T\Psi_{XX}+2\ii T|\Psi|^2\Psi&=0\\
(X\Psi^*)_X-2b^*-\ii T\Psi^*_{XX}-2\ii T|\Psi|^2\Psi^*&=0.
\end{split}
\label{eq:X-eqns-2-no-Phi}
\end{equation}
Next, \eqref{eq:X-eqns-2-no-Phi} can be used to explicitly eliminate $b$ and $b^*$, so that \eqref{eq:X-eqns-3} becomes
\begin{equation}
\begin{split}
M_X&=|\Psi|^2\\
(X\Psi)_{XX}+2X|\Psi|^2\Psi+2M\Psi + \ii T\Psi_{XXX}+6\ii T|\Psi|^2\Psi_X&=0\\
(X\Psi^*)_{XX}+2X|\Psi|^2\Psi^*+2M\Psi^*-\ii T\Psi^*_{XXX}-6\ii T|\Psi|^2\Psi^*_X&=0,
\end{split}
\end{equation}
where it has become convenient to introduce
\begin{equation}
M(X,T):=2a(X,T)-X|\Psi(X,T)|^2-\ii T\left(\Psi(X,T)^*\Psi_X(X,T)-\Psi(X,T)\Psi_X(X,T)^*\right).
\end{equation}
Finally, dividing by $\Psi$ and taking another derivative allows $M$ to be eliminated, leaving the following fourth-order ordinary differential equation for $\Psi(X,T)$ as a function of $X$ for fixed $T$:
\begin{multline}
X\Psi\Psi_{XXX}+3\Psi\Psi_{XX}-X\Psi_X\Psi_{XX}-2(\Psi_{X})^2 + 4\Psi^3\Psi^*+2X\Psi^2\Psi^*\Psi_X + 2X\Psi^3\Psi_X^* \\{}+\ii T\left(\Psi\Psi_{XXXX}-\Psi_X\Psi_{XXX}+6\Psi^2\Psi_X\Psi_X^*+6\Psi^2\Psi^*\Psi_{XX}\right)=0
\label{eq:Psi-eqn-X}
\end{multline}
and its complex conjugate.  Another ordinary differential equation of lower order can also be obtained by using the conservation law $a^2+|b|^2=4$ (cf., \eqref{eq:cons-law}).  First we use \eqref{eq:X-eqns-3} to express $a^2+|b|^2$ in the form
\begin{equation}
a^2+|b|^2=aa^*+bb^*=\frac{b_Xb^*_X}{4|\Psi|^2}+bb^*,
\end{equation}
and then explicitly eliminate $b$ and $b^*$ using \eqref{eq:X-eqns-2-no-Phi} to find
\begin{equation}
\left|\left((X\Psi)_X+\ii T\Psi_{XX}+2\ii T|\Psi|^2\Psi\right)_X\right|^2 + 4|\Psi|^2\left|(X\Psi)_X+\ii T\Psi_{XX}+2\ii T|\Psi|^2\Psi\right|^2=64|\Psi|^2.
\label{eq:Psi-eqn-X-conservation}
\end{equation}
Therefore we have proved the following.
\begin{theorem}[Ordinary differential equations in $X$ for rogue waves of infinite order]
The rogue wave $\Psi(X,T)=\Psi^+(X,T)$ of infinite order satisfies, for each fixed $T\in\mathbb{R}$, the ordinary differential equations \eqref{eq:Psi-eqn-X} and \eqref{eq:Psi-eqn-X-conservation} with respect to $X$.
\label{theorem:X-ODEs}
\end{theorem}

When $T\neq 0$, after making some necessary but unimportant rescalings, the compatible linear equations \eqref{eq:Lax-X-ODEs} and \eqref{eq:Lax-Lambda-ODEs} 
fit into the scheme of Sakka \cite{Sakka09} for a hierarchy generalizing the Painlev\'e-III equation.  In particular, for $T\neq 0$ the nonlinear equations \eqref{eq:Psi-eqn-X} and \eqref{eq:Psi-eqn-X-conservation} are connected to the second equation in Sakka's Painlev\'e-III hierarchy (see \cite[Sec.\@ 4, Example 1]{Sakka09} in which, after correcting for a typo, Sakka's matrices $A$ and $B$ correspond with $\mathbf{L}$ and $\mathbf{A}$ in our notation respectively).  The special solution $\Psi(X,T)$ corresponds to Sakka's integrals having values $\gamma_1=\gamma_2=0$ and $\gamma_3=16$.

When $T=0$, the equations \eqref{eq:Lax-X-ODEs} and \eqref{eq:Lax-Lambda-ODEs} correspond instead to the first member of the hierarchy, namely the Painlev\'e-III equation itself.  To make this connection more concrete, we first recall from Corollary~\ref{cor:Psi-real} that $\Psi(X,0)$ is real-valued, so the equations \eqref{eq:Psi-eqn-X} and \eqref{eq:Psi-eqn-X-conservation} take the simpler form
\begin{equation}
X\Psi\Psi_{XXX}+3\Psi\Psi_{XX}-X\Psi_X\Psi_{XX}-2(\Psi_X)^2+4X\Psi^3\Psi_X+4\Psi^4=0,\quad T=0,
\label{eq:Psi-eqn-X-T0}
\end{equation}
and 
\begin{equation}
((X\Psi)_{XX})^2+4\Psi^2((X\Psi)_X)^2=64\Psi^2,\quad T=0.
\label{eq:Psi-eqn-X-conservation-T0}
\end{equation}
Dividing both of these equations through by $\Psi^2$ and introducing $V:=\Psi_X/\Psi$, they can be written respectively as
\begin{equation}
XV_{XX}+2XVV_X+3V_X+V^2+4X\Psi\Psi_X+4\Psi^2=0
\end{equation}
and
\begin{equation}
(XV_X+XV^2+2V)^2+(XV+1)(4X\Psi\Psi_X+4\Psi^2)=64.
\end{equation}
Eliminating $4X\Psi\Psi_X+4\Psi^2$ yields a second-order quasilinear equation on $V$ alone:
\begin{equation}
X^2VV_{XX}+XV_{XX}+XVV_X-X^2(V_X)^2+3V_X-X^2V^4-3XV^3-3V^2+64=0.
\label{eq:V-eqn}
\end{equation}
The motivation for combining \eqref{eq:Psi-eqn-X-T0}--\eqref{eq:Psi-eqn-X-conservation-T0} in such a way as to obtain a single differential equation for $V=\Psi_X/\Psi$ alone can be explained as follows.  When $T=0$, the exponent in Riemann-Hilbert Problem~\ref{rhp:limit-simpler} can be rescaled by 
\begin{equation}
X=-\frac{1}{8}x^2\quad\text{and}\quad \Lambda=\frac{4}{x}\lambda
\label{eq:PIII-scalings}
\end{equation}
so as to yield the identity $\Lambda X+2\Lambda^{-1}=-\tfrac{1}{2}x(\lambda-\lambda^{-1})$; on the right-hand side we now have the exponent appearing in the inverse monodromy problem for the Painlev\'e-III equation (see \cite[Theorem 5.4]{FokasIKN06}) obtained from the Lax pair of Jimbo and Miwa \cite{JimboM81}.  Indeed, combining the Lax systems \eqref{eq:Lax-X-ODEs} and \eqref{eq:Lax-Lambda-ODEs} using \eqref{eq:PIII-scalings} yields the Jimbo-Miwa Lax pair for Painlev\'e-III in the form
\begin{equation}
\frac{\partial\mathbf{W}}{\partial \lambda}=\widehat{\mathbf{L}}(\lambda;x)\mathbf{W}\quad\text{and}\quad
\frac{\partial\mathbf{W}}{\partial x}=\widehat{\mathbf{A}}(\lambda;x)\mathbf{W},
\label{eq:JM}
\end{equation}
in which the coefficient matrices are\footnote{Jimbo and Miwa used a slightly-different parametrization, preferring the combinations $U=st$ and $w:=1/t$ instead of $s$ and $t$.}
\begin{equation}
\begin{split}
\widehat{\mathbf{L}}(\lambda;x):=&\frac{4}{x}\mathbf{L}\left(\frac{4}{x}\lambda;-\frac{1}{8}x^2,0\right)\\
=&\frac{1}{2}\ii x\sigma_3 + \frac{1}{\lambda}\begin{bmatrix}-\tfrac{1}{2}\Theta_\infty & y\\v & \tfrac{1}{2}\Theta_\infty\end{bmatrix} + 
\frac{1}{\lambda^2}\begin{bmatrix}\tfrac{1}{2}\ii x-\ii st & \ii s\\-\ii t(st-x) & -\tfrac{1}{2}\ii x+\ii st\end{bmatrix}
\end{split}
\end{equation}
and
\begin{equation}
\begin{split}
\widehat{\mathbf{A}}(\lambda;x):=&-\frac{4\lambda}{x^2}\mathbf{L}\left(\frac{4}{x}\lambda;-\frac{1}{8}x^2,0\right)-\frac{1}{4}x\mathbf{A}\left(\frac{4}{x}\lambda;-\frac{1}{8}x^2,0\right)\\
=&\frac{1}{2}\ii\lambda\sigma_3 +\frac{1}{x}\begin{bmatrix}0 & y\\v & 0\end{bmatrix}-\frac{1}{x\lambda}
\begin{bmatrix}\tfrac{1}{2}\ii x-\ii st & \ii s\\-\ii t(st-x) & -\tfrac{1}{2}\ii x+\ii st\end{bmatrix},
\end{split}
\end{equation}
where
\begin{equation}
\Theta_\infty:=0, \quad y:=-\frac{1}{8}x^2\Psi,\quad v:=\frac{1}{8}x^2\Psi^*, \quad \frac{1}{2}\ii x-\ii st := \frac{1}{4}\ii xa,\quad s:=\frac{1}{4}xb.
\end{equation}
The combination $u(x):=-y(x)/s(x)$ was shown by Jimbo and Miwa to solve the (generic) Painlev\'e-III equation 
\begin{equation}
\frac{\dd^2u}{\dd x^2}=\frac{1}{u}\left(\frac{\dd u}{\dd x}\right)^2-\frac{1}{x}\frac{\dd u}{\dd x}+\frac{4\Theta_0 u^2+4(1-\Theta_\infty)}{x}+4u^3-\frac{4}{u}
\label{eq:PIII}
\end{equation}
in which the parameter $\Theta_\infty$ appears as an explicit coefficient in the $\lambda$-equation (here $\Theta_\infty=0$) and $\Theta_0$ is obtained as the value of an integral of motion.  Rather than compute this integral, we may simply note that 
\begin{equation}
u(x)=-\frac{y(x)}{s(x)}=\frac{x\Psi(-\tfrac{1}{8}x^2)}{2b(-\tfrac{1}{8}x^2)}=\frac{2x^2\Psi(-\tfrac{1}{8}x^2)}{(x^2\Psi(-\tfrac{1}{8}x^2))_x} = 2\left(\frac{\dd}{\dd x}\ln(x^2\Psi(-\tfrac{1}{8}x^2))\right)^{-1}
\label{eq:u-Psi}
\end{equation}
where we have used \eqref{eq:X-eqns-2-no-Phi} at $T=0$ to eliminate $b$.  Inverting this relationship we may find $V=\Psi_X/\Psi$ in terms of $u$:
\begin{equation}
V=\frac{\Psi_X}{\Psi}=\frac{8}{x}\left(\frac{1}{x}-\frac{1}{u}\right).
\end{equation}
Substituting this formula into \eqref{eq:V-eqn} and using the chain rule to express derivatives in terms of $x$ rather than $X$ yields the following result.
\begin{corollary}
The function $u(x)$ defined explicitly in terms of the rogue wave of infinite order $\Psi(X)=\Psi^+(X,0)$ by \eqref{eq:u-Psi} is a solution of the Painlev\'e-III equation in the standard form \eqref{eq:PIII}
in which both parameters vanish:  $\Theta_\infty=\Theta_0=0$.  
\label{corollary:PIII}
\end{corollary}

In general, the inverse monodromy problem for the system \eqref{eq:JM} can be formulated as a Riemann-Hilbert problem in the $\lambda$-plane for a $2\times 2$ matrix unknown that has jumps across two Stokes lines emanating in opposite directions from $\lambda=0$, two Stokes lines tending to $\lambda=\infty$ in opposite directions, as well as a jump relating the solution in a neighborhood of $\lambda=0$ to that in a neighborhood of $\lambda=\infty$ given in terms of a \emph{connection matrix}.  The parameters $\Theta_\infty$ and $\Theta_0$ measure the formal monodromy for solutions of the $\lambda$-equation about $\lambda=\infty$ and $\lambda=0$ respectively.  In the present setting, there is no formal monodromy because $\Theta_0=\Theta_\infty=0$, however, in principle there can still be Stokes phenomenon near $\lambda=\infty$ and $\lambda=0$.  On the other hand, since the jump condition in Riemann-Hilbert Problem~\ref{rhp:limit-simpler} is only across the unit circle, we see that for the particular solution $u(x)$ of \eqref{eq:PIII} with $\Theta_\infty=\Theta_0=0$ related to the function $\Psi(X,0)$ via \eqref{eq:u-Psi}, the Stokes constants all vanish as well, so the only monodromy data is the connection matrix.
\subsubsection{Ordinary differential equations in $T$}
Now we consider instead the compatibility condition $\mathbf{L}_T-\mathbf{B}_\Lambda + [\mathbf{L},\mathbf{B}]=\mathbf{0}$ that holds because $\mathbf{W}(\Lambda;X,T)$ is a simultaneous fundamental solution matrix of the linear problems \eqref{eq:Lax-T-ODEs} and \eqref{eq:Lax-Lambda-ODEs}.  There are now matrix coefficients for powers ranging from $\Lambda^3$ through $\Lambda^{-2}$.  The terms proportional to $\Lambda^3$ read $[\mathbf{L}^{[1]},-\ii\sigma_3]=\mathbf{0}$ which holds trivially.  Likewise, the terms proportional to $\Lambda^2$ read $[\mathbf{L}^{[0]},-\ii\sigma_3]+[\mathbf{L}^{[1]},\mathbf{B}^{[1]}]=\mathbf{0}$, which again is automatically satisfied.  The terms proportional to $\Lambda^1$ yield the equation $\mathbf{L}^{[1]}_T+2\ii\sigma_3 + [\mathbf{L}^{[-1]},-\ii\sigma_3] + [\mathbf{L}^{[0]},\mathbf{B}^{[1]}] + [\mathbf{L}^{[1]},\mathbf{B}^{[0]}]=\mathbf{0}$, which is also trivial.  The first nontrivial equations arise from the terms proportional to $\Lambda^0$, which give the equation $\mathbf{L}^{[0]}_T-\mathbf{B}^{[1]} + [\mathbf{L}^{[-2]},-\ii\sigma_3] + [\mathbf{L}^{[-1]},\mathbf{B}^{[1]}] + [\mathbf{L}^{[0]},\mathbf{B}^{[0]}]=\mathbf{0}$.  The diagonal part of this equation is trivial, but the off-diagonal part gives the equations
\begin{equation}
\begin{split}
2T\Psi_T+\Psi-2b-2\ii X\Phi&=0\\
2T\Psi_T^*+\Psi^*-2b^*+2\ii X\Phi^*&=0.
\end{split}
\label{eq:T-eqns-1}
\end{equation}
The terms proportional to $\Lambda^{-1}$ read $\mathbf{L}^{[-1]}_T + [\mathbf{L}^{[-2]},\mathbf{B}^{[1]}]+[\mathbf{L}^{[-1]},\mathbf{B}^{[0]}]=\mathbf{0}$.  The trace of this equation is trivial, and the difference of the diagonal terms gives an equation that is also implied by \eqref{eq:T-eqns-1}, but the off-diagonal terms yield new differential equations:
\begin{equation}
\begin{split}
X\Psi_T + 2(T\Phi)_T+2\ii a\Psi-\ii X|\Psi|^2\Psi&=0\\
X\Psi^*_T + 2(T\Phi^*)_T-2\ii a\Psi^*+\ii X|\Psi|^2\Psi^*&=0.
\end{split}
\label{eq:T-eqns-2}
\end{equation}
Finally, the terms proportional to $\Lambda^{-2}$ are $\mathbf{L}^{[-2]}_T + [\mathbf{L}^{[-2]},\mathbf{B}^{[0]}]=\mathbf{0}$.  The trace is trivial, but we obtain three additional equations:
\begin{equation}
\begin{split}
a_T -b\Phi^*-b^*\Phi&=0\\
b_T + 2a\Phi-\ii|\Psi|^2b&=0\\
b^*_T+2a\Phi^*+\ii|\Psi|^2b^*&=0.
\end{split}
\label{eq:T-eqns-3}
\end{equation}
The equations \eqref{eq:T-eqns-3} are consistent with the identity $a^2+|b|^2=4$.  

Using \eqref{eq:T-eqns-1} to eliminate $b$ and $b^*$, the equations \eqref{eq:T-eqns-3} become
\begin{equation}
\begin{split}
a_T-T(\Phi^*\Psi_T+\Phi\Psi_T^*)-\frac{1}{2}(\Phi^*\Psi+\Phi\Psi^*)&=0\\
(T\Psi_T)_T+\frac{1}{2}\Psi_T-\ii X\Phi_T-\ii T|\Psi|^2\Psi_T + 2a\Phi-\frac{1}{2}\ii|\Psi|^2\Psi-X|\Psi|^2\Phi&=0\\
(T\Psi^*_T)_T+\frac{1}{2}\Psi^*_T+\ii X\Phi^*_T+\ii T|\Psi|^2\Psi_T^*+2a\Phi^*+\frac{1}{2}\ii|\Psi|^2\Psi^*-X|\Psi|^2\Phi^*&=0.
\end{split}
\label{eq:T-eqns-3-no-b}
\end{equation}
The equations \eqref{eq:T-eqns-2} and \eqref{eq:T-eqns-3-no-b} constitute a closed coupled system on $a:\mathbb{R}\to\mathbb{R}$, $\Psi:\mathbb{R}\to\mathbb{C}$, and $\Phi:\mathbb{R}\to\mathbb{C}$ admitting the integral of motion $a^2+|b|^2=4$ with $b$ eliminated via \eqref{eq:T-eqns-1}.  We have not been able to identify it as a known system, but it appears to be integrable via the Lax pair \eqref{eq:Lax-T-ODEs} and \eqref{eq:Lax-Lambda-ODEs}.  This proves the following.
\begin{theorem}[Ordinary differential equations in $T$ for rogue waves of infinite order] 
The rogue wave $\Psi(X,T)=\Psi^+(X,T)$ of infinite order satisfies, for each fixed $X\in\mathbb{R}$, the ordinary differential equations \eqref{eq:T-eqns-2} and \eqref{eq:T-eqns-3-no-b} involving also the auxiliary real-valued field $a(X,T)$ and complex valued field $\Phi(X,T)$.
\label{theorem:Psi-ODEs-T}
\end{theorem}

\section{Asymptotic Properties of the Near-Field Limit}
\label{sec:Psi-asymptotic}
\subsection{Asymptotic behavior of $\Psi^\pm(X,T)$ for large $X$}
\label{sec:large-X}
We now study $\Psi^\pm(X,T)$ when $X$ is large.  To this end, we write $X=\sigma |X|$, $T=v|X|^{3/2}$, and $\Lambda=|X|^{-1/2}z$.
The phase conjugating the jump matrix for $\mathbf{R}^\pm(\Lambda;X,T)$ then takes the form
\begin{equation}
\Lambda X+\Lambda^2T\pm 2\Lambda^{-1}=|X|^{1/2}(\sigma z +vz^2\pm 2z^{-1}).
\end{equation}
It is most convenient to deduce the asymptotic behavior of $\Psi^\pm(X,T)$ in the case that the sign $\pm$ coincides with the sign $\sigma$ of $X$, i.e., we shall study $\Psi^\pm(X,T)$ in the limit $X\to\pm\infty$.  In fact, from Corollaries~\ref{cor:Psi-pm} and \ref{cor:Psi-even}, it is sufficient to consider $\Psi^+(X,T)$ as $X\to+\infty$.  We assume that $v\in\mathbb{R}$ is held fixed.
Defining $\mathbf{S}(z;X,v):=\mathbf{R}^+(X^{-1/2}z;X,X^{3/2}v)$ for $X>0$, from \eqref{eq:Psi-R} we have
\begin{equation}
\Psi^+(X,X^{3/2}v)=2\ii X^{-1/2}\lim_{z\to\infty} zS_{12}(z;X,v),\quad X>0.
\label{eq:PsiPlus-S}
\end{equation}
Clearly $\mathbf{S}(z;X,v)\to\mathbb{I}$ as $z\to\infty$ for each $X>0$, and $\mathbf{S}(z;X,v)$ is analytic in the complement of an arbitrary Jordan curve $\Gamma$ surrounding $z=0$ in the clockwise sense, across which the following jump condition holds:
\begin{equation}
\mathbf{S}_+(z;X,v)=\mathbf{S}_-(z;X,v)\ee^{-\ii X^{1/2}\vartheta(z;v)\sigma_3}\mathbf{Q}^{-1}\ee^{\ii X^{1/2}\vartheta(z;v)\sigma_3},\quad z\in\Gamma,\quad\vartheta(z;v):=z+vz^2+2z^{-1}.
\label{eq:S-jump}
\end{equation}
\subsubsection{Exponent analysis and steepest descent}
Given $v\in\mathbb{R}$, the critical points of $\vartheta(z;v)$ are the roots of a real cubic.  The critical points are all real for $|v|$ sufficiently small, but a conjugate pair appears if $|v|$ becomes too large.  The threshold value of $|v|$ is obtained from the cubic discriminant:  $|v|<54^{-1/2}$ is necessary and sufficient for the existence of three real critical points of $\vartheta(z;v)$.  Subject to this inequality on $v$, there exists a component of the level curve $\mathrm{Im}(\vartheta(z;v))=0$ that is a Jordan curve enclosing the origin in the $z$-plane, and that passes through two of the three real critical points, with the remaining critical point in the exterior domain.  See Figure~\ref{fig:v-plots}.
\begin{figure}[h]
\begin{center}
\includegraphics{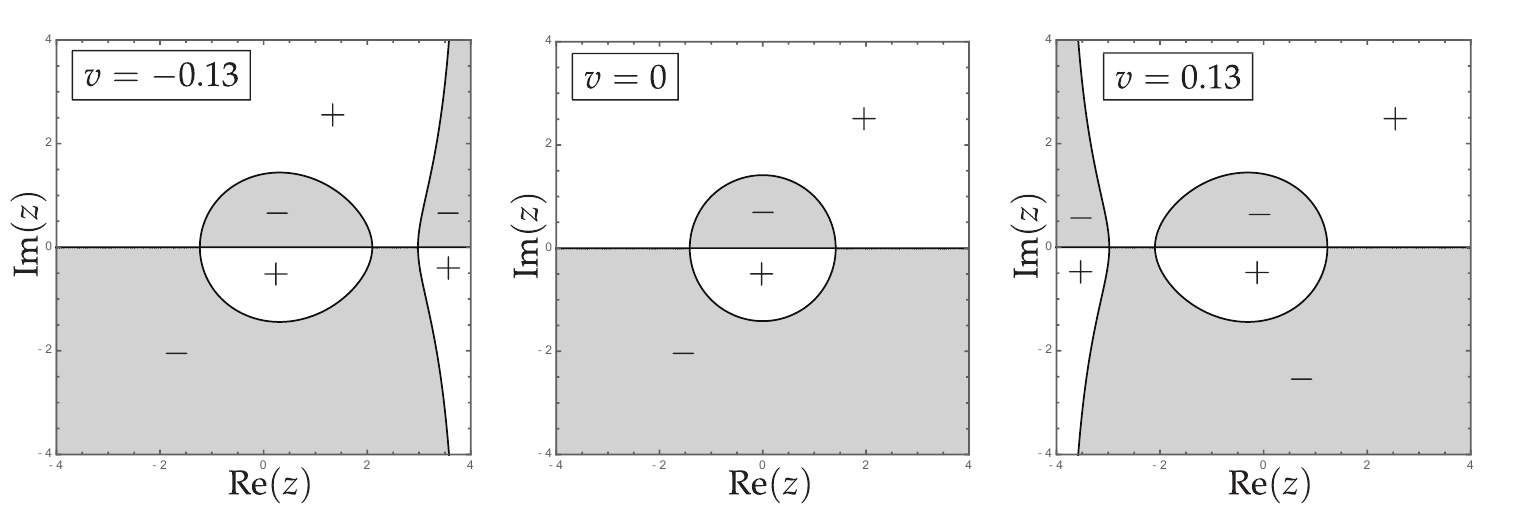}
\end{center}
\caption{Sign charts for $\mathrm{Im}(\vartheta(z;v))$ as $v$ varies over the interval $|v|<54^{-1/2}\approx 0.1361$.}
\label{fig:v-plots}
\end{figure}
We select this curve as the jump contour $\Gamma$ for $\mathbf{S}$ and denote the two real critical points of $\vartheta(z;v)$ through which it passes as $a<b$ where $a=a(v)$ and $b=b(v)$.  The real axis divides $\Gamma$ into an arc $\Gamma^+$ in the upper half-plane and an arc $\Gamma^-$ in the lower half-plane.  We introduce thin lens-shaped domains $L^\pm$ and $R^\pm$ on the left and right sides respectively of $\Gamma^\pm$ whose outer boundary arcs $C_L^\pm$ and $C_R^\pm$ meet the real axis at $45^\circ$ angles as shown in the left-hand panel of Figure~\ref{fig:regions-Tjump}, and along each of which $\mathrm{Im}(\vartheta(z;v))$ has a definite sign.
\begin{figure}[h]
\begin{center}
\includegraphics{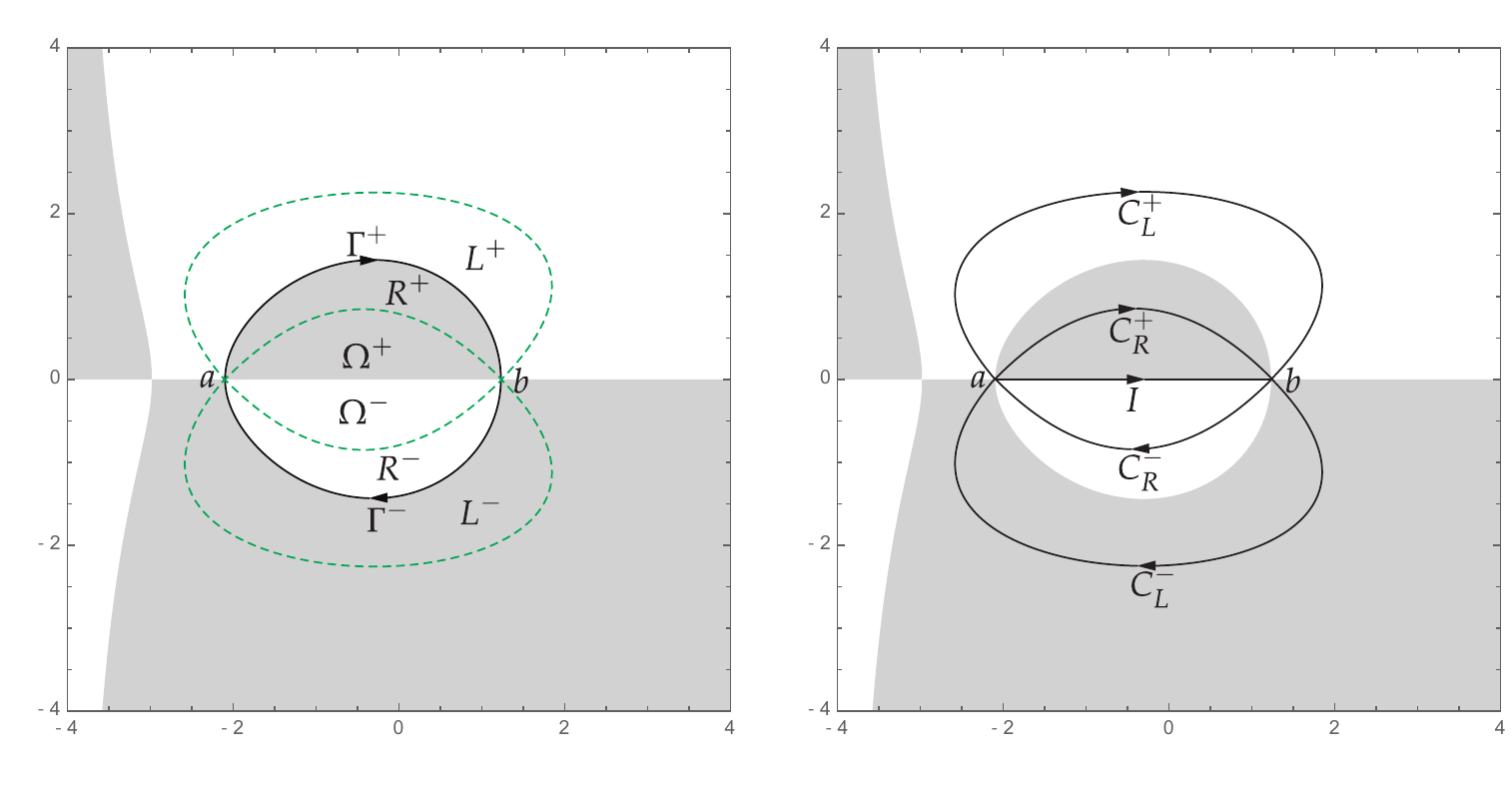}
\end{center}
\caption{Left:  the jump contour $\Gamma=\Gamma^+\cup\Gamma^-$ for $\mathbf{S}$ and the regions $L^\pm$, $R^\pm$, and $\Omega^\pm$.  Right:  the jump contour for $\mathbf{T}$.}
\label{fig:regions-Tjump}
\end{figure}
The region between $C_R^\pm$ and the real axis is denoted $\Omega^\pm$.  We separate the exponential factors $\ee^{\pm 2\ii X^{1/2}\vartheta(z;v)}$ appearing in the jump condition \eqref{eq:S-jump} by the following substitutions:
\begin{equation}
\mathbf{T}(z;X,v):=
\mathbf{S}(z;X,v)\begin{bmatrix}1 & 0\\\ee^{2\ii X^{1/2}\vartheta(z;v)} & 1\end{bmatrix},\quad z\in L^+,
\end{equation}
\begin{equation}
\mathbf{T}(z;X,v):=
\mathbf{S}(z;X,v)2^{\sigma_3/2}\begin{bmatrix}1 & \tfrac{1}{2}\ee^{-2\ii X^{1/2}\vartheta(z;v)}\\0 & 1\end{bmatrix},\quad  z\in R^+,
\end{equation}
\begin{equation}
\mathbf{T}(z;X,v):=
\mathbf{S}(z;X,v)2^{\sigma_3/2},\quad  z\in\Omega^+,
\end{equation}
\begin{equation}
\mathbf{T}(z;X,v):=
\mathbf{S}(z;X,v)2^{-\sigma_3/2},\quad  z\in\Omega^-,
\end{equation}
\begin{equation}
\mathbf{T}(z;X,v):=
\mathbf{S}(z;X,v)2^{-\sigma_3/2}\begin{bmatrix}1 & 0\\-\tfrac{1}{2}\ee^{2\ii X^{1/2}\vartheta(z;v)} & 1\end{bmatrix},\quad  z\in R^-,
\end{equation}
\begin{equation}
\mathbf{T}(z;X,v):=
\mathbf{S}(z;X,v)\begin{bmatrix}1 & -\ee^{-2\ii X^{1/2}\vartheta(z;v)}\\0 & 1\end{bmatrix},\quad  z\in L^-,
\end{equation}
and in the complementary domain exterior to the Jordan curve $C_L^+\cup C_L^-$ we simply take $\mathbf{T}(z;X,v)=\mathbf{S}(z;X,v)$.  One then can check easily that $\mathbf{T}(z;X,v)$ takes equal boundary values from each side on the two arcs of $\Gamma$, so $\mathbf{T}(z;X,v)$ can be considered to be a well-defined analytic function on $\Gamma^+$ and $\Gamma^-$.  The jump contour for $\mathbf{T}(z;X,v)$ is illustrated in the right-hand panel of Figure~\ref{fig:regions-Tjump}.  On the five arcs of the jump contour with the indicated orientation, the jump conditions satisfied by $\mathbf{T}(z;X,v)$ are the following.
\begin{equation}
\mathbf{T}_+(z;X,v)=\mathbf{T}_-(z;X,v)\begin{bmatrix}1 & 0\\-\ee^{2\ii X^{1/2}\vartheta(z;v)} & 1\end{bmatrix},\quad z\in C_L^+,
\label{eq:Tjump-1}
\end{equation}
\begin{equation}
\mathbf{T}_+(z;X,v)=\mathbf{T}_-(z;X,v)\begin{bmatrix}1 & \tfrac{1}{2}\ee^{-2\ii X^{1/2}\vartheta(z;v)}\\0 & 1\end{bmatrix},\quad z\in C_R^+,
\label{eq:Tjump-2}
\end{equation}
\begin{equation}
\mathbf{T}_+(z;X,v)=\mathbf{T}_-(z;X,v)2^{\sigma_3},\quad z\in I,
\label{eq:T-diagonal}
\end{equation}
\begin{equation}
\mathbf{T}_+(z;X,v)=\mathbf{T}_-(z;X,v)\begin{bmatrix}1 & 0\\-\tfrac{1}{2}\ee^{2\ii X^{1/2}\vartheta(z;v)} & 1\end{bmatrix},\quad z\in C_R^-,
\label{eq:Tjump-4}
\end{equation}
and
\begin{equation}
\mathbf{T}_+(z;X,v)=\mathbf{T}_-(z;X,v)\begin{bmatrix}1 & \ee^{-2\ii X^{1/2}\vartheta(z;v)}\\0 & 1\end{bmatrix},\quad z\in C_L^-.
\label{eq:Tjump-5}
\end{equation}
Since $\mathrm{Im}(\vartheta(z;v))>0$ holds on $C_L^+$ and $C_R^-$ while $\mathrm{Im}(\vartheta(z;v))<0$ holds on $C_L^-$ and $C_R^+$, the jump matrices on these four contour arcs are exponentially small (as $X\to +\infty$) perturbations of the identity uniformly except near the endpoints $a$ and $b$.  

\subsubsection{Parametrix construction}
To deal with the jump condition on $I$ as well as the non-uniformity of the exponential decay near $a$ and $b$, we construct a parametrix for $\mathbf{T}(z;X,v)$.  We first define an \emph{outer parametrix} $\dot{\mathbf{T}}^\mathrm{out}(z,v)$ for $z\in\mathbb{C}\setminus I$ by the formula
\begin{equation}
\dot{\mathbf{T}}^\mathrm{out}(z,v):=\left(\frac{z-a(v)}{z-b(v)}\right)^{\ii p\sigma_3},\quad p:=\frac{\ln(2)}{2\pi}>0,\quad z\in\mathbb{C}\setminus I.
\label{eq:T-out}
\end{equation}
Here, the powers $\pm \ii p$ refer to the principal branch, i.e., $w^{\pm \ii p}:=\ee^{\pm \ii p\log(w)}$ where $-\pi<\mathrm{Im}(\log(w))<\pi$; since the locus where $(z-b)/(z-a)$ is negative real coincides precisely with the interval $I$ this gives the indicated domain of analyticity.  Obviously $\mathbf{T}^{\mathrm{out}}(z;v)\to\mathbb{I}$ as $z\to\infty$.  Also, the jump condition 
\begin{equation}
\dot{\mathbf{T}}^\mathrm{out}_+(z;v)=\dot{\mathbf{T}}^\mathrm{out}_-(z;v)2^{\sigma_3},\quad z\in I
\end{equation}
clearly holds (compare with \eqref{eq:T-diagonal}).  

Next, we define \emph{inner parametrices} by finding local matrix functions defined near $z=a,b$ that exactly satisfy the jump conditions and also match well with the outer parametrix at some small distance independent of $X$ from these points.  Noting that while $\vartheta'(a(v);v)=\vartheta'(b(v),v)=0$, for $|v|<54^{-1/2}$ we have $\vartheta''(a(v);v)<0$ and $\vartheta''(b(v);v)>0$, we define conformal mappings $f_a(z;v)$ and $f_b(z;v)$ locally near $z=a$ and $z=b$ respectively by the equations
\begin{equation}
f_a(z;v)^2=2\left(\vartheta(a(v);v)-\vartheta(z;v)\right)\quad\text{and}\quad f_b(z;v)^2=2\left(\vartheta(z;v)-\vartheta(b(v);v)\right)
\label{eq:conformal}
\end{equation}
and we choose the solutions for which $f_a'(a(v);v)<0$ and $f_b(b(v);v)>0$.  Let $\zeta_a:=X^{1/4}f_a$ and $X^{1/4}f_b$ denote rescaled versions of these conformal coordinates.  The jump conditions satisfied by 
\begin{equation}
\mathbf{U}^a:=\mathbf{T}\ee^{-\ii X^{1/2}\vartheta(a;v)\sigma_3}(\ii\sigma_2),\quad\text{near $z=a$}
\end{equation}
and by 
\begin{equation}
\mathbf{U}^b:=\mathbf{T}\ee^{-\ii X^{1/2}\vartheta(b;v)\sigma_3},\quad\text{near $z=b$}
\end{equation}
then take exactly the same form when expressed in terms of the respective variables $\zeta=\zeta_a$ and $\zeta=\zeta_b$ and the jump contours are locally taken to coincide with the five rays $\arg(\zeta)=\pm\pi/4$, $\arg(\zeta)=\pm 3\pi/4$, and $\arg(-\zeta)=0$.  See Figure~\ref{fig:PC-jumps}.
\begin{figure}[h]
\begin{center}
\includegraphics{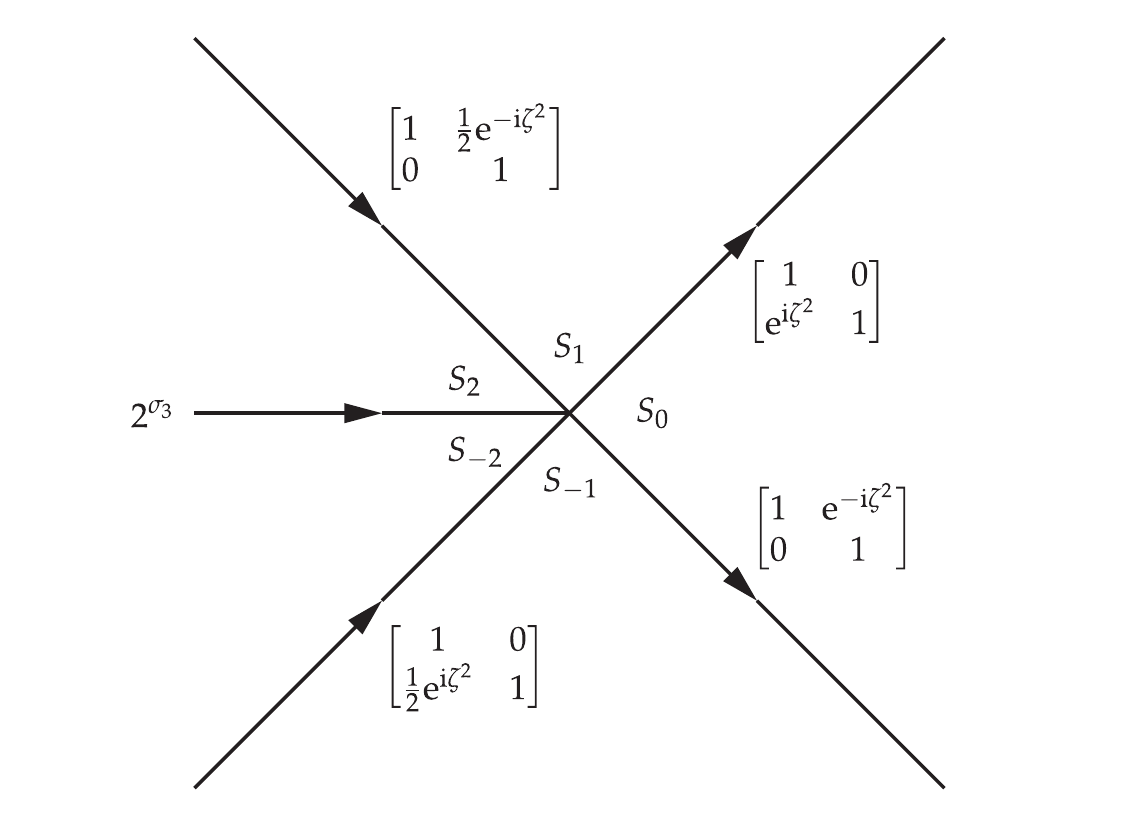}
\end{center}
\caption{The jump conditions satisfied by $\mathbf{U}=\mathbf{U}^a$ near $z=a$ and by $\mathbf{U}=\mathbf{U}^b$ near $z=b$ take exactly the same form when written in terms of the scaled conformal coordinates $\zeta_a$ and $\zeta_b$ respectively, namely $\mathbf{U}_+=\mathbf{U}_-\mathbf{V}^\mathrm{PC}$ where $\mathbf{V}^\mathrm{PC}$ is defined on five rays in the $\zeta$-plane as indicated.}
\label{fig:PC-jumps}
\end{figure}
The jump matrix in Figure~\ref{fig:PC-jumps} corresponds to a special case of the standard parabolic cylinder parametrix typically occurring in the Deift-Zhou steepest descent method \cite{DeiftZ93} for phase functions with simple critical points as is the case here.  The outer parametrix can also be expressed near $z=a$ or $z=b$ in terms of the relevant conformal coordinate:
\begin{multline}
\dot{\mathbf{T}}^\mathrm{out}(z;v)\ee^{-\ii X^{1/2}\vartheta(a;v)\sigma_3}(\ii\sigma_2)= 
X^{-\ii p\sigma_3/4}\ee^{-\ii X^{1/2}\vartheta(a;v)\sigma_3}\mathbf{H}^a(z;v)\zeta_a^{-\ii p\sigma_3},\\
\mathbf{H}^a(z;v):=(b-z)^{-\ii p\sigma_3}\left(\frac{a-z}{f_a(z;v)}\right)^{\ii p\sigma_3}(\ii\sigma_2),
\label{eq:Ha}
\end{multline}
and
\begin{multline}
\dot{\mathbf{T}}^\mathrm{out}(z;v)\ee^{-\ii X^{1/2}\vartheta(b;v)\sigma_3}=
X^{\ii p\sigma_3/4}\ee^{-\ii X^{1/2}\vartheta(b;v)\sigma_3}\mathbf{H}^b(z;v)\zeta_b^{-\ii p\sigma_3},\\
\mathbf{H}^b(z;v):=(z-a)^{\ii p\sigma_3}\left(\frac{f_b(z;v)}{z-b}\right)^{\ii p\sigma_3}.
\label{eq:Hb}
\end{multline}
Once again, all power functions in these formulae are defined as principal branches, so it is easy to confirm that $\mathbf{H}^a(z;v)$ and $\mathbf{H}^b(z;v)$ are analytic matrix-valued functions of $z$ in neighborhoods of $z=a$ and $z=b$ respectively.  Taking into account the last factor on the right in these expressions, $\zeta^{-\ii p\sigma_3}$, we now properly define a matrix $\mathbf{U}(\zeta)$ as the solution of the following Riemann-Hilbert problem.  
\begin{rhp}[Parabolic cylinder parametrix]
Seek a $2\times 2$ matrix-valued function $\mathbf{U}(\zeta)$ with the following properties.
\begin{itemize}
\item[]\textbf{Analyticity:}  $\mathbf{U}(\zeta)$ is analytic for $\zeta$ in the five sectors shown in Figure~\ref{fig:PC-jumps}, namely $S_0$:  $|\arg(\zeta)|<\tfrac{1}{4}\pi$, $S_1$:  $\tfrac{1}{4}\pi<\arg(\zeta)<\tfrac{3}{4}\pi$, $S_{-1}$:  $-\tfrac{3}{4}\pi<\arg(\zeta)<-\tfrac{1}{4}\pi$, $S_2$:  $\tfrac{3}{4}\pi<\arg(\zeta)<\pi$, and $S_{-2}$:  $-\pi<\arg(\zeta)<-\tfrac{3}{4}\pi$.  It takes continuous boundary values on the excluded rays and at the origin from each sector.
\item[]\textbf{Jump conditions:}  $\mathbf{U}_+(\zeta)=\mathbf{U}_-(\zeta)\mathbf{V}^\mathrm{PC}(\zeta)$, where $\mathbf{V}^\mathrm{PC}(\zeta)$ is the matrix function defined on the jump contour shown in Figure~\ref{fig:PC-jumps}.
\item[]\textbf{Normalization:}  $\mathbf{U}(\zeta)\zeta^{\ii p\sigma_3}\to\mathbb{I}$ as $\zeta\to\infty$ uniformly in all directions, where $p=\ln(2)/(2\pi)$.
\end{itemize}
\label{rhp:PC}
\end{rhp}
The solution of this problem can be expressed explicitly in terms of the parabolic cylinder function $U(\cdot,\cdot)$ as defined in \cite[Ch.\@ 12]{DLMF}, but we will not require any details of these formul\ae.  The solution has the following important properties.  The diagonal (resp., off-diagonal) part of $\mathbf{U}(\zeta)\zeta^{\ii p\sigma_3}$ has a complete asymptotic expansion in descending even (resp., odd) integer powers of $\zeta$ as $\zeta\to\infty$, with all coefficients being independent of the sector in which $\zeta\to\infty$.  In particular, the solution satisfies
\begin{equation}
\mathbf{U}(\zeta)\zeta^{\ii p\sigma_3}=\mathbb{I} +\frac{1}{2\ii\zeta}\begin{bmatrix}0 & \alpha\\-\beta & 0\end{bmatrix} + O(\zeta^{-2}),\quad\zeta\to\infty,
\label{eq:U-matrix-asymp}
\end{equation}
where
\begin{equation}
\alpha:=2^{3/4}\sqrt{2\pi}\Gamma\left(\frac{\ii\ln(2)}{2\pi}\right)^{-1}\ee^{\ii\pi/4}\ee^{\ii(\ln(2))^2/(2\pi)}\quad\text{and}\quad
\beta:=-\alpha^*.
\label{eq:alpha-beta-special}
\end{equation}

From $\mathbf{U}(\zeta)$ we define the inner parametrices near $z=a,b$ as follows.  Let $D_z(\delta)$ denote the disk with center $z$ and radius $\delta$.  Then for $\delta$ sufficiently small given $v$ but independent of $X$, we define
\begin{equation}
\dot{\mathbf{T}}^a(z;X,v):=X^{-\ii p\sigma_3/4}\ee^{-\ii X^{1/2}\vartheta(a;v)\sigma_3}\mathbf{H}^a(z;v)
\mathbf{U}(X^{1/4}f_a(z;v))(-\ii\sigma_2)\ee^{\ii X^{1/2}\vartheta(a;v)\sigma_3},\quad z\in D_a(\delta),
\end{equation}
from which it follows that 
\begin{multline}
\dot{\mathbf{T}}^a(z;X,v)\dot{\mathbf{T}}^\mathrm{out}(z;v)^{-1}=X^{-\ii p\sigma_3/4}\ee^{-\ii X^{1/2}\vartheta(a;v)\sigma_3}\mathbf{H}^a(z;v)\mathbf{U}(\zeta_a)\zeta_a^{\ii p\sigma_3}\mathbf{H}^a(z;v)^{-1}\ee^{\ii X^{1/2}\vartheta(a;v)\sigma_3}X^{\ii p\sigma_3/4},\\ \zeta_a=X^{1/4}f_a(z;v),\quad z\in\partial D_a(\delta),
\label{eq:a-mismatch}
\end{multline}
and
\begin{equation}
\dot{\mathbf{T}}^b(z;X,v):=X^{\ii p\sigma_3/4}\ee^{-\ii X^{1/2}\vartheta(b;v)\sigma_3}\mathbf{H}^b(z;v)
\mathbf{U}(X^{1/4}f_b(z;v))\ee^{\ii X^{1/2}\vartheta(b;v)\sigma_3/4},\quad z\in D_b(\delta),
\end{equation}
from which it follows that
\begin{multline}
\dot{\mathbf{T}}^b(z;X,v)\dot{\mathbf{T}}^\mathrm{out}(z;v)^{-1}=X^{\ii p\sigma_3/4}\ee^{-\ii X^{1/2}\vartheta(b;v)\sigma_3}\mathbf{H}^b(z;v)\mathbf{U}(\zeta_b)\zeta_b^{\ii p\sigma_3}\mathbf{H}^b(z;v)^{-1}\ee^{\ii X^{1/2}\vartheta(b;v)\sigma_3}X^{-\ii p\sigma_3/4},\\ \zeta_b=X^{1/4}f_b(z;v),\quad z\in\partial D_b(\delta).
\label{eq:b-mismatch}
\end{multline}

The \emph{global parametrix} for $\mathbf{T}(z;X,v)$ is defined when $|v|<54^{-1/2}$ as follows:
\begin{equation}
\dot{\mathbf{T}}(z;X,v):=\begin{cases}\dot{\mathbf{T}}^a(z;X,v),&\quad z\in D_a(\delta)\\
\dot{\mathbf{T}}^b(z;X,v),&\quad z\in D_b(\delta)\\
\dot{\mathbf{T}}^\mathrm{out}(z;v),&\quad z\in\mathbb{C}\setminus(I\cup \overline{D_a(\delta)}\cup\overline{D_b(\delta)}).
\end{cases}
\label{eq:global-parametrix-large-X}
\end{equation}
Note that $\det(\dot{\mathbf{T}}(z;X,v))=1$.

\subsubsection{Error analysis}
The \emph{error} in approximating $\mathbf{T}$ with its parametrix $\dot{\mathbf{T}}$ is defined by
\begin{equation}
\mathbf{F}(z;X,v):=\mathbf{T}(z;X,v)\dot{\mathbf{T}}(z;X,v)^{-1}
\label{eq:F-def}
\end{equation}
wherever both factors are defined.  The domain of analyticity of $\mathbf{F}(z;X,v)$ is $\mathbb{C}\setminus\Sigma_\mathbf{F}$, where the contour $\Sigma_\mathbf{F}$ consists of (i) the oriented arcs of $C_\mathrm{L}^\pm$ and $C_\mathrm{R}^\pm$ lying in the exterior of $D_a(\delta)$ and $D_b(\delta)$ and (ii) the circular boundaries $\partial D_a(\delta)$ and $\partial D_b(\delta)$ which we take to have clockwise orientation.  The interval $I=[a,b]$ is not part of the jump contour $\Sigma_\mathbf{F}$ because $\mathbf{T}$ and $\dot{\mathbf{T}}$ satisfy exactly the same jump condition across $I$.  Likewise, $\mathbf{F}$ is analytic within the disks $D_a(\delta)$ and $D_b(\delta)$ because the inner parametrices $\dot{\mathbf{T}}^a$ and $\dot{\mathbf{T}}^b$ are exact local solutions of the Riemann-Hilbert jump conditions for $\mathbf{T}$.  Across any arc of $\Sigma_\mathbf{F}$, the jump of $\mathbf{F}$ can be expressed in the form $\mathbf{F}_+=\mathbf{F}_-\mathbf{V}^\mathbf{F}$.  For $z$ in the arcs of $C_\mathrm{L}^\pm$ or $C_\mathrm{R}^\pm$ contained in $\Sigma_\mathbf{F}$, it is convenient to use the fact that $\dot{\mathbf{T}}=\dot{\mathbf{T}}^\mathrm{out}$ is analytic on such arcs to express the jump matrix $\mathbf{V}^\mathbf{F}$ in the form
\begin{equation}
\begin{split}
\mathbf{V}^\mathbf{F}(z;X,v):=&\mathbf{F}_-(z;X,v)^{-1}\mathbf{F}_+(z;X,v)\\
=&\dot{\mathbf{T}}^\mathrm{out}(z;v)\mathbf{T}_-(z;X,v)^{-1}\mathbf{T}_+(z;X,v)\dot{\mathbf{T}}^\mathrm{out}(z;v)^{-1},\quad z\in(C_\mathrm{L}^\pm\cup C_\mathrm{R}^\pm)\cap\Sigma_\mathbf{F},
\end{split}
\end{equation}
where the central two factors are defined in \eqref{eq:Tjump-1}--\eqref{eq:Tjump-2} and \eqref{eq:Tjump-4}--\eqref{eq:Tjump-5}.  Because the exponential factors appearing in the latter jump conditions are restricted to the exterior of the disks $D_a(\delta)$ and $D_b(\delta)$, and since $\dot{\mathbf{T}}^\mathrm{out}(z;v)$ is independent of $X$, there is a positive constant $K(v)>0$ such that
\begin{equation}
\sup_{z\in(C_\mathrm{L}^\pm\cup C_\mathrm{R}^\pm)\cap\Sigma_\mathbf{F}}\|\mathbf{V}^\mathbf{F}(z;X,v)-\mathbb{I}\|=O(\ee^{-X^{1/2}K(v)}),\quad X\to+\infty,
\label{eq:VF-exponential-bound}
\end{equation}
where $\|\cdot\|$ denotes the matrix norm induced from an arbitrary norm on $\mathbb{C}^2$.  On the other hand, for $z\in\partial D_{a,b}(\delta)$, we use the fact that $\mathbf{T}(z;X,v)$ is analytic at all but finitely-many points of the circle while $\dot{\mathbf{T}}_+=\dot{\mathbf{T}}^\mathrm{out}$ and $\dot{\mathbf{T}}_-=\dot{\mathbf{T}}^{a,b}$ to obtain
\begin{equation}
\mathbf{V}^\mathbf{F}(z;X,v)=\dot{\mathbf{T}}^{a,b}(z;X,v)\dot{\mathbf{T}}^\mathrm{out}(z;v)^{-1},\quad z\in\partial D_{a,b}(\delta)\subset\Sigma_\mathbf{F}.
\label{eq:VF-circles}
\end{equation}
The right-hand side is given explicitly by \eqref{eq:a-mismatch} and \eqref{eq:b-mismatch}.  Since $\zeta_{a,b}$ is proportional to $X^{1/4}$ when $z\in\partial D_{a,b}(\delta)$ while the conjugating factors in \eqref{eq:a-mismatch} and \eqref{eq:b-mismatch} are bounded on $\partial D_{a,b}(\delta)$ as $X\to +\infty$, it follows from \eqref{eq:U-matrix-asymp} that 
\begin{equation}
\sup_{z\in\partial D_{a,b}(\delta)}\|\mathbf{V}^\mathbf{F}(z;X,v)-\mathbb{I}\|=O(X^{-1/4}),\quad X\to +\infty.
\label{eq:VF-circles-bound}
\end{equation}

To study $\mathbf{F}(z;X,v)$ we reformulate the jump condition in the form $\mathbf{F}_+-\mathbf{F}_-=\mathbf{F}_-(\mathbf{V}^\mathbf{F}-\mathbb{I})$ and use the fact that both factors in the definition \eqref{eq:F-def} of $\mathbf{F}$ tend to the identity as $z\to\infty$ to obtain from the Plemelj formula
\begin{equation}
\mathbf{F}(z;X,v)=\mathbb{I}+\frac{1}{2\pi\ii}\int_{\Sigma_\mathbf{F}}\frac{\mathbf{F}_-(w;X,v)(\mathbf{V}^\mathbf{F}(w;X,v)-\mathbb{I})}{w-z}\,\dd w,\quad z\in\mathbb{C}\setminus\Sigma_\mathbf{F}.
\label{eq:F-F-minus}
\end{equation}
Letting $z$ tend to a point on an arc of $\Sigma_\mathbf{F}$ from the right side by orientation leads to a closed integral equation for the boundary value $\mathbf{F}_-(z;X,v)$ defined on $\Sigma_\mathbf{F}$ away from self-intersection points:
\begin{equation}
\mathbf{F}_-(z;X,v)=\mathbb{I}+\mathcal{C}^{\Sigma_\mathbf{F}}_-(\mathbf{F}_-(\cdot;X,v)(\mathbf{V}^\mathbf{F}(\cdot;X,v)-\mathbb{I}))(z),\quad z\in\Sigma_\mathbf{F},
\label{eq:F-minus-integral-equation}
\end{equation}
where $\mathcal{C}^{\Sigma_\mathbf{F}}_-(f)$ is the Cauchy projection defined by
\begin{equation}
\mathcal{C}^{\Sigma_\mathbf{F}}_-(f)(z):=\frac{1}{2\pi\ii}\int_{\Sigma_\mathbf{F}}\frac{f(w)\,\dd w}{w-z_-},\quad z\in\Sigma_\mathbf{F}.
\end{equation}
It is now a well-known fact that for a contour such as $\Sigma_\mathbf{F}$ being a finite union of Lipschitz arcs with non-tangential intersections, $\mathcal{C}_-^{\Sigma_\mathbf{F}}$ is a bounded operator $L^2(\Sigma_\mathbf{F})\to L^2(\Sigma_\mathbf{F})$ with respect to arc-length measure.  Its operator norm depends on the contour and hence in our setting on $v$ but not on $X$.  The estimates   
\eqref{eq:VF-exponential-bound} and \eqref{eq:VF-circles-bound} then imply that the integral equation \eqref{eq:F-minus-integral-equation} is uniquely solvable by iteration or Neumann series on $L^2(\Sigma_\mathbf{F})$ for sufficiently large $X>0$, and its solution satisfies
\begin{equation}
\mathbf{F}_-(\cdot;X,v)-\mathbb{I}=O(X^{-1/4}),\quad X\to +\infty
\label{eq:F-minus-L2}
\end{equation}
in the $L^2(\Sigma_\mathbf{F})$ sense.  Note that since $\Sigma_\mathbf{F}$ is a compact contour, we may identify the identity matrix $\mathbb{I}$ with the associated constant function in $L^2(\Sigma_\mathbf{F})$.  Now from \eqref{eq:F-F-minus} we easily obtain the Laurent expansion of $\mathbf{F}(z;X,v)$ convergent for sufficiently large $|z|$:
\begin{equation}
\mathbf{F}(z;X,v)=\mathbb{I} - \frac{1}{2\pi\ii}\sum_{k=1}^\infty z^{-p}\int_{\Sigma_\mathbf{F}}
\mathbf{F}_-(w;X,v)(\mathbf{V}^\mathbf{F}(w;X,v)-\mathbb{I})w^{p-1}\,\dd w,\quad |z|>|\Sigma_\mathbf{F}|:=\sup_{w\in\Sigma_\mathbf{F}}|w|.
\label{eq:F-Laurent-series}
\end{equation}
Now recall \eqref{eq:PsiPlus-S} and the fact that $\mathbf{S}(z;X,v)=\mathbf{T}(z;X,v)=\mathbf{F}(z;X,v)\dot{\mathbf{T}}^\mathrm{out}(z;v)$ holds for $|z|$ sufficiently large; therefore since $\dot{\mathbf{T}}^\mathrm{out}(z;v)$ is a diagonal matrix tending to the identity as $z\to\infty$,
\begin{equation}
\Psi^+(X,X^{3/2}v)=2\ii X^{-1/2}\lim_{z\to\infty}zF_{12}(z;X,v).
\end{equation}
Now using \eqref{eq:F-Laurent-series}, we obtain an expression in terms of the solution of the integral equation \eqref{eq:F-minus-integral-equation}:
\begin{equation}
\Psi^+(X,X^{3/2}v)=-\frac{1}{\pi X^{1/2}}\left[\int_{\Sigma_\mathbf{F}}F_{11-}(w;X,v)V_{12}^\mathbf{F}(w;X,v)\,\dd w
+\int_{\Sigma_\mathbf{F}}F_{12-}(w;X,v)(V_{22}^\mathbf{F}(w;X,v)-1)\,\dd w\right].
\end{equation}
Since on the compact contour $\Sigma_\mathbf{F}$, the $L^1(\Sigma_\mathbf{F})$ norm is subordinate to the $L^2(\Sigma_\mathbf{F})$ norm, combining the $L^\infty(\Sigma_\mathbf{F})$ estimates \eqref{eq:VF-exponential-bound} and \eqref{eq:VF-circles-bound} with the $L^2(\Sigma_\mathbf{F})$ estimate \eqref{eq:F-minus-L2}, we get
\begin{equation}
\Psi^+(X,X^{3/2}v)=-\frac{1}{\pi X^{1/2}}\int_{\Sigma_\mathbf{F}}V_{12}^\mathbf{F}(w;X,v)\,\dd w + O(X^{-1}),\quad X\to +\infty
\label{eq:Psi-Plus-Integral}
\end{equation}
uniformly for $|v|\le 54^{-1/2}-\epsilon$.  Due to the exponential estimate \eqref{eq:VF-exponential-bound} the same formula holds true (with a different implicit constant in the error term) if the integration is taken just over the circles $\partial D_{a,b}(\delta)$.  Furthermore, using \eqref{eq:a-mismatch} and \eqref{eq:b-mismatch} with \eqref{eq:U-matrix-asymp} in \eqref{eq:VF-circles}
shows that as $X\to +\infty$,
\begin{equation}
V^\mathbf{F}_{12}(z;X,v)=\frac{X^{-\ii p/2}\ee^{-2\ii X^{1/2}\vartheta(a;v)}}{2\ii X^{1/4}f_a(z;v)}\left(\alpha H_{11}^a(z;v)^2 + \beta H_{12}^a(z;v)^2\right) + O(X^{-1/2}),\quad z\in\partial D_a(\delta)
\end{equation}
and
\begin{equation}
V^\mathbf{F}_{12}(z;X,v)=\frac{X^{\ii p/2}\ee^{-2\ii X^{1/2}\vartheta(b;v)}}{2\ii X^{1/4}f_b(z;v)}\left(\alpha H_{11}^b(z;v)^2+\beta H_{12}^b(z;v)^2\right) + O(X^{-1/2}),\quad z\in\partial D_b(\delta)
\end{equation}
with both error estimates being uniform on the indicated circles.  The integrals of the explicit leading terms over the respective circles can then be evaluated by residues at $z=a,b$, since $f_{a,b}(z;v)$ has a simple zero at $z=a,b$, while the elements of $\mathbf{H}^{a,b}(z;v)$ are analytic in $D_{a,b}(\delta)$.
Therefore,
\begin{multline}
\Psi^+(X,X^{3/2}v)=X^{-3/4}\left[X^{-\ii p/2}\ee^{-2\ii X^{1/2}\vartheta(a;v)}\frac{\alpha H_{11}^a(a;v)^2+\beta H_{12}^a(a;v)^2}{f'_a(a;v)}\right.\\
\left.{} + X^{\ii p/2}\ee^{- 2\ii X^{1/2}\vartheta(b;v)}
\frac{\alpha H_{11}^b(b;v)^2 + \beta H_{12}^b(b;v)^2}{f'_b(b;v)}\right]+ O(X^{-1}),\quad X\to +\infty.
\end{multline}

It remains to calculate $H_{11}^a(a;v)$, $H_{12}^a(a;v)$, $f'_a(a;v)$, $H_{11}^b(b;v)$, $H_{12}^b(b;v)$, and $f'_b(b;v)$.  Firstly, from \eqref{eq:conformal}, 
\begin{equation}
f_a'(a;v)=-\sqrt{-\vartheta''(a;v)}\quad\text{and}\quad f_b'(b;v)=\sqrt{\vartheta''(b;v)}.
\end{equation}
Then, using \eqref{eq:Ha} and \eqref{eq:Hb} and l'H\^opital's rule,
\begin{equation}
\mathbf{H}^a(a;v)=(b-a)^{-\ii p\sigma_3}\left(\frac{-1}{f_a'(a;v)}\right)^{\ii p\sigma_3}(\ii\sigma_2)
\quad\text{and}\quad
\mathbf{H}^b(b;v)=(b-a)^{\ii p\sigma_3}\left(f_b'(b;v)\right)^{\ii p\sigma_3}.
\end{equation}
Therefore,
\begin{equation}
\begin{split}
\frac{\alpha H_{11}^a(a;v)^2+\beta H_{12}^a(a;v)^2}{f_a'(a;v)}&=-(b-a)^{-2\ii p}(-\vartheta''(a;v))^{-\ii p}\frac{\beta}{\sqrt{-\vartheta''(a;v)}}\\
\frac{\alpha H_{11}^b(b;v)^2+\beta H_{12}^b(b;v)^2}{f_b'(b;v)}&=(b-a)^{2\ii p}\vartheta''(b;v)^{\ii p}\frac{\alpha}{\sqrt{\vartheta''(b;v)}}.
\end{split}
\end{equation}
Finally, since $\beta=-\alpha^*$ and using \cite[Eq.\@ 5.4.3]{DLMF} we have $|\alpha|=\sqrt{2p}$,
we obtain the following result.
\begin{theorem}[Large-$X$ asymptotics of rogue waves of infinite order]
Let $v\in\mathbb{R}$ be fixed with $|v|<54^{-1/2}$, and let $\vartheta(z;v):=z+vz^2+2z^{-1}$.  Then $\vartheta(\cdot;v)$ has three real simple critical points, and 
\begin{multline}
\Psi^+(X,X^{3/2}v)=\frac{\sqrt{2p}}{X^{3/4}}\left(\frac{\ee^{-2\ii X^{1/2}\vartheta(a;v)}(-\vartheta''(a;v))^{-\ii p}}{\sqrt{-\vartheta''(a;v)}}\ee^{\ii\phi(X,v)} + 
\frac{\ee^{-2\ii X^{1/2}\vartheta(b;v)}\vartheta''(b;v)^{\ii p}}{\sqrt{\vartheta''(b;v)}}\ee^{-\ii\phi(X,v)}\right)\\
+O(X^{-1}),\quad X\to +\infty,
\label{eq:Psi-large-X}
\end{multline}
where
\begin{equation}
\phi(X,v):=-\frac{p}{2}\ln(X)-2p\ln(b-a)-\frac{1}{4}\pi-2\pi p^2+\arg\left(\Gamma(\ii p)\right)
\end{equation}
and $p:=\ln(2)/(2\pi)$ while $a=a(v)<b=b(v)$ are the two critical points of $\vartheta(z;v)$ nearest the origin.  The $O(X^{-1})$ estimate is uniform on compact subintervals of $|v|<54^{-1/2}$.
\label{theorem:large-X}
\end{theorem}
In the formula \eqref{eq:Psi-large-X}, we may use the critical point equations $\vartheta'(a;v)=\vartheta'(b;v)=0$ to obtain $\vartheta''(a;v)=6v+2a^{-1}<0$ and $\vartheta''(b;v)=6v+2b^{-1}>0$.  

In the special case of $v=0$, the asymptotic formula \eqref{eq:Psi-large-X} becomes even more explicit because
\begin{equation}
a=a(0)=-\sqrt{2},\quad\theta(a(0);0)=-2\sqrt{2},\quad\theta''(a(0);0)=-\sqrt{2}
\end{equation}
and
\begin{equation}
b=b(0)=\sqrt{2},\quad\theta(b(0);0)=2\sqrt{2},\quad\theta''(b(0);0)=\sqrt{2}.
\end{equation}
Therefore, we have the following.
\begin{corollary}
\begin{multline}
\Psi^+(X,0)=\\
\frac{2^{5/4}}{X^{3/4}}\sqrt{\frac{\ln(2)}{2\pi}}\cos\left(4\sqrt{2}X^{1/2}-\frac{\ln(2)}{4\pi}\ln(X) -\frac{9(\ln(2))^2}{4\pi}-\frac{1}{4}\pi + \arg\left(\Gamma\left(\frac{\ii\ln(2)}{2\pi}\right)\right)\right) + O(X^{-1}),\\\quad X\to +\infty.
\label{eq:Psi-large-X-T0}
\end{multline}
\label{corollary:large-X-T0}
\end{corollary}
The accuracy of the asymptotic formul\ae\ recorded in Theorem~\ref{theorem:large-X} and Corollary~\ref{corollary:large-X-T0} is illustrated in plots in Section~\ref{sec:plots-large-X}.

\subsection{Asymptotic behavior of $\Psi^\pm(X,T)$ for large $T$}
\label{sec:large-T}
It suffices to analyze $\Psi^+(X,T)$ for $X\ge 0$ and $T>0$ large.  We therefore introduce a non-negative parameter $w\ge 0$ and set $X=wT^{2/3}$ (note that $w=v^{-2/3}$ where $v=TX^{-3/2}$ parametrizes the large-$X$ asymptotics as described in Section~\ref{sec:large-X}), and rescale the spectral parameter $\Lambda$ by $\Lambda=T^{-1/3}z$.  The phase conjugating the jump matrix for $\mathbf{R}^+(\Lambda;X,T)$ then takes the form
\begin{equation}
\Lambda X+\Lambda^2T + 2\Lambda^{-1}=T^{1/3}\theta(z;w),\quad\theta(z;w):=wz+z^2+2z^{-1}.
\end{equation}
Setting $\mathbf{S}(z;T,w):=\mathbf{R}^+(T^{-1/3}z;T^{2/3}w,T)$, from \eqref{eq:Psi-R} we get
\begin{equation}
\Psi^+(T^{2/3}w;T)=2\ii T^{-1/3}\lim_{z\to\infty}zS_{12}(z;T,w).
\label{eq:Psi-T}
\end{equation}
As before, it is easy to see that $\mathbf{S}(z;T,w)\to\mathbb{I}$ as $z\to\infty$ for each $T>0$ and that $\mathbf{S}(z;T,w)$ is analytic in the complement of an arbitrary Jordan curve $\Gamma$ about $z=0$ in the clockwise sense, across which we have the jump condition
\begin{equation}
\mathbf{S}_+(z;T,w)=\mathbf{S}_-(z;T,w)\ee^{-\ii T^{1/3}\theta(z;w)\sigma_3}\mathbf{Q}^{-1}\ee^{\ii T^{1/3}\theta(z;w)\sigma_3},\quad z\in\Gamma.
\end{equation}
Since the analysis in Section~\ref{sec:large-X} is uniformly valid for $|v|$ bounded below the critical value of $54^{-1/2}$, i.e., for $w$ bounded above the corresponding critical value of $54^{1/3}\approx 3.78$, we will henceforth assume that $0\le w < 54^{1/3}$.

\subsubsection{Spectral curve, $g$-function, and steepest descent}  Suppose that $g(z;w)$ is a scalar function bounded and analytic for $z$ in the complement of a finite number of arcs of $\Gamma$ (cuts), that satisfies $g(z;w)\to 0$ as $z\to\infty$, and for which the boundary values taken on each cut from the interior and exterior of $\Gamma$ satisfy
\begin{equation}
g_+(z;w)+g_-(z;w)+2\theta(z;w)=\text{constant}
\end{equation}
where the constant in question can depend parametrically on $w$ and can be different in each cut.  It is straightforward to check that the function $(g'(z;w)+\theta'(z;w))^2$ is necessarily analytic for $z\in\mathbb{C}\setminus\{0\}$.  Expanding for large $z$ shows that
\begin{equation}
(g'(z;w)+\theta'(z;w))^2=4z^2+4wz+w^2+O(z^{-1}),\quad z\to\infty
\end{equation}
because $g'(z;w)=O(z^{-2})$ as $z\to\infty$.  Similarly, expanding for small $z$ shows that
\begin{equation}
(g'(z;w)+\theta'(z;w))^2=4z^{-4}+O(z^{-2}),\quad z\to 0
\end{equation}
because $g$ is analytic at the origin.  By Liouville's theorem it follows that for some coefficients $C_3(w)$ and $C_2(w)$,
\begin{equation}
(g'(z;w)+\theta'(z;w))^2=z^{-4}P(z;w),\quad P(z;w):=4z^6+4wz^5+w^2z^4+C_3(w)z^3+C_2(w)z^2 + 4.
\label{eq:spectral-curve-general}
\end{equation}
This algebraic relation is the relevant \emph{spectral curve} for the problem at hand.  It can take different forms under various additional assumptions on $C_3(w)$ and $C_2(w)$.  

The main case we will be interested in here is that in which $C_3(w)$ and $C_2(w)$ are such that the sextic $P$ factors as the product of the square of a quadratic factor and a second quadratic factor, i.e., $P$ has two double roots and two simple roots:
\begin{equation}
P(z;w)=4(z^2+d_1(w)z+d_0(w))^2(z^2+s_1(w)z+s_0(w)).
\label{eq:P-factors}
\end{equation}
Expanding out the right-hand side and comparing with the determinate coefficients of $z^5$, $z^4$, $z^1$, and $z^0$ obtained from \eqref{eq:spectral-curve-general} on the left-hand side gives the relations
\begin{equation}
\begin{split}4w&=8d_1(w)+4s_1(w)\\
w^2&=4d_1(w)^2+8d_0(w)+8d_1(w)s_1(w)+4s_0(w)\\
0&=4d_0(w)^2s_1(w)+8d_1(w)d_0(w)s_0(w)\\
4&=4d_0(w)^2s_0(w).
\end{split}
\label{eq:abcd-equations}
\end{equation}
From the first, third, and fourth equations, $d_1(w)$, $s_1(w)$, and $s_0(w)$ can be explicitly eliminated in favor of $d_0(w)$ and $w$:
\begin{equation}
d_1(w)=\frac{1}{2}\frac{wd_0(w)^3}{d_0(w)^3-1},\quad s_1(w)=-\frac{w}{d_0(w)^3-1},\quad s_0(w)=d_0(w)^{-2}.
\end{equation}
The second equation then becomes a relation between $d_0(w)$ and $w$ only:
\begin{equation}
8d_0(w)^9-12d_0(w)^6-2w^2d_0(w)^5-w^2d_0(w)^2+4=0.
\end{equation}
Remarkably, this equation factors as a product of three cubics:
\begin{equation}
(2d_0(w)^3+1)(2d_0(w)^3-wd_0(w)-2)(2d_0(w)^3+wd_0(w)-2)=0.
\end{equation}
By taking $d_0(w)=-2^{-1/3}$, the equations \eqref{eq:abcd-equations} have a simple particular solution:
\begin{equation}
d_1(w):=\frac{1}{6}w,\quad
d_0(w):=-2^{-1/3},\quad
s_1(w):=\frac{2}{3}w,\quad
s_0(w):=2^{2/3}.
\label{eq:coefficient-values}
\end{equation}
With these values, the undetermined coefficients $C_3(w)$ and $C_2(w)$ become explicit functions of $w$ via the identity \eqref{eq:P-factors}, but we will not need these going forward.  The double roots of $P(z;w)$ are therefore 
\begin{equation}
a(w):=\frac{1}{2}\left(-\frac{1}{6}w-\sqrt{\frac{w^2}{36}+2^{5/3}}\right)<0\quad\text{and}\quad
b(w):=\frac{1}{2}\left(-\frac{1}{6}w+\sqrt{\frac{w^2}{36}+2^{5/3}}\right)>0.
\end{equation}
Furthermore, the simple roots of $P(z;w)$ form a complex-conjugate pair $z_0,z_0^*$ with $\mathrm{Im}(z_0)>0$ exactly when $0\le w <54^{1/3}$:
\begin{equation}
z_0(w):=\frac{1}{3}\left(-w+\ii\sqrt{54^{2/3}-w^2}\right),\quad 0\le w<54^{1/3}.
\label{eq:z0}
\end{equation}
In this situation, there is only one cut for the $g$-function, namely an arc $\Sigma$ connecting the conjugate pair of simple roots $z_0$ and $z_0^*$ of $P(z;w)$.  Since this cut must be an arc $\Sigma\subset\Gamma$, we choose $\Sigma$ to cross the real axis at the negative value $z=a(w)$, and complete $\Gamma$ with a complementary arc that crosses the real axis at the positive value $z=b(w)$.

Combining \eqref{eq:spectral-curve-general} with \eqref{eq:P-factors}, we then obtain $g'(z;w)$ in the form
\begin{equation}
g'(z;w)=-\theta'(z;w) + 2z^{-2}(z^2+d_1(w)z+d_0(w))R(z;w),\quad R(z;w)^2=z^2+s_1(w)z+s_0(w),\quad z\in\mathbb{C}\setminus\Sigma,
\end{equation}
where $R(z;w)$ is analytic for $z\in\mathbb{C}\setminus\Sigma$ and satisfies $R(z;w)=z+O(1)$ as $z\to\infty$.  Note that any apparent singularity at $z=0$ necessarily cancels since the form of the sextic $P(z;w)$ was predicated on the assumed analyticity of $g'(z;w)$ at the origin.  Similarly, the above formula automatically satisfies $g'(z;w)=O(z^{-2})$ as $z\to\infty$.  Therefore $g(z;w)$ is well-defined for $z\in\mathbb{C}\setminus\Sigma$ by integration from infinity:
\begin{equation}
g(z;w)=\int_\infty^zg'(\zeta;w)\,\dd\zeta,\quad z\in\mathbb{C}\setminus\Sigma
\end{equation}
where the path of integration is arbitrary in the indicated domain.  It remains to specify $\Sigma$ precisely.

To fully determine $\Sigma$, note that the exponent function that will play a key role below is given by
\begin{equation}
h(z;w):=g(z;w)+\theta(z;w).
\end{equation}
The curves in the complex $z$-plane along which $\mathrm{Im}(h(z;w))$ is constant may be described as trajectories of a rational quadratic differential, i.e., they satisfy the condition $h'(z;w)^2\,\dd z^2>0$, where $h'(z;w)^2$ is the rational function 
\begin{equation}
h'(z;w)^2=z^{-4}P(z;w)=4z^{-4}(z-a(w))^2(z-b(w))^2(z-z_0(w))(z-z_0(w)^*).
\end{equation}
Whereas $h(z;w)$ is a multi-valued function with a branch cut $\Sigma$, the trajectories defined by $h'(z;w)^2\,\dd z>0$ form a well-defined system of curves in the $z$-plane.  Indeed, from standard existence/uniqueness theory for ordinary differential equations, it follows that each point $z$ that is not a pole or zero of $h'(z;w)^2$ lies on a unique trajectory.  Local analysis shows that there are precisely three trajectories emanating from each of the simple zeros of $h'(z;w)^2$, i.e., from the points $z_0(w)$ and $z_0(w)^*$.  Similarly, there are precisely four trajectories emanating from each of the double zeros of $h'(z;w)^2$, i.e., from the points $a(w)$ and $b(w)$, two emanating horizontally and two vertically from each.  The real axis in the $z$-plane is the union of trajectories $(-\infty,a(w))$, $(a(w),0)$, $(0,b(w))$, and $(b(w),+\infty)$ and the three exceptional points $a(w)<0<b(w)$.  These are clearly part of the level set $\mathrm{Im}(h(z;w))=0$.  Now, the fact that 
\begin{equation}
\oint_C h'(z;w)\,\dd z=0
\end{equation}
holds when $C$ is any Jordan curve enclosing $\Sigma$ in its interior (because $g'(z;w)=O(z^{-2})$ as $z\to\infty$ and $\theta'(z;w)$ has no residues) can be combined with the generalized Cauchy integral theorem to yield
\begin{equation}
\int_{z_0(w)^*}^{z_0(w)}h'_+(z;w)\,\dd z +\int_{z_0(w)}^{z_0(w)^*}h'_+(z;w)\,\dd z=0
\end{equation}
where the integration is taken along opposite sides of the branch cut $\Sigma$ (the subscript $+$ indicates a boundary value from the left as $\Sigma$ is traversed in the indicated direction).  But since $h'(z;w)$ is proportional to $R(z;w)$, it changes sign across $\Sigma$, and consequently both terms on the left-hand side are equal.  Therefore
\begin{equation}
h(z_0(w)^*;w)=h(z_0(w);w)\quad\implies\quad \mathrm{Im}(h(z_0(w)^*;w))=\mathrm{Im}(h(z_0(w);w)),\quad 0\le w<54^{1/3}.
\end{equation}
Since the two points $z=z_0(w)$ and $z=z_0(w)^*$ lie on the same level set of $\mathrm{Im}(h(z;w))$,  it is possible that they may be connected by a union of trajectories and one of the exceptional points on the real axis, and that this is so is easily confirmed by making plots of the trajectories emanating from $z_0(w)$.  In fact, of the three trajectories emanating from $z_0(w)$, one terminates at $z=a(w)$, one terminates at $z=b(w)$, and the third goes to infinity in the upper half-plane.  Denoting the trajectory joining $z=z_0(w)$ and $z=a(w)$ as $\Sigma^+$, we define $\Sigma$ precisely as the closure of the union of $\Sigma^+$ with its Schwarz reflection $\Sigma^-=(\Sigma^+)^*$.  With this choice, it follows that $\mathrm{Im}(h(z;w))$ can be defined on the whole $z$-plane as a continuous function.  Indeed, no matter where the branch cut $\Sigma$ is placed, it holds that $\mathrm{Im}(h_+(z;w))=-\mathrm{Im}(h_-(z;w))$ for $z\in\Sigma$ because the sum of the boundary values of $h$ is constant along $\Sigma$ and obviously real at the point $\{a(w)\}=\Sigma\cap\mathbb{R}$.  Hence the condition that $\mathrm{Im}(h(z;w))$ is continuous across $\Sigma$ is precisely that $\Sigma$ be a component of the zero level set of $\mathrm{Im}(h(z;w))$.  Note, however, that the normal derivative of $\mathrm{Im}(h(z;w))$ is not continuous across $\Sigma$; indeed $\mathrm{Im}(h(z;w))$ takes the same sign on both sides of $\Sigma^+$.  
The sign chart of $\mathrm{Im}(h(z;w))$ with the above choice of $\Sigma$ is illustrated in Figure~\ref{fig:w-plots}.
\begin{figure}[h]
\begin{center}
\includegraphics{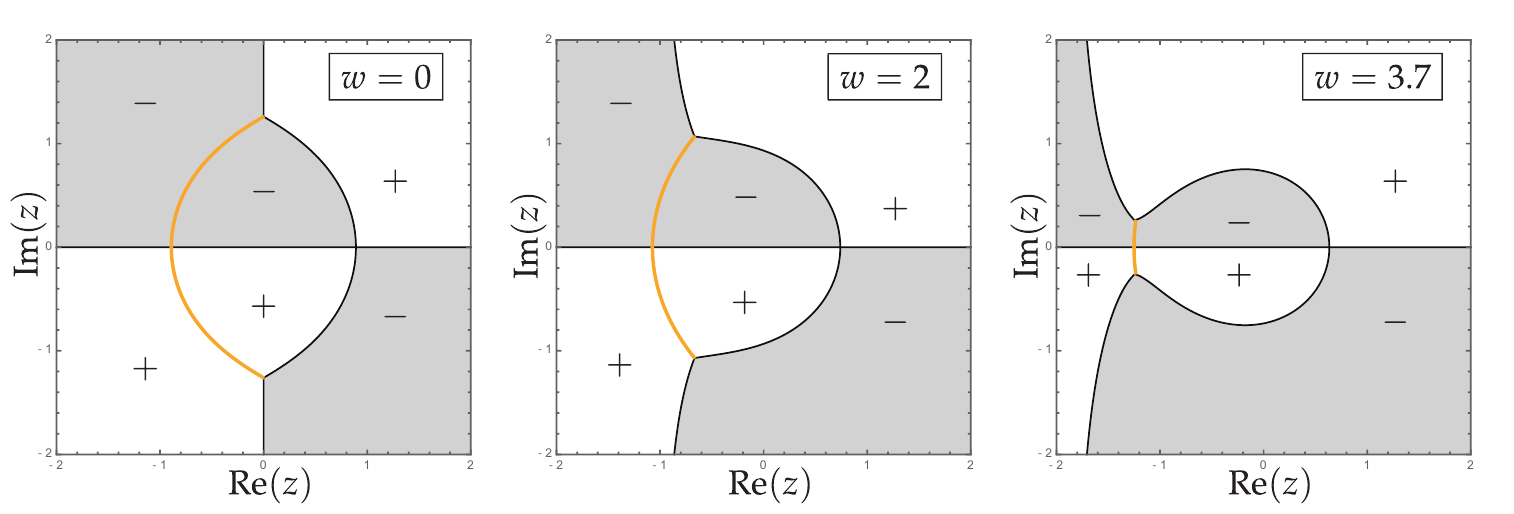}
\end{center}
\caption{Sign charts for $\mathrm{Im}(h(z;w))$ as $w$ varies in the interval $0\le w <54^{1/3}\approx 3.78$.  The orange arc in each plot is $\Sigma$, and it is part of the zero level curve $\mathrm{Im}(h(z;w))=0$ although the sign of $\mathrm{Im}(h(z;w))$ does not change upon crossing it.}
\label{fig:w-plots}
\end{figure}
Thus, $\mathrm{Im}(h(z;w))$ becomes a continuous function on the whole $z$-plane that is harmonic except for $z\in\Sigma$.  

We take the jump contour $\Gamma$ so that $\mathrm{Im}(h(z;w))=0$ holds for $z\in\Gamma$.  It consists of four arcs, $\Gamma^\pm$ and $\Sigma^\pm$ as indicated in the left-hand panel of Figure~\ref{fig:regions-Tjump-Tlim}.
\begin{figure}[h]
\begin{center}
\includegraphics{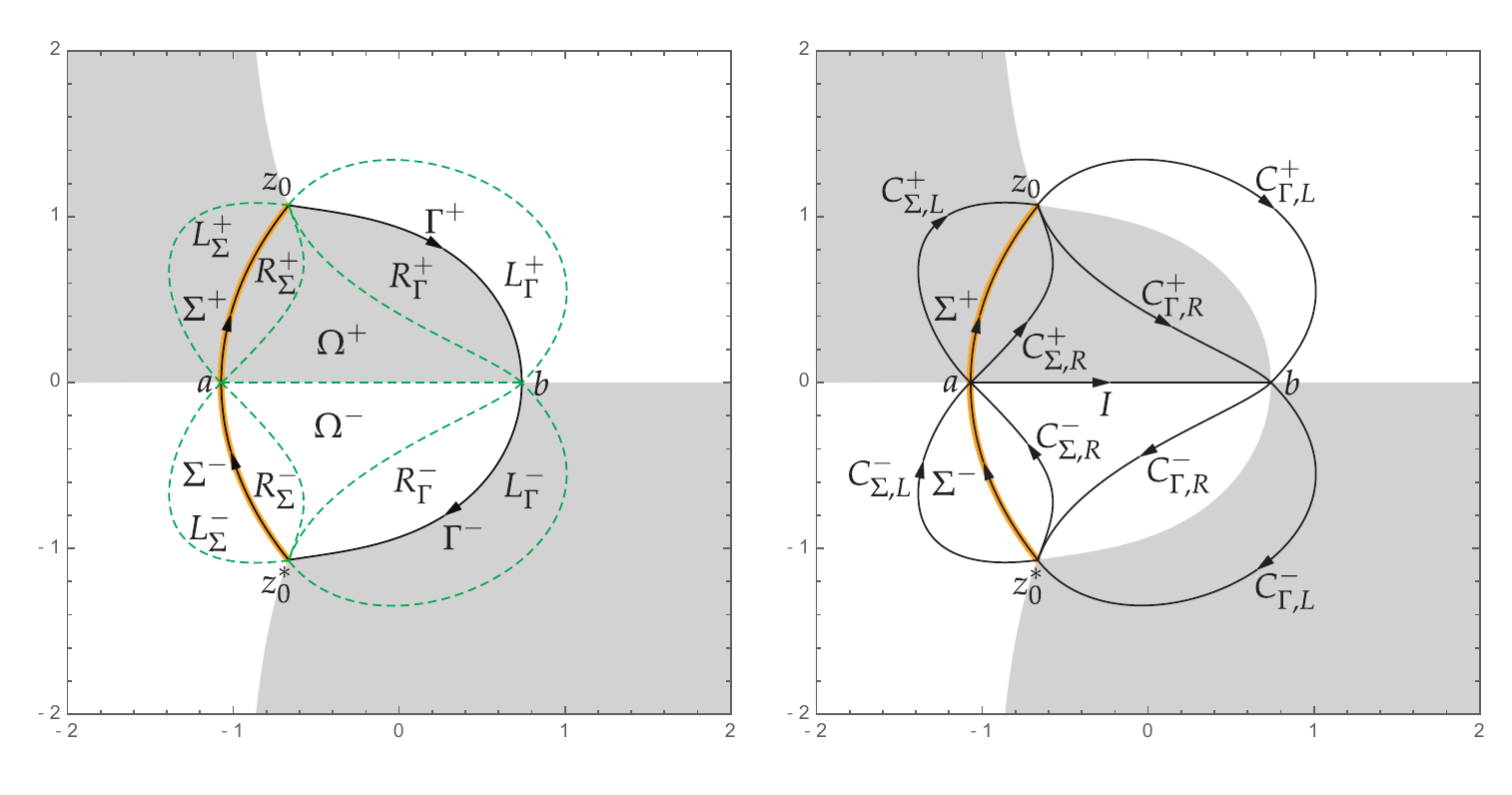}
\end{center}
\caption{Left:  the jump contour $\Gamma=\Gamma^+\cup\Gamma^-\cup\Sigma^+\cup\Sigma^-$ for $\mathbf{S}$ and the regions $L^\pm_\Gamma$, $L^\pm_\Sigma$, $R^\pm_\Gamma$, $R^\pm_\Sigma$, and $\Omega^\pm$.  Right.  The jump contour for $\mathbf{T}$.}
\label{fig:regions-Tjump-Tlim}
\end{figure}
Referring also to the left-hand panel of Figure~\ref{fig:regions-Tjump-Tlim}, we introduce the $g$-function and take advantage of matrix factorizations to separate the exponential factors via the following substitutions:
\begin{equation}
\mathbf{T}(z;T,w):=
\mathbf{S}(z;T,w)\begin{bmatrix}1 & 0\\\ee^{2\ii T^{1/3}\theta(z;w)} & 1\end{bmatrix}\ee^{\ii T^{1/3}g(z;w)\sigma_3},\quad z\in L^+_\Gamma,
\label{eq:T-T-large-first}
\end{equation}
\begin{equation}
\mathbf{T}(z;T,w):=
\mathbf{S}(z;T,w)2^{\sigma_3/2}\begin{bmatrix}1 & \tfrac{1}{2}\ee^{-2\ii T^{1/3}\theta(z;w)}\\0 & 1\end{bmatrix}\ee^{\ii T^{1/3}g(z;w)\sigma_3},\quad z\in R^+_\Gamma,
\end{equation}
\begin{equation}
\mathbf{T}(z;T,w):=
\mathbf{S}(z;T,w)2^{\sigma_3/2}\ee^{\ii T^{1/3}g(z;w)\sigma_3},\quad z\in\Omega^+,
\end{equation}
\begin{equation}
\mathbf{T}(z;T,w):=
\mathbf{S}(z;T,w)2^{-\sigma_3/2}\ee^{\ii T^{1/3}g(z;w)\sigma_3},\quad z\in\Omega^-,
\end{equation}
\begin{equation}
\mathbf{T}(z;T,w):=
\mathbf{S}(z;T,w)2^{-\sigma_3/2}\begin{bmatrix}1 & 0\\-\tfrac{1}{2}\ee^{2\ii T^{1/3}\theta(z;w)} & 1\end{bmatrix}\ee^{\ii T^{1/3}g(z;w)\sigma_3},\quad z\in R^-_\Gamma,
\end{equation}
\begin{equation}
\mathbf{T}(z;T,w):=
\mathbf{S}(z;T,w)\begin{bmatrix}1 & -\ee^{-2\ii T^{1/3}\theta(z;w)}\\0 & 1\end{bmatrix}
\ee^{\ii T^{1/3}g(z;w)\sigma_3},\quad z\in L^-_\Gamma,
\end{equation}
\begin{equation}
\mathbf{T}(z;T,w):=
\mathbf{S}(z;T,w)2^{\sigma_3/2}\begin{bmatrix}1 & -\tfrac{1}{2}\ee^{-2\ii T^{1/3}\theta(z;w)}\\0 & 1
\end{bmatrix}\ee^{\ii T^{1/3}g(z;w)\sigma_3},\quad z\in R^+_\Sigma,
\end{equation}
\begin{equation}
\mathbf{T}(z;T,w):=
\mathbf{S}(z;T,w)\begin{bmatrix}1 & \ee^{-2\ii T^{1/3}\theta(z;w)}\\0 & 1\end{bmatrix}\ee^{\ii T^{1/3}g(z;w)\sigma_3},\quad z\in L^+_\Sigma,
\end{equation}
\begin{equation}
\mathbf{T}(z;T,w):=
\mathbf{S}(z;T,w)2^{-\sigma_3/2}\begin{bmatrix}1 & 0\\\tfrac{1}{2}\ee^{2\ii T^{1/3}\theta(z;w)} & 1\end{bmatrix}\ee^{\ii T^{1/3}g(z;w)\sigma_3},\quad z\in R^-_\Sigma,
\end{equation}
\begin{equation}
\mathbf{T}(z;T,w):=
\mathbf{S}(z;T,w)\begin{bmatrix}1 & 0\\-\ee^{2\ii T^{1/3}\theta(z;w)} & 1\end{bmatrix}\ee^{\ii T^{1/3}g(z;w)\sigma_3},\quad z\in L^-_\Sigma,
\label{eq:T-T-large-last}
\end{equation}
and elsewhere we simply set $\mathbf{T}(z;T,w):=\mathbf{S}(z;T,w)\ee^{\ii T^{1/3}g(z;w)\sigma_3}$.  The jump contour for $\mathbf{T}(z;T,w)$ is illustrated in the right-hand panel of Figure~\ref{fig:regions-Tjump-Tlim}.  As in the large-$X$ analysis of Section~\ref{sec:large-X}, $\mathbf{T}(z;T,w)$ extends continuously and hence analytically to the arcs $\Gamma^\pm$ of the original jump contour.  The remaining jump conditions are the following.
\begin{equation}
\mathbf{T}_+(z;T,w)=\mathbf{T}_-(z;T,w)\begin{bmatrix}1 & 0\\-\ee^{2\ii T^{1/3}h(z;w)} & 1\end{bmatrix},\quad z\in C^+_{\Gamma,L},
\label{eq:Tjump-1-Tlim}
\end{equation}
\begin{equation}
\mathbf{T}_+(z;T,w)=\mathbf{T}_-(z;T,w)\begin{bmatrix}1 & \tfrac{1}{2}\ee^{-2\ii T^{1/3}h(z;w)}\\0 & 1
\end{bmatrix},\quad z\in C^+_{\Gamma,R},
\end{equation}
\begin{equation}
\mathbf{T}_+(z;T,w)=\mathbf{T}_-(z;T,w)2^{\sigma_3},\quad z\in I,
\label{eq:I-jump}
\end{equation}
\begin{equation}
\mathbf{T}_+(z;T,w)=\mathbf{T}_-(z;T,w)\begin{bmatrix}1 & 0\\-\tfrac{1}{2}\ee^{2\ii T^{1/3}h(z;w)} & 1\end{bmatrix},\quad z\in C^-_{\Gamma,R},
\end{equation}
\begin{equation}
\mathbf{T}_+(z;T,w)=\mathbf{T}_-(z;T,w)\begin{bmatrix}1 & \ee^{-2\ii T^{1/3}h(z;w)} \\0 & 1\end{bmatrix},\quad z\in C^-_{\Gamma,L},
\end{equation}
\begin{equation}
\mathbf{T}_+(z;T,w)=\mathbf{T}_-(z;T,w)\begin{bmatrix}1 & -\ee^{-2\ii T^{1/3}h(z;w)}\\0 & 1\end{bmatrix},\quad z\in C^+_{\Sigma,L},
\end{equation}
\begin{equation}
\mathbf{T}_+(z;T,w)=\mathbf{T}_-(z;T,w)\begin{bmatrix}1 & -\tfrac{1}{2}\ee^{-2\ii T^{1/3}h(z;w)}\\0 & 1\end{bmatrix},\quad z\in C^+_{\Sigma,R},
\end{equation}
\begin{equation}
\mathbf{T}_+(z;T,w)=\mathbf{T}_-(z;T,w)\begin{bmatrix}1 & 0\\\tfrac{1}{2}\ee^{2\ii T^{1/3}h(z;w)} & 1\end{bmatrix},\quad z\in C^-_{\Sigma,R},
\end{equation}
\begin{equation}
\mathbf{T}_+(z;T,w)=\mathbf{T}_-(z;T,w)\begin{bmatrix}1 & 0\\\ee^{2\ii T^{1/3}h(z;w)} & 1\end{bmatrix},\quad z\in C^-_{\Sigma,L},
\end{equation}
and finally,
\begin{equation}
\mathbf{T}_+(z;T,w)=\mathbf{T}_-(z;T,w)\begin{bmatrix}0 & \ee^{-\ii T^{1/3}\kappa(w)}\\
-\ee^{\ii T^{1/3}\kappa(w)} & 0\end{bmatrix},\quad z\in\Sigma=\Sigma^+\cup\Sigma^-,
\label{eq:twist-jump}
\end{equation}
where $\kappa(w)$ is the real constant value of $g_++g_-+2\theta$ along $\Sigma$:
\begin{equation}
\kappa(w):=g_+(z;w)+g_-(z;w)+2\theta(z;w)\in\mathbb{R},\quad z\in\Sigma.
\label{eq:kappa-def}
\end{equation}
Since $\mathrm{Im}(h(z;w))>0$ holds on $C^+_{\Gamma,L}$, $C^-_{\Gamma,R}$, $C^-_{\Sigma,L}$, and $C^-_{\Sigma,R}$ while $\mathrm{Im}(h(z;w))<0$ holds on $C^+_{\Gamma,R}$, $C^-_{\Gamma,L}$, $C^+_{\Sigma,L}$, and $C^+_{\Sigma,R}$, the jump matrix on all of these arcs converges exponentially to the identity as $T\to +\infty$, with the convergence being uniform away from the endpoints of the arcs.

\subsubsection{Parametrix construction}
We first construct an outer parametrix $\dot{\mathbf{T}}^\mathrm{out}(z;T,w)$ satisfying exactly the jump conditions on $I$ and $\Sigma$ (cf., \eqref{eq:I-jump} and \eqref{eq:twist-jump}) that do not become asymptotically trivial as $T\to +\infty$.  Recalling from Section~\ref{sec:large-X} the corresponding outer parametrix that satisfies the jump condition \eqref{eq:I-jump} on $I$ we may write $\dot{\mathbf{T}}^\mathrm{out}(z;T,w)$ in the form
\begin{equation}
\dot{\mathbf{T}}^\mathrm{out}(z;T,w)=\mathbf{G}(z;T,w)\left(\frac{z-a(w)}{z-b(w)}\right)^{\ii p\sigma_3},\quad p:=\frac{\ln(2)}{2\pi},
\label{eq:T-out-F}
\end{equation}
where the power function is defined as the principal branch.  Then, $\mathbf{G}(z;T,w)$ extends analytically to $I$, and we will assume that it is bounded near $z=a(w),b(w)$ in particular making it analytic at $z=b(w)$.  Therefore, $\mathbf{G}(z;T,w)$ is analytic for $z\in\mathbb{C}\setminus\Sigma$ and tends to the identity as $z\to\infty$.  Across $\Sigma$, the constant jump condition \eqref{eq:twist-jump} required of $\dot{\mathbf{T}}^\mathrm{out}(z;T,w)$ becomes modified for $\mathbf{G}(z;T,w)$:
\begin{equation}
\mathbf{G}_+(z;T,w)=\mathbf{G}_-(z;T,w)\left(\frac{z-a(w)}{z-b(w)}\right)^{\ii p\sigma_3}
\begin{bmatrix}0 & \ee^{-\ii T^{1/3}\kappa(w)}\\-\ee^{\ii T^{1/3}\kappa(w)} & 0\end{bmatrix}
\left(\frac{z-a(w)}{z-b(w)}\right)^{-\ii p\sigma_3},\quad z\in \Sigma.
\end{equation}
To solve for $\mathbf{G}(z;T,w)$, we will convert this back into a constant jump condition on $\Sigma$ alone by the following substitution:
\begin{equation}
\mathbf{G}(z;T,w)=\mathbf{H}(z;T,w)\ee^{-k(z;w)\sigma_3},
\label{eq:F-G}
\end{equation}
where $k(z;w)$ is given by 
\begin{equation}
k(z;w):=\ii p\log\left(\frac{z-a(w)}{z-b(w)}\right)+\ii p R(z;w)\int_{a(w)}^{b(w)}\frac{\dd s}{R(s;w)(s-z)}+\frac{1}{2}\ii\mu(w),
\end{equation}
in which the logarithm is given by the principal branch $-\pi<\mathrm{Im}(\log(\cdot))<\pi$, and where the constant $\mu(w)$ is given by
\begin{equation}
\mu(w):=2p\int_{a(w)}^{b(w)}\frac{\dd s}{R(s;w)}>0.
\label{eq:mu-def}
\end{equation}
It is straightforward to confirm that $k(z;w)$ has the following properties.  By definition of $\mu(w)$, it satisfies $k(z;w)=O(z^{-1})$ as $z\to\infty$.  Despite appearances, there is no jump across $(a(w),b(w))$ as is easily confirmed by comparing the boundary values of the logarithm and using the Plemelj formula.  The apparent singularities at $z=a(w),b(w)$ are removable, so the domain of analyticity for $k(z;w)$ is $z\in\mathbb{C}\setminus\Sigma$, and $k(z;w)$ takes continuous boundary values on $\Sigma$, including at the endpoints.  These boundary values are related by the condition
\begin{equation}
k_+(z;w)+k_-(z;w)=2\ii p\log\left(\frac{z-a(w)}{z-b(w)}\right)+\ii\mu(w),\quad z\in\Sigma.
\end{equation}
It follows that $\mathbf{H}(z;T,w)$ is a matrix function analytic for $z\in\mathbb{C}\setminus\Sigma$ that tends to $\mathbb{I}$ as $z\to\infty$, and that satisfies the jump condition
\begin{equation}
\mathbf{H}_+(z;T,w)=\mathbf{H}_-(z;T,w)\begin{bmatrix}0 & \ee^{-\ii (T^{1/3}\kappa(w)+\mu(w))}\\
-\ee^{\ii (T^{1/3}\kappa(w)+\mu(w))} & 0\end{bmatrix},\quad z\in\Sigma.
\end{equation}
It is straightforward to solve for $\mathbf{H}(z;T,w)$ by diagonalizing the constant jump matrix, which has eigenvalues $\pm\ii$.  All solutions of the jump condition for $\mathbf{H}(z;T,w)$ have singularities at the endpoints of $\Sigma$, and we select the unique solution with the mildest rate of growth as $z\to z_0(w),z_0(w)^*$:
\begin{equation}
\mathbf{H}(z;T,w)=\ee^{-\ii (T^{1/3}\kappa(w)+\mu(w))\sigma_3/2}\mathbf{U}\left(\frac{z-z_0(w)}{z-z_0(w)^*}\right)^{\sigma_3/4}\mathbf{U}^{-1}\ee^{\ii (T^{1/3}\kappa(w)+\mu(w))\sigma_3/2},\quad\mathbf{U}:=\frac{1}{\sqrt{2}}\begin{bmatrix}1 & 1\\\ii & -\ii\end{bmatrix}.
\label{eq:G}
\end{equation}
Here, the power function in the central factor is defined to be analytic for $z\in\mathbb{C}\setminus\Sigma$ and to tend to $\mathbb{I}$ as $z\to\infty$.  Combining \eqref{eq:T-out-F}, \eqref{eq:F-G}, and \eqref{eq:G} completes the construction of the outer parametrix $\dot{\mathbf{T}}^\mathrm{out}(z;T,w)$.

This problem requires four inner parametrices, $\dot{\mathbf{T}}^a(z;T,w)$, $\dot{\mathbf{T}}^b(z;T,w)$, $\dot{\mathbf{T}}^{z_0}(z;T,w)$, and $\dot{\mathbf{T}}^{z_0^*}(z;T,w)$ to be defined in neighborhoods of $z=a$, $z=b$, $z=z_0$, and $z=z_0^*$ respectively.  Those defined near $z=a,b$ will be constructed in terms of parabolic cylinder functions exactly as in Section~\ref{sec:large-X}.  Those defined near $z=z_0,z_0^*$ can be constructed in terms of Airy functions.  In all four cases, the inner parametrix constitutes an exact local solution of the jump conditions for $\mathbf{T}(z;T,w)$.  The inner parametrices will have the following key properties:
\begin{equation}
\sup_{z\in\partial D_{a,b}(\delta)}\|\dot{\mathbf{T}}^{a,b}(z;T,w)\dot{\mathbf{T}}^\mathrm{out}(z;T,w)^{-1}-\mathbb{I}\|=O(T^{-1/6}),\quad T\to +\infty,
\label{eq:T-PC-mismatch}
\end{equation}
\begin{equation}
\sup_{z\in\partial D_{z_0,z_0^*}(\delta)}\|\dot{\mathbf{T}}^{z_0,z_0^*}(z;T,w)\dot{\mathbf{T}}^\mathrm{out}(z;T,w)^{-1}-\mathbb{I}\|=O(T^{-1/3}),\quad T\to +\infty,
\label{eq:T-Airy-mismatch}
\end{equation}
with both estimates\footnote{It is standard that for parabolic cylinder (resp., Airy) parametrices the mismatch error is proportional to the large parameter in the exponent, here $T^{1/3}$, to the power $-1/2$  (resp. $-1$).} holding uniformly for $w\ge 0$ bounded below the critical value of $54^{1/3}$.  
The global parametrix $\dot{\mathbf{T}}(z;T,w)$ is defined as in Section~\ref{sec:large-X} by setting $\dot{\mathbf{T}}(z;T,w)$ equal to  $\dot{\mathbf{T}}^\mathrm{out}(z;T,w)$ outside of the four disks and defining $\dot{\mathbf{T}}(z;T,w)$ within each of the four disks as the corresponding inner parametrix.  As in Section~\ref{sec:large-X}, the global parametrix has unit determinant.

\subsubsection{Error analysis}
It is straightforward to confirm that the error matrix $\mathbf{F}(z;T,w):=\mathbf{T}(z;T,w)\dot{\mathbf{T}}(z;T,w)^{-1}$ satisfies all of the necessary conditions of a small-norm Riemann-Hilbert problem.  The jump contour $\Sigma_\mathbf{F}$ for $\mathbf{F}(z;T,w)$ consists of the restrictions of the arcs $C_{\Sigma,L}^\pm$, $C_{\Sigma,R}^\pm$, $C_{\Gamma,L}^\pm$, and $C_{\Gamma,R}^\pm$ to the exterior of all four disks together  with the boundaries of all four disks.  The dominant contribution to the jump discrepancy $\mathbf{V}^\mathbf{F}-\mathbb{I}$ for $\mathbf{F}(z;T,w)$ lies on the boundaries of the disks $D_{a,b}(\delta)$, leading to the estimate
\begin{equation}
\sup_{z\in\Sigma_\mathbf{F}}\|\mathbf{V}^\mathbf{F}(z;T,w)-\mathbb{I}\|=O(T^{-1/6}),\quad T\to +\infty
\end{equation}
holding uniformly for $w\ge 0$ bounded below the critical value of $54^{1/3}$.  By the $L^2(\Sigma_\mathbf{F})$ theory of small-norm Riemann-Hilbert problems, some of which was described in Section~\ref{sec:large-X}, it follows that every coefficient $\mathbf{F}^n(T,w)$ in the Laurent series for $\mathbf{F}(z;T,w)$ convergent for sufficiently large $|z|$:
\begin{equation}
\mathbf{F}(z;T,w)=\mathbb{I}+\sum_{n=1}^\infty z^{-n}\mathbf{F}^n(T,w)
\end{equation}
satisfies $\|\mathbf{F}^n(T,w)\|=O(T^{-1/6})$ as $T\to +\infty$ uniformly for $w\ge 0$ bounded below $54^{1/3}$.  Since $\mathbf{S}(z;T,w)=\mathbf{T}(z;T,w)\ee^{-\ii T^{1/3}g(z;w)\sigma_3}$ and $\dot{\mathbf{T}}(z;T,w)=\dot{\mathbf{T}}^\mathrm{out}(z;T,w)$ both hold for $|z|$ sufficiently large, 
from \eqref{eq:Psi-T} we have
\begin{equation}
\begin{split}
\Psi^+(T^{2/3}w,T)&=2\ii T^{-1/3}\lim_{z\to\infty}zT_{12}(z;T,w)\ee^{\ii T^{1/3}g(z;w)}\\
&=2\ii T^{-1/3}\lim_{z\to\infty}z\left[F_{11}(z;T,w)\dot{T}^\mathrm{out}_{12}(z;T,w)+
F_{12}(z;T,w)\dot{T}^\mathrm{out}_{22}(z;T,w)\right]\ee^{\ii T^{1/3}g(z;w)}\\
&=2\ii T^{-1/3}\lim_{z\to\infty}z\left[\dot{T}^\mathrm{out}_{12}(z;T,w)+F_{12}(z;T,w)\right]\\
&=2\ii T^{-1/3}\lim_{z\to\infty}z\dot{T}^\mathrm{out}_{12}(z;T,w)+O(T^{-1/2}),\quad T\to +\infty.
\end{split}
\end{equation}
Explicitly substituting for the outer parametrix and using (from \eqref{eq:z0}) $\mathrm{Im}(z_0(w))=\tfrac{1}{3}\sqrt{54^{2/3}-w^2}$ completes the proof of the following result.
\begin{theorem}[Large-$T$ asymptotics of rogue waves of infinite order]
Let $0\le w<54^{1/3}$ be fixed.  Then
\begin{equation}
\Psi^+(T^{2/3}w,T)=-\ii T^{-1/3}\frac{1}{3}\sqrt{54^{2/3}-w^2}\ee^{-\ii (T^{1/3}\kappa(w)+\mu(w))} + O(T^{-1/2}),\quad T\to +\infty,
\end{equation}
where $\kappa(w)$ is defined by \eqref{eq:kappa-def} and $\mu(w)$ is defined by \eqref{eq:mu-def}.
The estimate $O(T^{-1/2})$ is uniform for $w$ in compact subintervals of $[0,54^{1/3})$.  
\label{theorem:large-T}
\end{theorem}
Simplifying the formula in the special case of $w=0$ we obtain:
\begin{equation}
\kappa(0)=-108^{1/3}\quad \text{and}\quad \mu(0)=\frac{2}{\pi}\ln(2)\mathrm{Arcsinh}(2^{-1/2}).
\end{equation}
Thus, we have the following corollary.
\begin{corollary}
\begin{equation}
\Psi^+(0,T)=\left(\frac{2}{T}\right)^{1/3}\exp\left(\ii\left[(108T)^{1/3}-\frac{1}{2}\pi-\frac{2}{\pi}\ln(2)\mathrm{Arcsinh}(2^{-1/2})\right]\right) + O(T^{-1/2}),\quad T\to +\infty.
\label{eq:Psi-large-T-X0}
\end{equation}
\label{corollary:large-T-X0}
\end{corollary}

\subsection{Transitional asymptotic behavior}
\label{sec:Painleve}
The large-$X$ analysis of Section~\ref{sec:large-X} fails as $v\uparrow 54^{-1/2}$ while the large-$T$ analysis of Section~\ref{sec:large-T} fails as $w\uparrow 54^{1/3}$.  These two upper bounds actually correspond to the same curve in the $(X,T)$-plane, namely $T=\pm 54^{-1/2}|X|^{3/2}$.  In this section, we obtain transitional asymptotics of $\Psi^+(X,T)$ uniformly valid for $v=T|X|^{-3/2}$ in the neighborhood of the critical value $v=v_\mathrm{c}:=54^{-1/2}$ with either $X$ or $T$ taken to be large.  Since we are taking $v$ as the parameter, we return to the setting of Section~\ref{sec:large-X} and try to extend that approach to a neighborhood of the threshold $v=v_\mathrm{c}$.

\begin{remark}
The curves $T=\pm 54^{-1/2}|X|^{3/2}$ appear to also be relevant in the ``far-field'' asymptotic description of fundamental rogue waves or large order.  Indeed, given a value of $n$ (recall $k=2n$ or $k=2n-1$) these curves can be plotted in the $(x,t)$-plane via the substitutions $T=n^2t$ and $X=nx$; these can be seen as the red curves in Figure~\ref{fig:density-plots}.  The results in this paper do not justify any connection between these red curves and the behavior of $\psi_k(x,t)$ for $k$ large except in a neighborhood of the origin where $X$ and $T$ are bounded so that Theorem~\ref{theorem:main} applies.  The asymptotic analysis of fundamental rogue waves outside of this small neighborhood is the subject of ongoing work \cite{BilmanLMT18} that we hope to be able to report on soon.
\end{remark}

When $v\approx v_\mathrm{c}$, there is one real critical point $b(v)$ of $\vartheta(z;v)$ near $b_\mathrm{c}:=b(v_\mathrm{c})=\sqrt{\tfrac{3}{2}}$
and a pair of critical points (real for $v<v_\mathrm{c}$ and complex-conjugate for $v>v_\mathrm{c}$) near the double critical point $a_\mathrm{c}:=a(v_\mathrm{c})=-\sqrt{6}$.  Note that $\vartheta(a_\mathrm{c},v_\mathrm{c})=-\sqrt{6}$.  The Taylor expansion of $\vartheta(z;v)$ about $z=a_\mathrm{c}$ reads
\begin{multline}
\vartheta(z;v)=-\sqrt{6}+6(v-v_\mathrm{c})-2\sqrt{6}(v-v_\mathrm{c})(z-a_\mathrm{c}) + (v-v_\mathrm{c})(z-a_\mathrm{c})^2 \\{}-\frac{1}{18}(z-a_\mathrm{c})^3 -\frac{1}{18\sqrt{6}}(z-a_\mathrm{c})^4+O((z-a_\mathrm{c})^5),\quad z\to a_\mathrm{c},
\label{eq:Taylor}
\end{multline}
and therefore at the critical value of $v=v_\mathrm{c}$ one has
\begin{equation}
\vartheta(z;v_\mathrm{c})=-\sqrt{6}-\frac{1}{18}(z-a_\mathrm{c})^3-\frac{1}{18\sqrt{6}}(z-a_\mathrm{c})^4+O((z-a_\mathrm{c})^5),\quad z\to a_\mathrm{c}.
\end{equation}
Following \cite{ChesterFU57}, we may define a Schwarz-symmetric conformal mapping $z\mapsto W$ in the neighborhood of $z=a_\mathrm{c}$ and $v=v_\mathrm{c}$ by the equation
\begin{equation}
2\vartheta(z;v)=W^3+rW-s,\quad W=W(z;v),\quad r=r(v),\quad s=s(v)
\label{eq:CFU}
\end{equation}
where $r$ and $s$ are real analytic functions of $v$ near $v_\mathrm{c}$ determined so that the two critical points of the left-hand side near $z=a_\mathrm{c}$ are mapped onto the two critical points of the cubic on the right-hand side, and where $r_\mathrm{c}:=r(v_\mathrm{c})=0$, $s_\mathrm{c}:=s(v_\mathrm{c})=2\sqrt{6}$, and $W'_\mathrm{c}:=W'(a_\mathrm{c};v_\mathrm{c})=-9^{-1/3}<0$.  We denote by $z_*(v)$ the pre-image of $W=0$.  It is an analytic function of $v$ that satisfies $z_*(v_\mathrm{c})=a_\mathrm{c}$.  Taking the derivative of \eqref{eq:CFU} with respect to $v$ and evaluating at $z=a_\mathrm{c}$ and $v=v_\mathrm{c}$ gives $s'_\mathrm{c}:=s'(v_\mathrm{c})=-12$.  Similarly, comparing the mixed second derivative of \eqref{eq:CFU} with respect to $z$ and $v$ with the third derivative of the same with respect to $z$ at $z=a_\mathrm{c}$ and $v=v_\mathrm{c}$ with the help of the Taylor expansion \eqref{eq:Taylor} one finds easily that $r'_\mathrm{c}:=r'(v_\mathrm{c})=4\cdot 6^{1/2}9^{1/3}>0$.

\subsubsection{Parametrix modification}
To extend the analysis from Section~\ref{sec:large-X} to this situation, we need only replace the outer parametrix formerly defined by \eqref{eq:T-out} by the slightly-modified definition
\begin{equation}
\dot{\mathbf{T}}^\mathrm{out}(z;v):=\left(\frac{z-z_*(v)}{z-b(v)}\right)^{\ii p\sigma_3},\quad p:=\frac{\ln(2)}{2\pi}>0,\quad z\in\mathbb{C}\setminus [z_*(v),b(v)],
\label{eq:T-out-PII}
\end{equation}
and then it is necessary to replace the inner parametrix formerly defined near $z=a$ in terms of parabolic cylinder functions with another one that takes into account the collision of critical points.  Let $\zeta=\zeta_a:=X^{1/6}W$ and $y:=X^{1/3}r$.  The jump conditions satisfied by $\mathbf{U}^a:=\mathbf{T}\ee^{\ii X^{1/2}s(v)\sigma_3/2}(\ii\sigma_2)$ near $z=a_\mathrm{c}$ can then be written in the form indicated in Figure~\ref{fig:PII} when the jump contours are locally taken to coincide with the five rays $\arg(\zeta)=\pm\pi/2$, $\arg(\zeta)=\pm 5\pi/6$, and $\arg(-\zeta)=0$.
\begin{figure}[h]
\begin{center}
\includegraphics{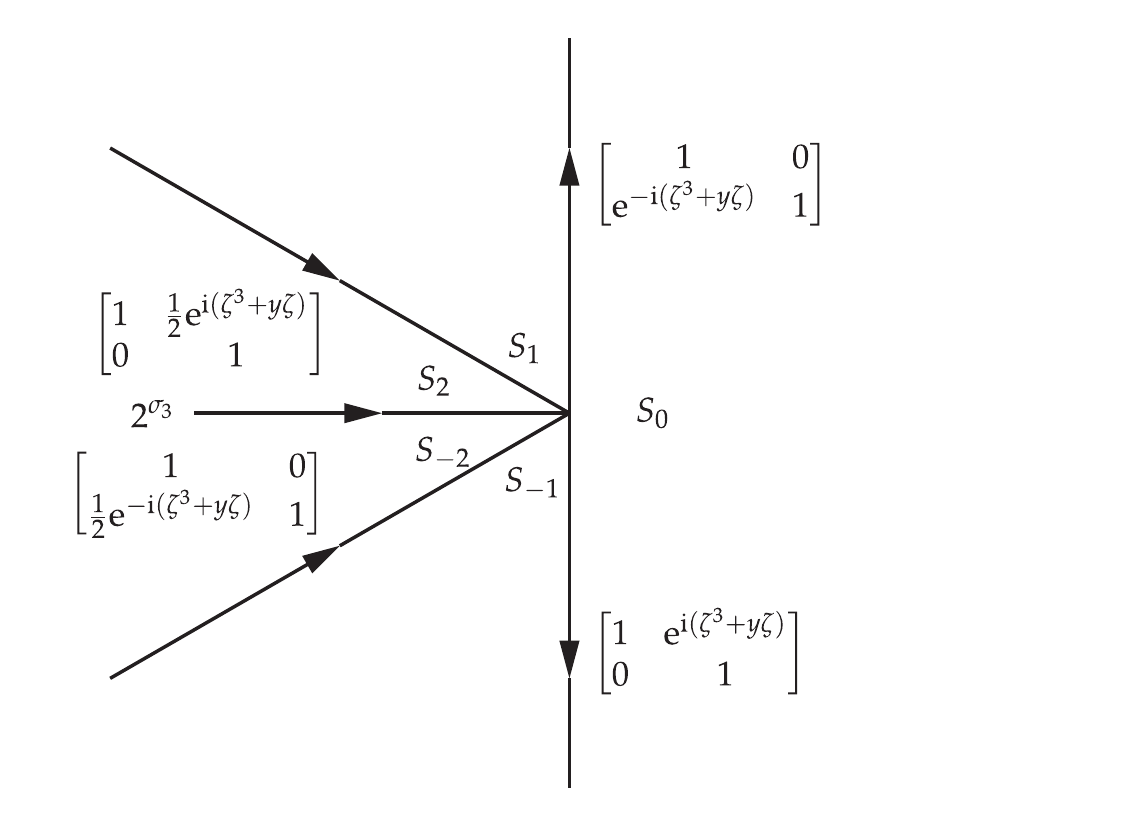}
\end{center}
\caption{The jump conditions satisfied by $\mathbf{U}^a$ take the form $\mathbf{U}^a_+=\mathbf{U}^a_-\mathbf{V}^\mathrm{PII}$ where the jump matrix $\mathbf{V}^\mathrm{PII}$ is defined on five rays in the $\zeta$-plane as shown.}
\label{fig:PII}
\end{figure}
As usual, we write the outer parametrix $\dot{\mathbf{T}}^\mathrm{out}(z;v)$ from Section~\ref{sec:large-X} locally near $z=a_\mathrm{c}$ in terms of the conformal coordinate $\zeta_a$:
\begin{multline}
\dot{\mathbf{T}}^\mathrm{out}(z;v)\ee^{\ii X^{1/2}s(v)\sigma_3/2}(\ii\sigma_2)=X^{-\ii p\sigma_3/6}\ee^{\ii X^{1/2}s(v)\sigma_3/2}
\mathbf{H}^a(z;v)\zeta_a^{-\ii p\sigma_3},\\
\mathbf{H}^a(z;v):=(b(v)-z)^{-\ii p\sigma_3}\left(\frac{z_*(v)-z}{W(z;v)}\right)^{\ii p\sigma_3}(\ii\sigma_2).
\end{multline}
As before, $\mathbf{H}^a(z;v)$ with the above modified definition is an analytic function near $z=a_\mathrm{c}$ and $v=v_\mathrm{c}$ and it is independent of $X$.  Taking into account the final factor on the right-hand side of this expression for $\dot{\mathbf{T}}^\mathrm{out}(z;v)$, we properly formulate a Riemann-Hilbert problem that is the analogue in the present setting of Riemann-Hilbert Problem~\ref{rhp:PC}.
\begin{rhp}[Painlev\'e-II parametrix]
Given $y\in\mathbb{R}$, seek a $2\times 2$ matrix-valued function $\mathbf{W}(\zeta;y)$ with the following properties.
\begin{itemize}
\item[]\textbf{Analyticity:}  $\mathbf{W}(\zeta;y)$ is analytic for $\zeta$ in the five sectors shown in Figure~\ref{fig:PII}, namely $S_0$:  $|\arg(\zeta)|<\tfrac{1}{2}\pi$, $S_1$:  $\tfrac{1}{2}\pi<\arg(\zeta)<\frac{5}{6}\pi$, $S_{-1}$:  $-\frac{5}{6}\pi<\arg(\zeta)<-\tfrac{1}{2}\pi$, $S_2$:  $\frac{5}{6}\pi<\arg(\zeta)<\pi$, and $S_{-2}$:  $-\pi<\arg(\zeta)<-\frac{5}{6}\pi$.  It takes continuous boundary values on the excluded rays and at the origin from each sector.
\item[]\textbf{Jump conditions:}  $\mathbf{W}_+(\zeta;y)=\mathbf{W}_-(\zeta;y)\mathbf{V}^{\mathrm{PII}}(\zeta;y)$, where $\mathbf{V}^\mathrm{PII}(\zeta;y)$ is the matrix defined on the jump contour shown in Figure~\ref{fig:PII}.  
\item[]\textbf{Normalization:}  $\mathbf{W}(\zeta;y)\zeta^{\ii p\sigma_3}\to\mathbb{I}$ as $\zeta\to\infty$ uniformly in all directions, where $p=\ln(2)/(2\pi)$.
\end{itemize}
\label{rhp:PII}
\end{rhp}
In \cite{Miller18} it is shown that 
this problem has a unique solution for all $y$ real.  The product $\mathbf{W}(\zeta;y)\zeta^{\ii p\sigma_3}$ admits a complete asymptotic expansion of the form
\begin{equation}
\mathbf{W}(\zeta;y)\zeta^{\ii p\sigma_3}\sim\mathbb{I}+\sum_{j=1}^\infty \mathbf{W}^j(y)\zeta^{-j},\quad \zeta\to\infty
\label{eq:W-expansion}
\end{equation}
uniformly in all directions of the complex $\zeta$-plane.  Furthermore (see \cite[Corollary 1]{Miller18}), the function $\mathcal{V}(y)$ defined by the formula
\begin{equation}
\mathcal{V}(y):=\lim_{\zeta\to\infty}\zeta W_{21}(\zeta;y)\zeta^{\ii p} = W^1_{21}(y)
\label{eq:cal-V-def}
\end{equation}
can be equivalently represented as follows.  There exists a unique \emph{tritronqu\'ee} solution $\mathcal{Q}(y)$ of the Painlev\'e-II differential equation 
\begin{equation}
\frac{\dd^2\mathcal{Q}}{\dd y^2}+\frac{2}{3}y\mathcal{Q}-2\mathcal{Q}^3-\frac{2}{3}\ii p-\frac{1}{3}=0
\label{eq:Q-PII}
\end{equation}
determined by the asymptotic behavior
\begin{equation}
\mathcal{Q}(y)=\ii\left(-\frac{y}{3}\right)^{\tfrac{1}{2}} -\left(\frac{1}{4}+\ii \frac{p}{2}\right)\frac{1}{y}+O(|y|^{-5/2}),\quad y\to\infty,\quad
|\arg(-y)|<\frac{2}{3}\pi.
\label{eq:Q-asymp-sector}
\end{equation}
This solution is asymptotically pole-free in maximally-wide sector of opening angle $4\pi/3$ of the complex $y$-plane, and it is also analytic for all $y\in\mathbb{R}$ and has trigonometric/algebraic asymptotic behavior as $y\to+\infty$, whereas if $y\to\infty$ in any other direction of the complementary sector $|\arg(y)|<\pi/3$, $\mathcal{Q}(y)$ behaves like an elliptic function.  The alternate formula for $\mathcal{V}(y)$ is then
\begin{equation}
\mathcal{V}(y)=\begin{cases}\displaystyle\frac{\ii\beta}{2}
\ee^{-\tfrac{2}{9}\sqrt{3}\ii(-y)^{\tfrac{3}{2}}}(-3y)^{-(\tfrac{1}{4}+\ii \tfrac{p}{2})}\exp\left(\int_{-\infty}^y\left[\mathcal{Q}(\eta)-\ii\left(-\frac{\eta}{3}\right)^{\tfrac{1}{2}}+\left(\frac{1}{4}+\ii \frac{p}{2}\right)\frac{1}{\eta}\right]\,\dd\eta\right),& y<0\\
\displaystyle
\mathcal{V}(-1)\exp\left(\int_{-1}^y\mathcal{Q}(\eta)\,\dd\eta\right),& y\ge 0.
\end{cases}
\label{eq:cal-V-alternate}
\end{equation}
Finally, $\mathcal{V}(y)$ has the asymptotic behavior
\begin{equation}
\mathcal{V}(y)=-\sqrt{\frac{y}{6}}\left(\frac{y}{6}\right)^{\ii p} + O(y^{-1/4}),\quad y\to +\infty.
\end{equation}

From the solution of Riemann-Hilbert Problem~\ref{rhp:PII} we define the inner parametrix near $z=a_\mathrm{c}$ as follows:
\begin{equation}
\dot{\mathbf{T}}^a(z;X,v):=X^{-\ii p\sigma_3/6}\ee^{\ii X^{1/2}s(v)\sigma_3/2}\mathbf{H}^a(z;v)\mathbf{W}(X^{1/6}W(z;v);X^{1/3}r(v))(-\ii\sigma_2)\ee^{-\ii X^{1/2}s(v)\sigma_3/2},\quad z\in D_{a_\mathrm{c}}(\delta).
\label{eq:T-a-PII}
\end{equation}
The analogue of \eqref{eq:a-mismatch} is then
\begin{multline}
\dot{\mathbf{T}}^a(z;X,v)\dot{\mathbf{T}}^\mathrm{out}(z;v)^{-1}=X^{-\ii p\sigma_3/6}\ee^{\ii X^{1/2}s(v)\sigma_3/2}\mathbf{H}^a(z;v)\mathbf{W}(\zeta_a;y)\zeta_a^{\ii p\sigma_3}\mathbf{H}^a(z;v)^{-1}\ee^{-\ii X^{1/2}s(v)\sigma_3/2}X^{\ii p\sigma_3/6},\\
\zeta_a=X^{1/6}W(z;v),\quad y=X^{1/3}r(v),\quad z\in \partial D_{a_\mathrm{c}}(\delta).
\label{eq:a-mismatch-PII}
\end{multline}
The inner parametrix near $z=b(v)$ is constructed from parabolic cylinder functions and installed in the disk $D_b(\delta)$ exactly as in Section~\ref{sec:large-X}.  The global parametrix is again given by the piecewise definition \eqref{eq:global-parametrix-large-X} with the understanding that $\dot{\mathbf{T}}^\mathrm{out}(z;v)$ has a slightly different definition (cf., \eqref{eq:T-out-PII}) and that $\dot{\mathbf{T}}^a(z;X,v)$ is built from Riemann-Hilbert Problem~\ref{rhp:PII} via \eqref{eq:T-a-PII} in the present situation.

\subsubsection{Error analysis}
The error $\mathbf{F}(z;X,v)$ is defined in terms of the relevant global parametrix exactly as in \eqref{eq:F-def}.  The analysis of the corresponding jump matrix $\mathbf{V}^\mathbf{F}(z;X,v)$ is exactly as in Section~\ref{sec:large-X} except that the dominant contribution to $\mathbf{V}^\mathbf{F}(z;X,v)-\mathbb{I}$ now arises only from the boundary of the disk $D_{a_\mathrm{c}}(\delta)$ and it is large compared to $X^{-1/4}$, proportional to $X^{-1/6}$.  Indeed, since $\zeta_a$ is proportional to $X^{1/6}$ when $z\in\partial D_{a_\mathrm{c}}(\delta)$ while $\mathbf{V}^\mathbf{F}(z;X,v)=\dot{\mathbf{T}}^a(z;X,v)\dot{\mathbf{T}}^\mathrm{out}(z;v)^{-1}$ and the conjugating factors in \eqref{eq:a-mismatch-PII} are bounded, this is a consequence of the expansion \eqref{eq:W-expansion}.  Therefore once again we have a small-norm Riemann-Hilbert problem for $\mathbf{F}(z;X,v)$ solvable by Neumann series applied to a corresponding singular integral equation for $\mathbf{F}_-(z;X,v)$ (cf., \eqref{eq:F-minus-integral-equation}).  The estimate $\mathbf{F}_-(\cdot;X,v)-\mathbb{I}=O(X^{-1/6})$ therefore holds in the $L^2(\Sigma_\mathbf{F})$ sense, and it follows that \eqref{eq:Psi-Plus-Integral} holds in which the error term is $O(X^{-5/6})$ instead of $O(X^{-1})$.  Without changing the order of the error we may then take the integration to be over the clockwise-oriented circle $\partial D_{a_\mathrm{c}}(\delta)$ instead of all of $\Sigma_\mathbf{F}$.  Using the first three terms in the expansion \eqref{eq:W-expansion} in \eqref{eq:a-mismatch-PII} shows that
\begin{equation}
V_{12}^\mathbf{F}(z;X,v)=X^{-1/6}X^{-\ii p/3}\ee^{\ii X^{1/2}s(v)}(\mathbf{H}^a(z;v)\mathbf{W}^1(X^{1/3}r(v))\mathbf{H}^a(z;v)^{-1})_{12}W(z;v)^{-1}+O(X^{-1/3})
\end{equation}
holds uniformly for $z\in\partial D_{a_\mathrm{c}}(\delta)$ as $X\to+\infty$, assuming that $y$ is bounded.  Therefore, evaluating an integral by residues at $z=z_*(v)$ where $W(z;v)$ has its only (simple) zero within $D_{a_\mathrm{c}}(\delta)$, 
\begin{equation}
\begin{split}
\Psi^+(X,X^{3/2}v)&=-\frac{1}{\pi X^{1/2}}\int_{\partial D_{a_\mathrm{c}}(\delta)}V^\mathbf{F}(w;X,v)\,\dd w + O(X^{-5/6})\\
&=\frac{2\ii X^{-\ii p/3}\ee^{\ii X^{1/2}s(v)}}{X^{2/3}W'(z_*(v);v)}(\mathbf{H}^a(z_*(v);v)\mathbf{W}^1(X^{1/3}r(v))\mathbf{H}^a(z_*(v);v)^{-1})_{12}+O(X^{-5/6}).
\end{split}
\end{equation}
By l'H\^opital's rule,
\begin{equation}
\mathbf{H}^a(z_*(v);v) = \left(-W'(z_*(v);v)(b(v)-z_*(v))\right)^{-\ii p\sigma_3}(\ii\sigma_2),
\end{equation}
and therefore
\begin{equation}
\Psi^+(X,X^{3/2}v)=-\frac{2\ii X^{-\ii p/3}\ee^{\ii X^{1/2}s(v)}}{X^{2/3}W'(z_*(v);v)}(-W'(z_*(v);v)(b(v)-z_*(v)))^{-2\ii p}\mathcal{V}(X^{1/3}r(v)) + O(X^{-5/6})
\end{equation}
where we recall the notation that $\mathcal{V}(y):=W^1_{21}(y)$. This holds uniformly for $v$ sufficiently close to $v_\mathrm{c}$ if also $\mathcal{V}(X^{1/3}r(v))$ remains bounded as $X\to +\infty$.  If we assume that $v-v_c=O(X^{-1/3})$, then the formula simplifies to
\begin{equation}
\Psi^+(X,X^{3/2}v)=-\frac{2\ii X^{-\ii p/3}\ee^{\ii X^{1/2}s_\mathrm{c}}\ee^{\ii X^{1/2}s'_\mathrm{c}\cdot(v-v_\mathrm{c})}}{X^{2/3}W'_\mathrm{c}}(-W'_\mathrm{c}\cdot(b_\mathrm{c}-a_\mathrm{c}))^{-2\ii p}\mathcal{V}(X^{1/3}r'_\mathrm{c}\cdot(v-v_\mathrm{c})) + O(X^{-5/6}).
\end{equation}
Using $p=\ln(2)/(2\pi)$ along with
\begin{equation}
s_\mathrm{c}=2^{3/2}3^{1/2},\quad s'_\mathrm{c}=-2^23,\quad a_\mathrm{c}=-2^{1/2}3^{1/2},\quad b_\mathrm{c}=2^{-1/2}3^{1/2},\quad r'_\mathrm{c}=2^{5/2}3^{7/6},\quad W'_\mathrm{c}=-3^{-2/3},
\end{equation}
we obtain the following result.
\begin{theorem}[Transitional asymptotics of rogue waves of infinite order]
Let $\mathcal{V}(y)$ be defined from Riemann-Hilbert Problem~\ref{rhp:PII} by \eqref{eq:cal-V-def} or
equivalently in terms of the tritronqu\'ee solution $\mathcal{Q}(y)$ of the Painlev\'e-II equation \eqref{eq:Q-PII} by \eqref{eq:cal-V-alternate}.  Then
\begin{equation}
\Psi^+(X,X^{3/2}v)=\frac{2\cdot 3^{2/3}}{X^{2/3}}\ee^{\ii\phi(X;v)}\mathcal{V}(X^{1/3}2^{5/2}3^{7/6}(v-v_c))+O(X^{-5/6}),\quad X\to\infty,\quad v-v_c=O(X^{-1/3}),
\end{equation}
where $v_\mathrm{c}:=54^{-1/2}$ and where the phase is
\begin{equation}
\phi(X;v):=2^{3/2}3^{1/2}X^{1/2}-2^23X^{1/2}(v-v_c)-\frac{\ln(2)}{6\pi}\ln(X) +\frac{\pi}{2}
-\frac{5\ln(2)\ln(3)}{6\pi}+\frac{(\ln(2))^2}{2\pi}.
\end{equation}
\label{theorem:PII}
\end{theorem}
\section{Numerical Computation of Rogue Waves of Infinite Order}
\label{sec:numerics}
\subsection{Numerical methods for Riemann-Hilbert problems}
In order to compute $\Psi^{+}(X,T)$ numerically, we make use of three Riemann-Hilbert problems that are considered in Section~\ref{sec:near-field} and Section~\ref{sec:Psi-asymptotic}, namely
\begin{itemize}
\item Riemann-Hilbert Problem~\ref{rhp:limit-simpler},
\item the large-$X$ deformation of Riemann-Hilbert Problem~\ref{rhp:limit-simpler}, satisfied by $\mathbf{T}(z;X,v)$ with the jump conditions given in \eqref{eq:Tjump-1} through \eqref{eq:Tjump-4} (see Figure~\ref{fig:regions-Tjump}) and the normalization $\mathbf{T}(z;X,v)\to\mathbb{I}$ as $z\to\infty$,
\item the large-$T$ deformation of Riemann-Hilbert Problem~\ref{rhp:limit-simpler}, satisfied by $\mathbf{T}(z;T,w)$ with the jump conditions given in \eqref{eq:Tjump-1-Tlim} through \eqref{eq:twist-jump} (see Figure~\ref{fig:regions-Tjump-Tlim}) and the normalization $\mathbf{T}(z;T,w)\to\mathbb{I}$ as $z\to\infty$.
\end{itemize}
These Riemann-Hilbert problems can be treated numerically with the aid of \texttt{RHPackage} \cite{RHPackage} in context of the numerical methodology developed in \cite{TrogdonO16-book} (see also \cite{Olver12} and \cite{TrogdonO13}).  The basic idea is to discretize the underlying singular integral equation associated with the given Riemann-Hilbert problem; an in-depth description and analysis of the accuracy of the numerical method employed can be found in \cite{TrogdonO13} and \cite[Chapter 2 and Chapter 7]{TrogdonO16-book}. 

Note that for a given $T>0$, the large-$X$ deformation algebraically makes sense only when $X>X_*(T)$, where $X_*(T)={54}^{1/3}T^{2/3}$, and the large-$T$ deformation algebraically makes sense only when $X<X_*(T)$. For $X>X_*(T)$ and large, we numerically encode the jump conditions associated with the deformed jump contour illustrated in Figure~\ref{fig:regions-Tjump} and compute the solution of the resulting Riemann-Hilbert problem (satisfied by $\mathbf{T}(z;X,v)$) to compute $\Psi^{+}(X,T)$. In practice, the jump contours are truncated if the jump matrix supported on these contours differs from the identity matrix by at most \texttt{machine epsilon}. See Figure~\ref{fig:large-X-contours} for these numerical contours. 
\begin{figure}[ht]
\begin{center}
\includegraphics{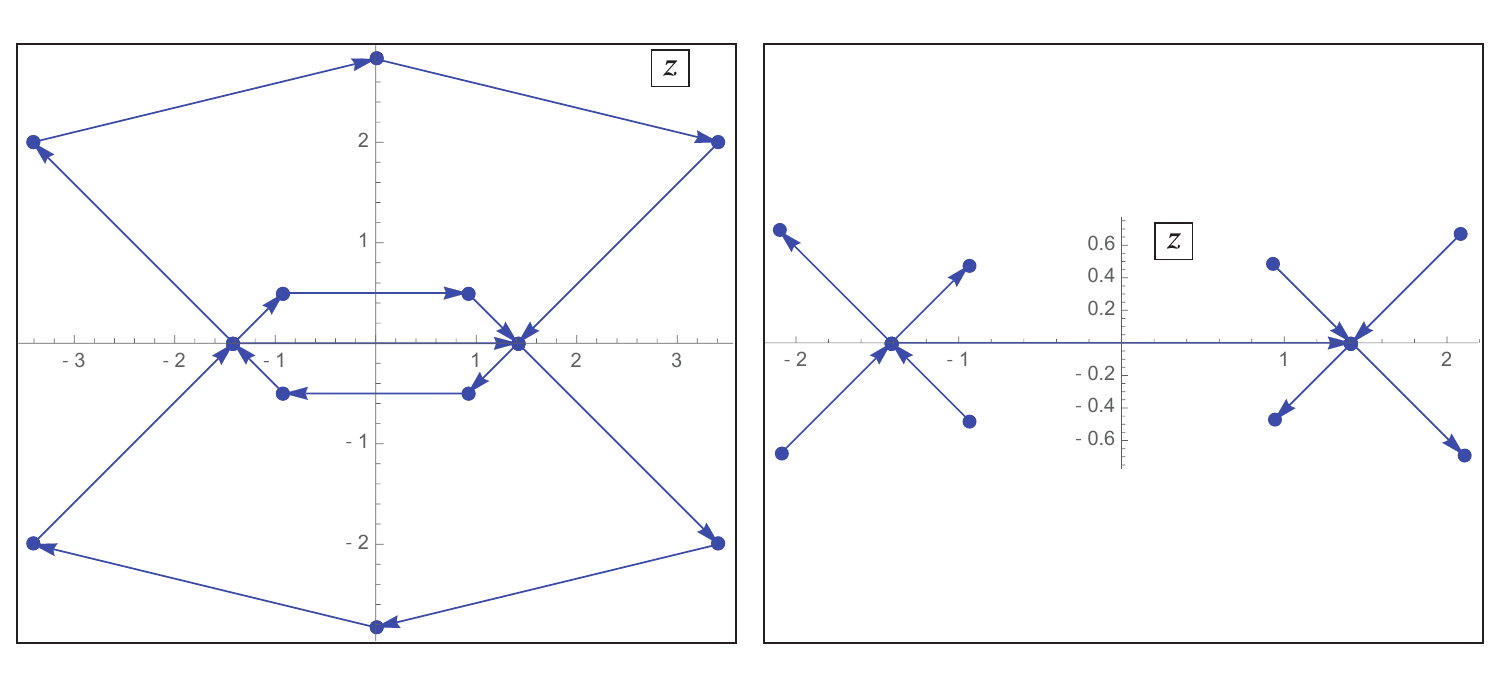}
\end{center}
\caption{Left: numerical parametrizations of the jump contours for the large-$X$ deformation. Right: truncated jump contours are used in practice if $X$ is large. For both plots, $X=2000$ and $v=0$.  Note the difference in scale of the two plots.}
\label{fig:large-X-contours}
\end{figure}

As $X$ becomes large, although the jump matrices tend to the identity matrix rapidly away from the stationary phase points $a=a(v)$ and $b=b(v)$ of the exponent $\vartheta(z;v)$, their Sobolev norms (derivatives with respect to $z$) grow and this presents a numerical challenge, which is overcome by a rescaling algorithm in the \texttt{RHPackage} (see \cite[Algorithm 7.1]{TrogdonO16-book}. Thus, in order to compute $\Psi^{+}(X,T)$ for large values of $X$ in the region $X>X_{*}(T)$ in a way that is \emph{asymptotically robust}, one needs to remove the \emph{connecting} jump condition \eqref{eq:T-diagonal} on the contour $I$ although the jump matrix is bounded there (in fact, a constant diagonal matrix). A detailed discussion on this issue and the method can be found in \cite[Chapter 7]{TrogdonO16-book}. As pointed out in Section~\ref{sec:Psi-asymptotic}, the outer parametrix given in \eqref{eq:T-out} by
\begin{equation}
\dot{\mathbf{T}}^\mathrm{out}(z,v)=\left(\frac{z-a(v)}{z-b(v)}\right)^{\ii p\sigma_3},\quad p:=\frac{\ln(2)}{2\pi}>0,\quad z\in\mathbb{C}\setminus I.
\end{equation}
exactly satisfies the jump condition
\begin{equation}
\dot{\mathbf{T}}^\mathrm{out}(z;v) =\dot{\mathbf{T}}^\mathrm{out}(z;v)2^{\sigma_3},\quad z\in I,
\end{equation}
and it is normalized as $\dot{\mathbf{T}}^\mathrm{out}(z;v)\to \mathbb{I}$ as $z\to\infty$. Thus, setting $\hat{\mathbf{T}}(z;X,v):= \mathbf{T}(z;X,v)\dot{\mathbf{T}}^\mathrm{out}(z;v)^{-1}$ for $z\in \mathbb{C}\setminus I$ removes the jump condition across $I$ while conjugating the existing other jump matrices given in \eqref{eq:Tjump-1} through \eqref{eq:Tjump-4} by $\dot{\mathbf{T}}^\mathrm{out}(z;v)$. However, $\dot{\mathbf{T}}^\mathrm{out}(z;v)$ has bounded singularities at $z=a$ and $z=b$.
As the remaining jump contours also pass from $a$ and $b$, this transformation introduces bounded singularities in the jump matrices at $z=a$ and $z=b$. To remedy this, we center small circles at the points $z=a$ and $z=b$ with counter-clockwise orientation and remove the jump matrices on the line segments inside these circles at the cost of having jump conditions on arcs of these circles connecting the endpoints of these line segments. While doing this removes the singular jump conditions, some components of the new jump matrices supported on the little circles centered at $z=a(v)$ and $z=b(v)$ now grow exponentially as $X\to +\infty$. Noting that for $\xi=a,b$
\begin{equation}
\vartheta(z;v) - \vartheta(\xi;v) = \frac{\vartheta''(\xi;v)}{2}(z-\xi)^2 + O((z-\xi)^3), \quad z\to \xi,
\end{equation}
we have
\begin{equation}
\ee^{\pm \ii |X|^{1/2}\vartheta(z;v)}=O(1),\quad X\to+\infty
\end{equation}
if $|z-\xi|=O(|X|^{-1/4})$ as $X\to+\infty$ for both $\xi =a(v)$ and $\xi=b(v)$. Therefore, we scale the common radius of these circles by $|X|^{-1/4}$ as $X$ becomes large. While shrinking the circles at a faster rate ensures boundedness of the exponentials supported on them, it also moves the support of the jump matrices closer to singularities at a faster rate and hence should be avoided. The jump contours of the Riemann-Hilbert problem used to compute $\Psi^{+}(X,T)$ numerically for large values of $X>X_{*}(T)$ is given in Figure~\ref{fig:large-X-truncated-contours}.

\begin{figure}[ht]
\begin{center}
\includegraphics{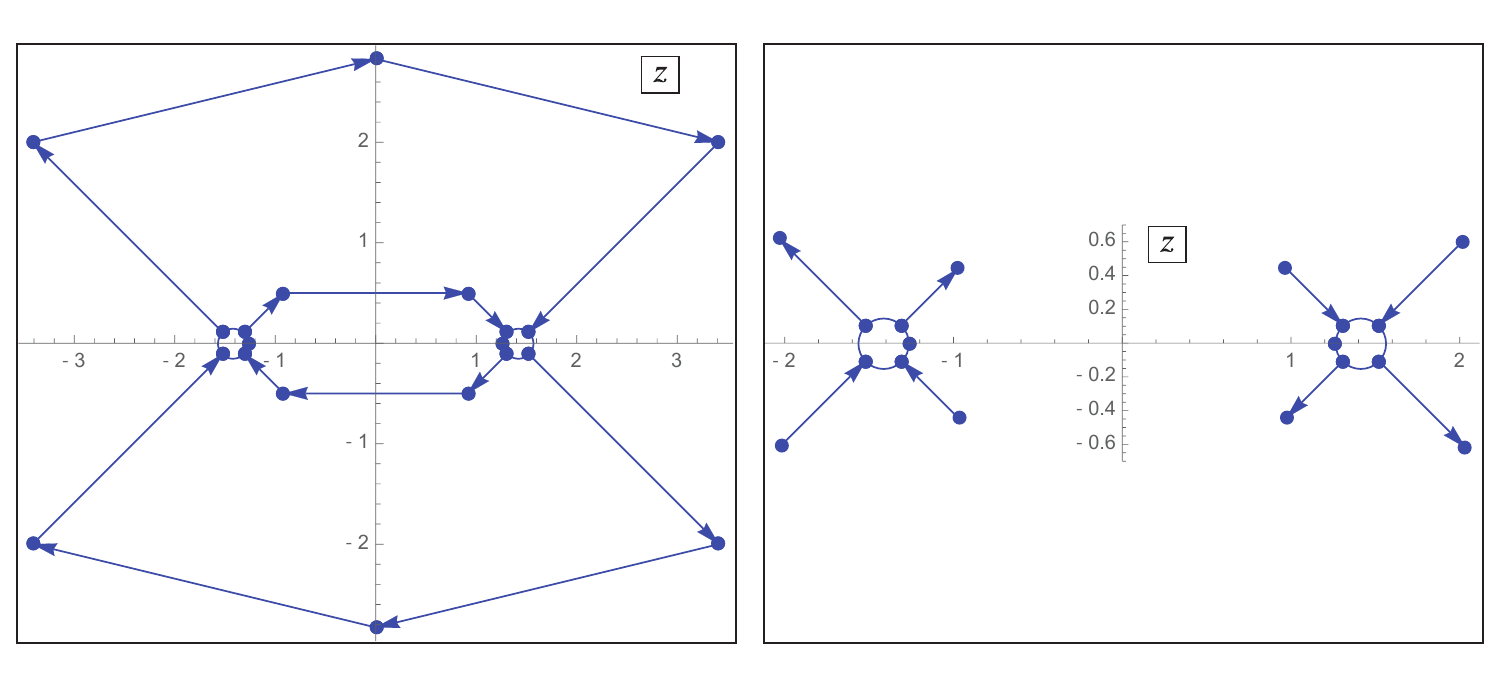}
\end{center}
\caption{Left: jump contours used in numerical solution of the Riemann-Hilbert problem satisfied by $\mathbf{T}(z;X,v)$, which is asymptotically and numerically well-adapted for large $X$. Right: truncated jump contours that are used in practice if $X$ is large. For both plots, $X=2000$ and $v=0$.  All circles are taken to have counterclockwise orientation.}
\label{fig:large-X-truncated-contours}
\end{figure}

Computing $\Psi^{+}(X,T)$ for $X<X_*(T)$ by solving the Riemann-Hilbert problem resulting from the large-$T$ contour deformation employed in Section~\ref{sec:large-T} (illustrated in Figure~\ref{fig:regions-Tjump-Tlim}) requires more machinery. The fundamental difference from the large-$X$ problem is the use of the function $g(z;w)$ and hence the appearance of the exponent function $h(z;w)=g(z;w)+ \theta(z;w)$ in the jump conditions \eqref{eq:Tjump-1-Tlim} through \eqref{eq:twist-jump}. Recall that $h(z;w)$ has a branch cut across the contour $\Sigma$, with end points $z_0(w)$ and its complex conjugate $z_0(w)^*$, along which $\mathrm{Im}(h(z;w))$ vanishes but does not change sign as $z$ crosses from left to right of $\Sigma$. Since
\begin{equation}
h(z;w)-h(\zeta;w)=O\Big(\big(z-\zeta\big)^{3/2}\Big),\quad z\to \zeta,~\zeta=z_0(w),\,z_0(w)^*,
\label{eq:h-rate}
\end{equation}
all of the jump matrices \eqref{eq:Tjump-1-Tlim} through \eqref{eq:twist-jump} involving the function $h(z;w)$ exhibit half-integer power type singularities at the points $z=z_0(w)$ and $z=z_0(w)^*$. These singularities can again be removed from the problem by introducing small clockwise-oriented circles around these points and defining $\hat{\mathbf{T}}(z;T,w)=\mathbf{T}(z;T,w)\ee^{-\ii T^{1/3}h(z;w)\sigma_3}$ inside these circles. Doing so results in constant jump matrices on the existing subarcs that lie inside the small circles and introduces a jump condition on the circles themselves where the corresponding jump matrix is given by $\ee^{\ii T^{1/3}h(z;w)\sigma_3}$. The rate in \eqref{eq:h-rate} implies
\begin{equation}
\ee^{\ii T^{1/3}h(z;w)\sigma_3} = O(1),\quad T\to+\infty
\end{equation}
if $|z-\zeta|T^{2/9}=O(1)$, $\zeta=z_0(w),\,z_0(w)^*$ as $T\to+\infty$, and hence we scale the common radius of the small circles centered at $z_0(w)$ and $z_0(w)^*$ by $T^{-2/9}$ as $T$ becomes large. A plot of numerical jump contours encoding the deformations is given in Figure~\ref{fig:large-T-contours}.
\begin{figure}
\begin{center}
\includegraphics{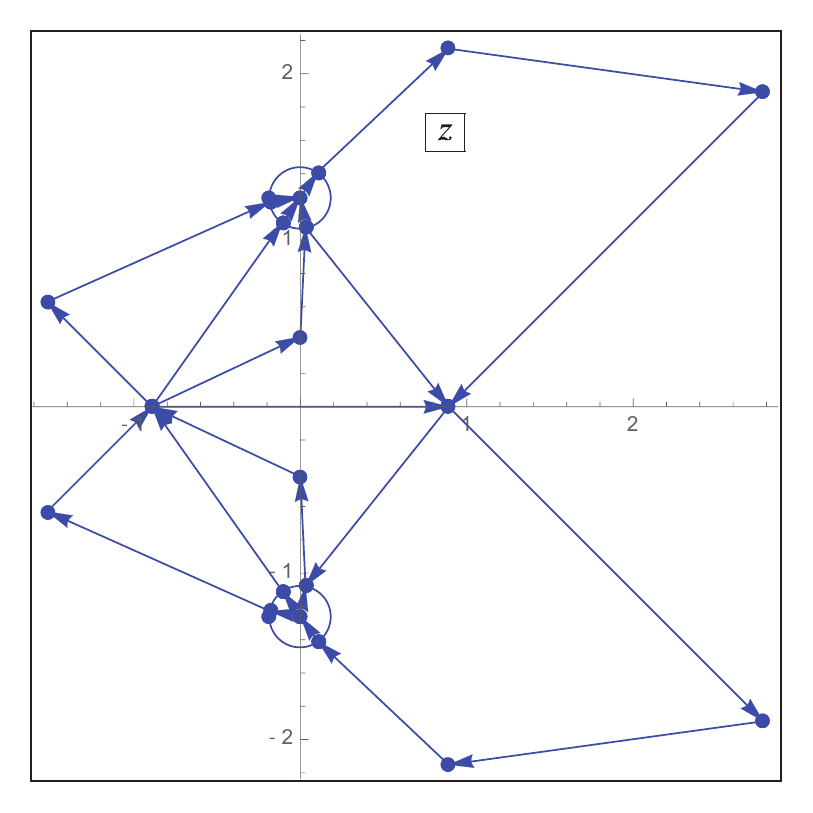}
\end{center}
\caption{Numerical jump contours for the large-$T$ deformation when $T=2000$ and $w=0$.  The circles have counterclockwise orientation.}
\label{fig:large-T-contours}
\end{figure}
A similar treatment for a Riemann-Hilbert problem with circular jump contours was done in \cite[Sections 4.3 and 4.4]{BilmanT17}, see also \cite{TrogdonO16-book} and the references therein, in particular, \cite{TrogdonOD14}. The numerical routines that are used to generate the data in this work are available from the \texttt{rogue-waves} online repository\footnote{\texttt{https://github.com/bilman/rogue-waves}}.

We note that for practical purposes it suffices to implement a square root function that has a branch cut consisting of the union of line segments connecting $z_0(w)$ to $a(w)$ and $a(w)$ to $z_0(w)^*$. We can explicitly compute $h(z;w)$ by finding an exact antiderivative using this square root function, and in this case $h(z;w)$ satisfies a jump condition on these line segments rather than on $\Sigma$ given in Figure~\ref{fig:regions-Tjump-Tlim}. This results in moving the jump condition \eqref{eq:twist-jump} from $\Sigma$ to the line segments described above (see Figure~\ref{fig:large-T-contours}).

It turns out that because the exponent function $h(z;w)$ is not analytic at the endpoints $z_0(w),z_0(w)^*$, the numerical solution is not as accurate when the circles are small as in the large-$X$ deformation case.  More work is therefore necessary to compute the solution of the large-$T$ problem in a way that is robust for large values of $T$. This involves removal of the constant jumps on $\Sigma$ and $I$, and contour truncation, although a solution that is more difficult to code but more elegant is simply to implement the Airy parametrices alluded to in Section~\ref{sec:large-T} (the latter approach avoids shrinking disks altogether). Such a refinement, together with the implementation for the transition region $X\approx X_*(T)$ described in Section~\ref{sec:Painleve} will appear in a forthcoming paper, where the special function $\Psi^{+}(X,T)$ will be computed accurately on the entire $(X,T)$ plane, including arbitrarily large values of the parameters, using different Riemann-Hilbert problems. The modules developed in these works will be merged and incorporated in the \texttt{ISTPackage} \cite{ISTPackage}.

\subsection{Plots of rogue waves of infinite order}
\label{sec:plots-of-Psi-plus}

For small values of $(X,T)$, e.g.\@ $0\leq X, T \leq 2$, Riemann-Hilbert Problem~\ref{rhp:limit-simpler} can be solved reliably without any deformations at all since the Sobolev norms of the jump matrices on $|\Lambda|=1$ remain small enough for numerical purposes. Thus, for $T>0$ small, one can cross-validate the computations by comparing numerical solutions of two different Riemann-Hilbert problems, i.e.\ calculating the difference $|\Psi^{+,X}(X,T)-\Psi^{+,\text{RHP-4}}(X,T)|$ for $X>X_*(T)$ and $|\Psi^{+,T}(X,T)-\Psi^{+,\text{RHP-4}}(X,T)|$ for $X<X_*(T)$, where $\Psi^{+,\alpha}(X,T)$ denotes the solution computed numerically using the deformed Riemann-Hilbert problem adapted to large-$\alpha$. The results of such a cross-validation are presented in Table~\ref{table:rhp-validation}. For $X$ and $T$ small, the transition region addressed in Section~\ref{sec:Painleve} can be avoided and Riemann-Hilbert~\ref{rhp:limit-simpler} can be used instead to compute $\Psi^{+}(X,T)$ when $X$ is near $X_*(T)$.

\begin{table}[ht]
\begin{tabular}{r|c|c|c|c|}
 & $T$=0.2 & $T$=0.5 & $T$=1 & $T=3$ \\\hline
RHP~\ref{rhp:limit-simpler} and large-$X$ deformation & $1.43046\times 10^{-15}$ & $7.85046\times 10^{-17}$ & $1.66279\times 10^{-15}$ & $2.09448\times 10^{-14}$\\\hline
RHP~\ref{rhp:limit-simpler} and large-$T$ deformation  & $6.86276\times 10^{-15}$& $4.40781\times10^{-15}$ & $8.3037\times10^{-15}$ & $1.40576\times10^{-14}$ \\\hline
\end{tabular}
\caption{Difference between solutions computed numerically by different methods in the overlapping regions for small values of $X$ and $T$. We use $X= 1.15 X_*(T)$ when comparing the numerical solution of Riemann-Hilbert Problem~\ref{rhp:limit-simpler} with that of the large-$X$ deformed problem, and we use $X= 0.85 X_*(T)$ when comparing with the solution computed from the large-$T$ deformation.}
\label{table:rhp-validation}
\end{table}

With this validation in hand, to compute the special function $\Psi^+(X,T)$ at a fixed small value of $T\geq 0$, we solve Riemann-Hilbert Problem~\ref{rhp:limit-simpler} numerically when $X\leq X_*(T)$, but we switch to the numerical solution of the large-$X$ deformation when $X>X_*(T)$. The results of such computations allow us to display reliable graphs of $\Psi^+(X,T)$ for the first time.  See Figures~\ref{fig:psi-snapshots-1} and \ref{fig:psi-snapshots-2} for plots of $\Psi^{+}(X,T)$ for $-10\leq X \leq 10$ computed at various values of $T\in[0,2]$.
\begin{figure}[ht]
\begin{center}
\includegraphics{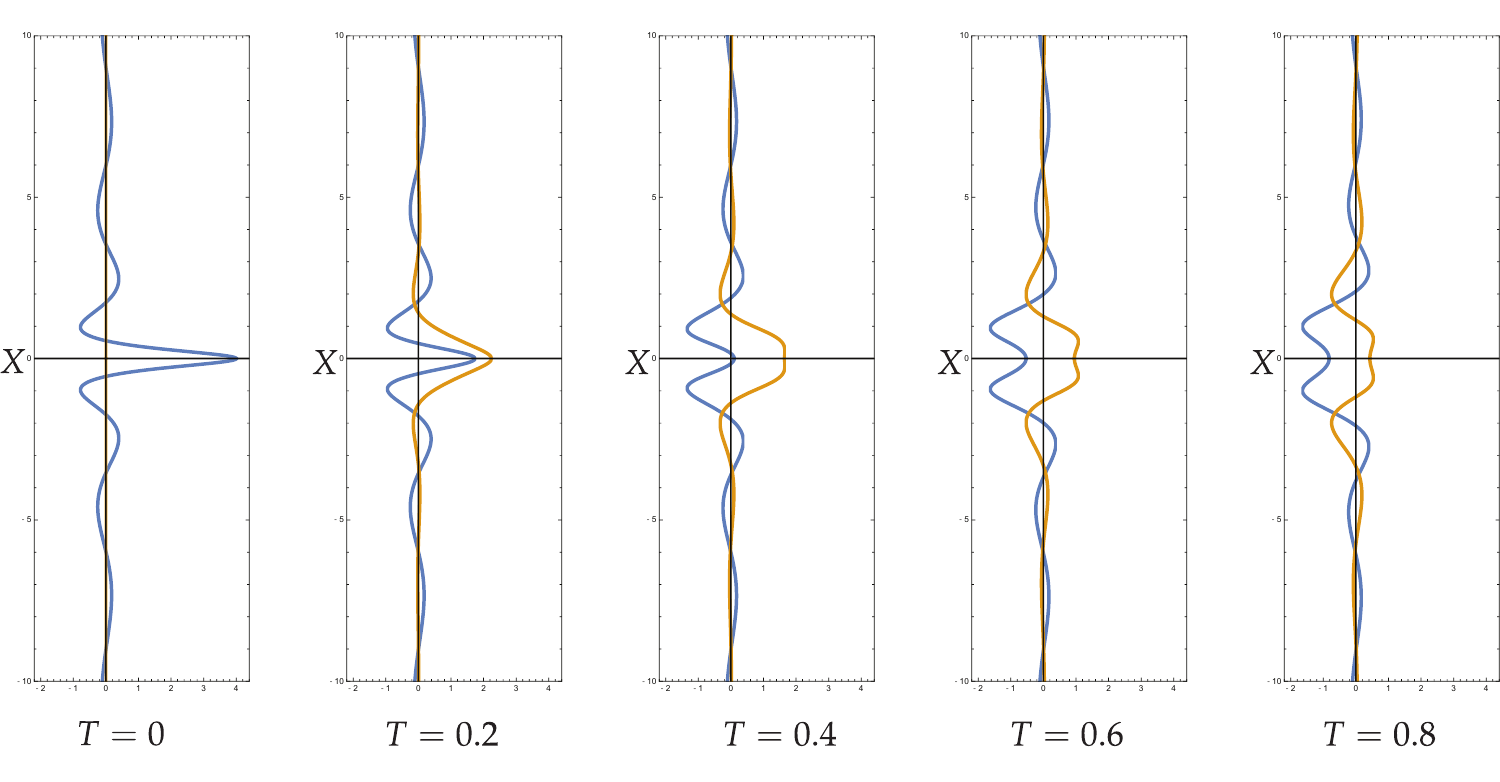}
\end{center}
\caption{Graphs of $\mathrm{Re}(\Psi^{+}(\cdot,T))$ (blue) and $\mathrm{Im}(\Psi^{+}(\cdot,T))$ (maize), from left to right for $T=0,0.2,0.4,0.6,0.8$.  In all plots the vertical axis measures 
$-10\leq X \leq 10$, and the horizontal axis measures the real and imaginary parts of $\Psi^+$. }
\label{fig:psi-snapshots-1}
\end{figure}
\begin{figure}[ht]
\begin{center}
\includegraphics{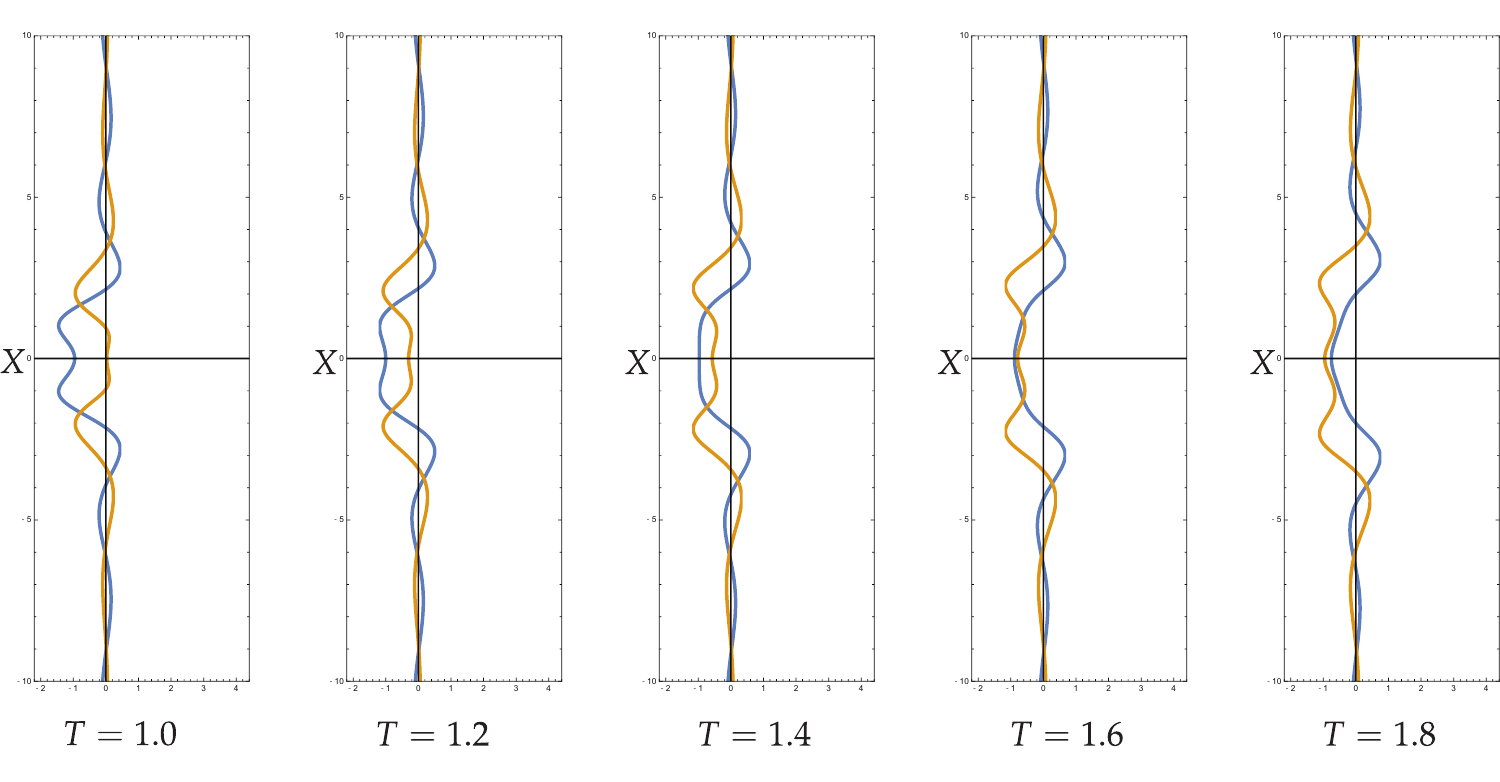}
\end{center}
\caption{As in Figure~\ref{fig:psi-snapshots-1} except for $T=1,1.2,1.4,1.6, 1.8$ left to right.}
\label{fig:psi-snapshots-2}
\end{figure}
A movie showing the evolution of $\Psi^{+}(X,T)$ from $T=0$ to $T=2$ can be found at \texttt{https://github.com/bilman/rogue-waves/blob/master/PsiTfrom0to2.gif} (see also the \texttt{rogue-waves} online repository\footnote{\texttt{https://github.com/bilman/rogue-waves}} for an \texttt{mpg} version.)

\subsection{Numerical validation of Theorem~\ref{theorem:main}}
\label{sec:plots-of-finite-vs-infinite-order}
The ability to reliably compute the special function $\Psi^+(X,T)$ at least for bounded $T$ allows us to illustrate the fundamental convergence result 
given in Theorem~\ref{theorem:main}.  We fix a compact subset of $\mathbb{R}^2$, $K:= [-2,2]\times[-2,2]$, and by evaluation on a suitably fine grid of values of $(X,T)\in K$, we compute
\begin{equation}
\mathcal{E}_{K}(n):= \sup_{(X,T)\in K} \left| \Psi^{+}(X,T) - \psi_{2n}(X n^{-1},T n^{-2})n^{-1}\right|
\end{equation}
for increasing values of $n$ chosen from the set $\{4, 8,10,16, 20, 25 \}$.   Here, $\Psi^+(X,T)$ is computed in the same manner as was used to make the plots in Figures~\ref{fig:psi-snapshots-1}--\ref{fig:psi-snapshots-2}, and $\psi_{2n}(x,t)$ is obtained from finite-dimensional linear algebra using a variant of Definition~\ref{def:rogue-wave}.  Note that by Proposition~\ref{prop:origin-value} and Proposition~\ref{prop:Psi-peak}, 
we have 
\begin{equation}
\left|\Psi^+(0,0)-\psi_{2n}(0,0)n^{-1}\right| = \frac{4n+1}{n} - 4=\frac{1}{n}.
\end{equation}
Therefore, as $(0,0)\in K$, the lower bound $\mathcal{E}_{K}(n)\geq n^{-1}$ must hold. Our numerical results show that this lower bound is the exact value of $\mathcal{E}_K(n)$, i.e., the maximum error over $K$ is achieved (at least) at the origin when the latter lies within $K$. In particular, the $O(n^{-1})$ error term in Theorem~\ref{theorem:main} is optimal.  Figure~\ref{fig:validate-convergence} shows a plot of $\ln(\mathcal{E}_{K}(n))$ versus $\ln(n)$. Performing a linear regression the data produces the best-fit line $\ln(\mathcal{E}_{K}(n))=9.19876*10^{-16} -  \ln(n)$ with the slope exactly equal to $-1$ and the intercept vanishing to machine precision.  The regression algorithm yields the $R^2$-value equal exactly to 1; this is the claimed numerical evidence that in fact $\mathcal{E}_K(n)=n^{-1}$ holds exactly for the indicated $K$ containing $(0,0)$.  

\begin{figure}[ht]
\begin{center}
\includegraphics{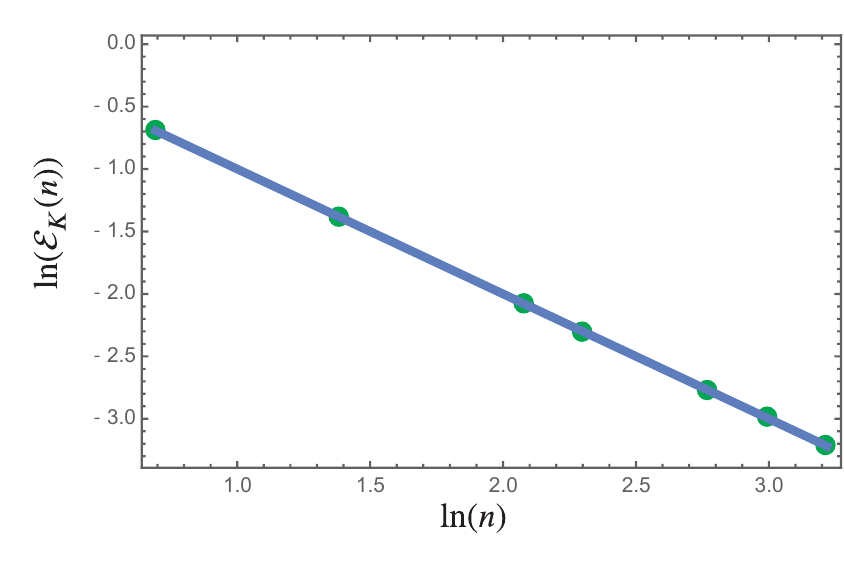}
\end{center}
\label{fig:validate-convergence}
\caption{A scatter plot of $\ln(\mathcal{E}_{K}(n))$ versus $\ln(n)$ for $n\in\{ 4, 8, 10, 16, 20, 25 \}$ (green dots) and the best-fit line (blue).}
\end{figure}

\subsection{Numerical validation of Theorem~\ref{theorem:large-X} and Corollary~\ref{corollary:large-X-T0}}
\label{sec:plots-large-X}
Since the numerical computation of the special function $\Psi^+(X,T)$ is reliable when $X>0$ is large and the parameter $v=TX^{-3/2}$ is sufficiently small, we can also illustrate the accuracy of the asymptotic results developed in Section~\ref{sec:large-X}.  
For notational convenience we let $L^{[X]}(X,v)$ denote the leading term in the asymptotic formula \eqref{eq:Psi-large-X} (i.e., the sum of the explicit terms on the first line of the right-hand side). Below we display plots and regression data for verification of the results in Theorem~\ref{theorem:large-X} and Corollary~\ref{corollary:large-X-T0}. We first fix $v=0.05$. The plots in Figure~\ref{fig:validate-large-X} compare real and imaginary parts of $L^{[X]}(X,v=0.05)$ and $\Psi^{+}(X,T)$, $v=0.05=TX^{-3/2}$, where $\Psi^+(X,T)$ is computed numerically using the large-$X$ deformation method. The graphs of the real and imaginary parts of $L^{[X]}(X,v=0.05)$ are plotted along with shaded strips centered on the graphs and having width $X^{-1}$ which is the size of the error term predicted in the formula \eqref{eq:Psi-large-X}. Superimposed in thicker dashed curves are the corresponding graphs of numerical computation of $\Psi^+(X,T)$, which not only lie within the strips but are indistinguishable to the eye from the predicted limits.
This is a striking illustration of the accuracy of Theorem~\ref{theorem:large-X}.
\begin{figure}[ht]
\begin{center}
\includegraphics{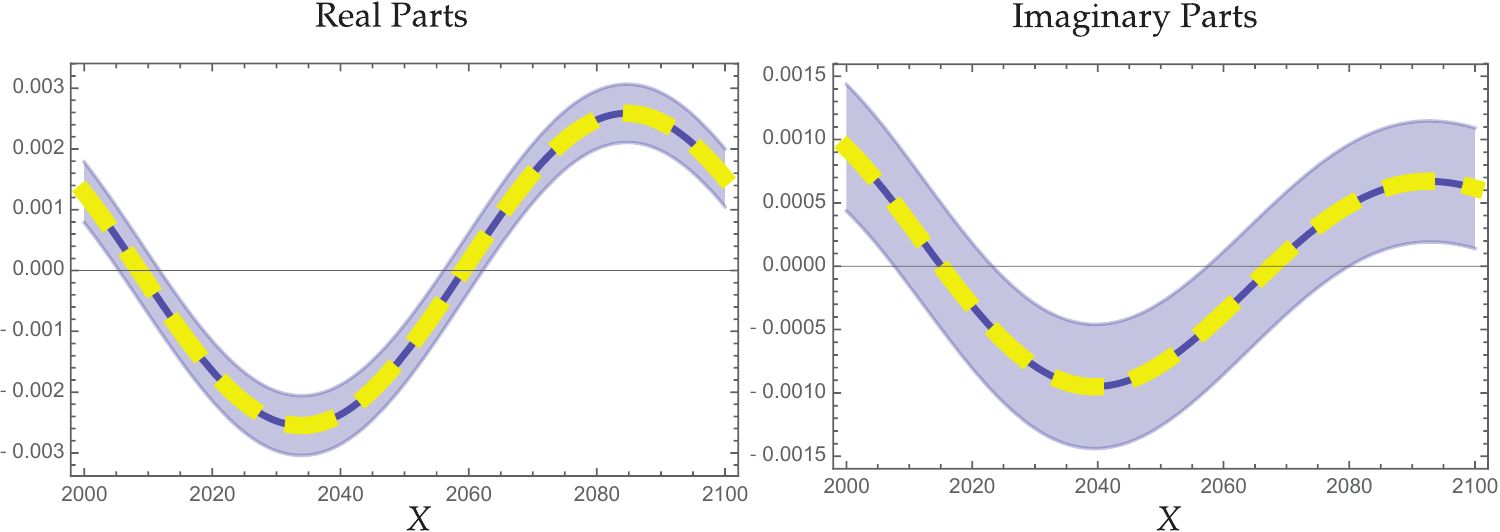}
\end{center}
\caption{Numerically computed solution $\Psi^{+}(X,T)$ (dashed yellow) and the leading term $L^{[X]}(X,v=0.05)$ (solid blue, centered in a shaded strip of width $X^{-1}$) of the asymptotic formula \eqref{eq:Psi-large-X}, plotted over $2000\leq X\leq 2100$ on the horizontal axis. Left: real parts.  Right: imaginary parts. }
\label{fig:validate-large-X}
\end{figure}
To illustrate Corollary~\ref{corollary:large-X-T0}, we set $v=0$ (see \eqref{eq:Psi-large-X-T0} for $L^{[X]}(X,0)$) in which case $\Psi^{+}(X,0)$ is real-valued. Figure~\ref{fig:validate-large-X-T0} shows a similar comparison of $L^{[X]}(X,0)$ and $\Psi^{+}(X,0)$.
\begin{figure}[ht]
\begin{center}
\includegraphics{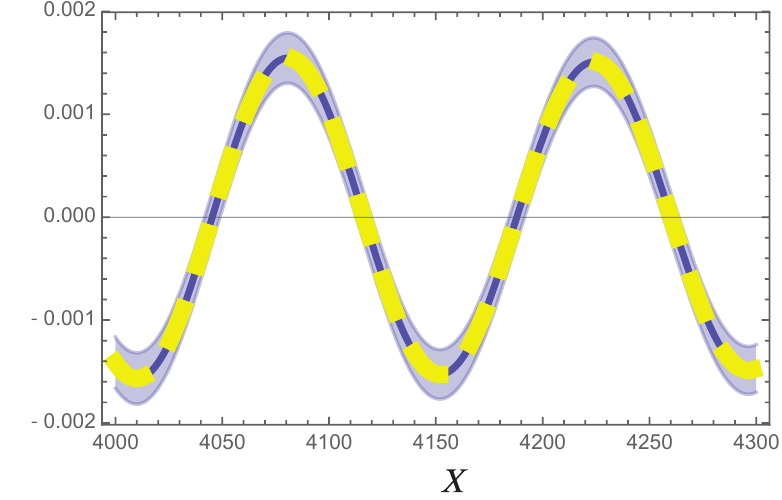}
\end{center}
\caption{As in Figure~\ref{fig:validate-large-X}, but with $v=0$ so all quantities are real-valued and the simpler asymptotic formula for $L^{[X]}(X,0)$ from Corollary~\ref{corollary:large-X-T0} can be used.  Here the plot range is 
 $4000\leq X\leq 4300$ on the horizontal axis.}
\label{fig:validate-large-X-T0}
\end{figure}
Finally, we use the data from the latter experiment to numerically recover the exponent in the error term in \eqref{eq:Psi-large-X-T0}. This is done by plotting $\ln(|\Psi^{+}(X,0) - L^{[X]}(X,0)|)$ versus $\ln(X)$ and performing linear regression, which yields the best-fit line $\ln(|\Psi^{+}(X,0) - L^{[X]}(X,0)|) = -4.82864 - 1.05592 \ln(X)$; see Figure~\ref{fig:largeX-loglog-plot}. The slope of this line gives as desired approximately the exponent of $-1$ as predicted in the error terms in the formul\ae\ \eqref{eq:Psi-large-X} and \eqref{eq:Psi-large-X-T0}. 
\begin{figure}[ht]
\begin{center}
\includegraphics{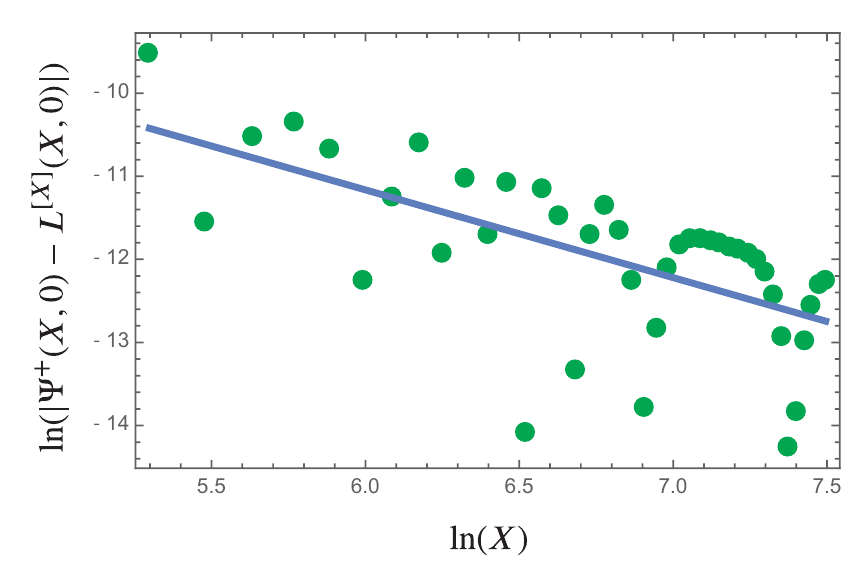}
\end{center}
\caption{A scatter plot of $\ln(|\Psi^{+}(X,0) - L^{[X]}(X,0)|)$ (vertical axis) versus $\ln(X)$ (horizontal axis) for $X=200+40(k-1)$, $k=0,1,\dots,41$ (green dots) and the best-fit line (blue). }
\label{fig:largeX-loglog-plot}
\end{figure}

\subsection{A larger domain of convergence for the near-field limit of rogue waves}

Recall that Theorem~\ref{theorem:main} establishes the locally uniform convergence of rescaled rogue waves of order $k=2n$, $\psi_k(X n^{-1}, T n^{-2})n^{-1}$, to the rogue wave of infinite order $\Psi^{+}(X,T)$, with an accuracy proportional to $n^{-1}$ (and similar convergence to $\Psi^{-}(X,T)$ if $k=2n-1$). Here we investigate whether this convergence might be valid on a larger domain in the $(X,T)$-plane that expands as $n\geq 0$ grows at a suitable rate, possibly with a reduced rate of decay of the error.  For simplicity, we restrict our study here to convergence along the $X$-axis ($T=0$) and $T$-axis ($X=0$).

\subsubsection{Restriction to $T=0$} Note that the size of the leading term $L^{[X]}(X,v=0)$ in the large-$X$ asymptotic formula \eqref{eq:Psi-large-X-T0} is proportional to $X^{-3/4}$ as $X\to+\infty$. Therefore, we shall study the \emph{relative} error between the leading term of \eqref{eq:Psi-large-X-T0} and the rescaled rogue wave of order $2n$ defined as 
\begin{equation}
\mathcal{R}_n^{[X]}(X):= \left|L^{[X]}(X,0) - \psi_{2n}(Xn^{-1},0)n^{-1} \right|X^{3/4}.
\label{eq:RnX-def}
\end{equation}
In Figure~\ref{fig:relative-errors} we plot $\mathcal{R}^{[X]}_n(X)$ over the domain $0\leq X \leq 2n^{8/5}$ for increasing values of $n$ chosen from the set $\{4,8,10,16,20,25\}$. We compute $X_1(n):= \min \{ X\colon \mathcal{R}^{[X]}_n(X)=1,\, 1\leq X \leq 2n^{8/5}\}$, the smallest value of $X$ at which $\mathcal{R}^{[X]}_n(X)=1$ in the interval $1\leq X \leq 2n^{8/5}$. See 
Figure~\ref{fig:relative-errors}. 
\begin{figure}[ht]
\begin{center}
\includegraphics{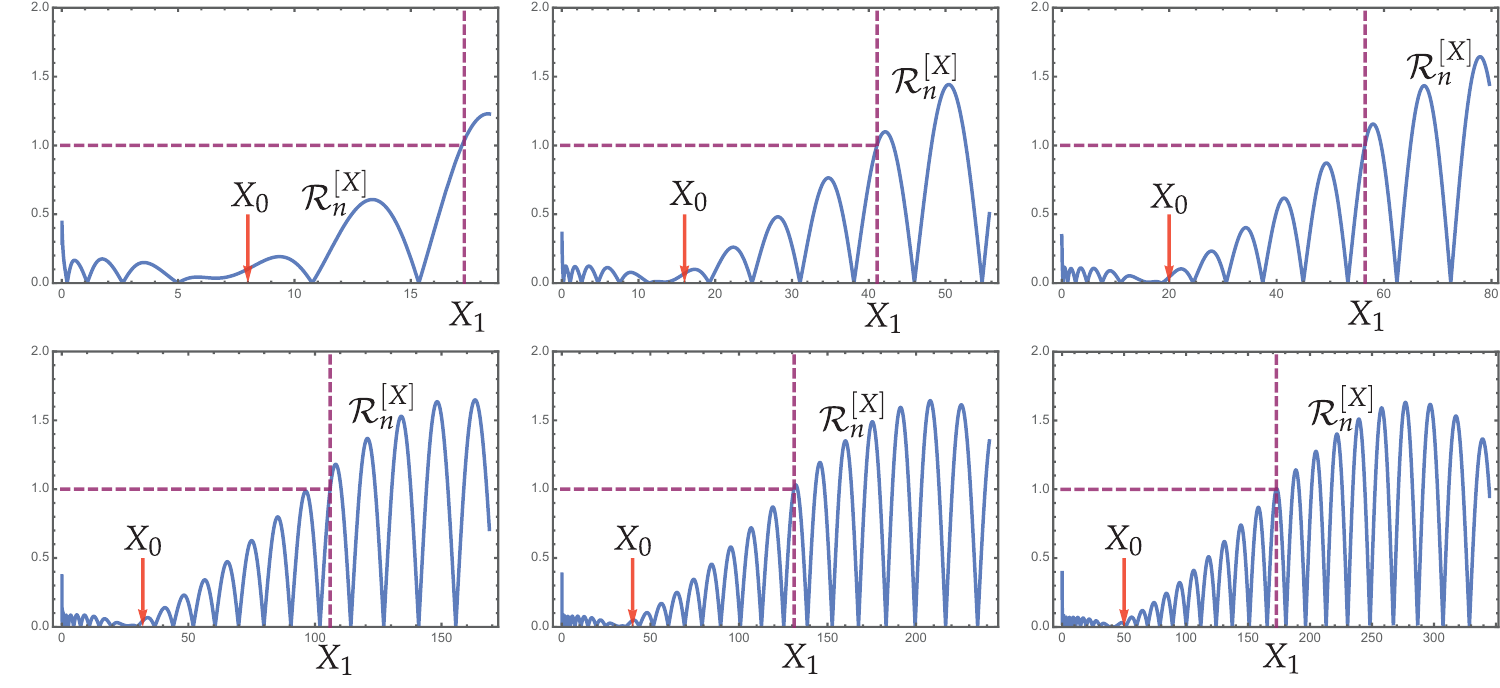}
\end{center}
\caption{Plots of $\mathcal{R}^{[X]}_n(X)$ versus $X$ for values of $n$ from left to right $n=4,\, 8,\,10$ in the first row, $n=16,\, 20,\,25$ in the second row. $0\leq X\leq 2n^{8/5}$ on the horizontal axis for each plot. The values of $X_0(n)$ and $X_1(n)$ are indicated on each plot. }
\label{fig:relative-errors}
\end{figure}
We deduce numerically that $X_1(n)$ obeys a power law as $n$ grows. To see this, we perform a linear regression on the data set $\ln(X_1(n))$ versus $\ln(n)$. See Figure~\ref{fig:overlap-T0-regression} for a plot comparing the data with the best-fit regression line
given by $\ln(X_1(n)) = 1.10384 +1.26525 \ln(n)$ with the $R^2$-value by $0.9988$. The slope of this line indicates that $X_1(n)$ grows roughly as $n^{5/4}$ as $n$ becomes large.
\begin{figure}[ht]
\begin{center}
\includegraphics{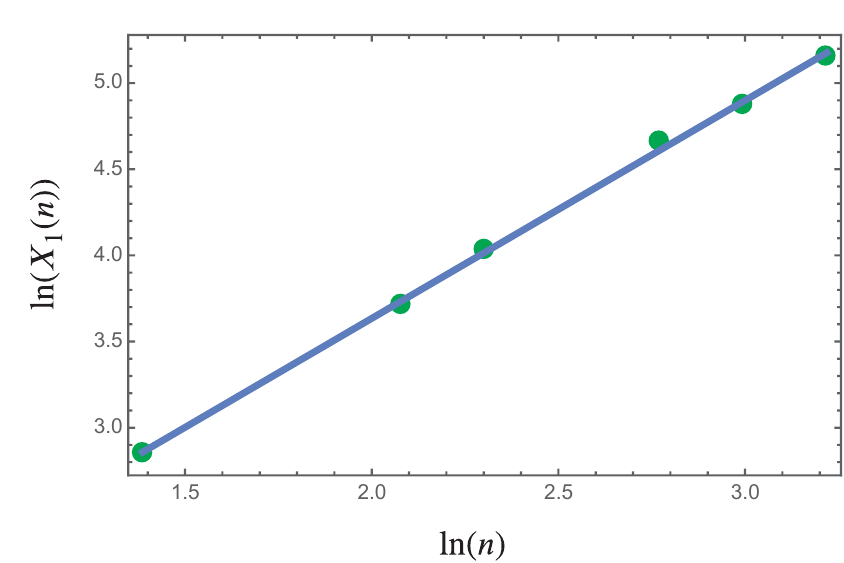}
\end{center}
\caption{Scatter plot of the data set $\ln(X_1(n))$ versus $\ln(n)$ (green dots) for $n\in\{4,8,10,16, 20,25\}$ and the best-fit line (blue).}
\label{fig:overlap-T0-regression}
\end{figure}

This analysis suggests that the right endpoint of an interval on which the rescaled rogue wave of order $2n$ is approximated accurately at $T=t=0$ by the rogue wave of infinite order $\Psi^+(X,0)$ can grow with $n$ at most at a rate that is slower than $n^{5/4}$. Thus we are led to study the relative error $\mathcal{R}^{[X]}_n(X)$ on intervals with right endpoints of the form
\begin{equation}
X^\text{R}_n(q):= C n^{5/4 - q},\quad q>0,
\end{equation}
where the constant $C\approx \ee^{1.1}$ is taken from the regression analysis so that $X^\text{R}_n(0)\approx X_1(n)$. On the other hand, as can be seen from the plots in Figure~\ref{fig:relative-errors}, there exists a region around $X=0$ on which $\mathcal{R}^{[X]}_n(X)$ is not small when $n$ is large. The reason for this is that the asymptotic formula $L^{[X]}(X,0)$ appearing in the definition of $\mathcal{R}^{[X]}_n(X)$ is a poor approximation\footnote{It would be better to use $\Psi^+(X,0)$ itself in place of $L^{[X]}(X,0)$ in \eqref{eq:RnX-def} but it would also be more computationally expensive and simultaneously less relevant to the present study which of course is concerned with investigating accuracy of the infinite-order rogue wave approximation precisely where Theorem~\ref{theorem:main} makes no prediction.} of $\Psi^{+}(X,0)$ unless $X$ is large. We observe numerically that the boundary of this region, denoted by $X=X_0(n)$, may be taken as $X_0(n):=2 n$, which as shown in the plots in Figure~\ref{fig:relative-errors} is approximately where $\mathcal{R}^{[X]}_n(X)$ seems to be minimal on $[0,X_1(n)]$. 

Therefore, we expect to see convergence of $\psi_{2n}(X n^{-1},0)n^{-1}$ to $\Psi^{+}(X,0)$ on the interval $|X|\le C n^{5/4 - q}$ as $n\to\infty$, for any $q>0$. For $0<q<5/4$ this interval expands as $n$ grows rather than remaining bounded  as in the premise of Theorem~\ref{theorem:main}. To be able to verify such convergence on this set numerically and make use of the large-$X$ asymptotic formula $L^{[X]}(X,0)$ to approximate $\Psi^+(X,0)$, we need to make sure to consider values of $X>X_0(n)$ for each $n$ and hence we study the quantity
\begin{equation}
\bar{\mathcal{R}}^{[X]}_n(q) := \sup_{X_0(n)\leq X \leq X^\text{R}_n(q)}| \mathcal{R}^{[X]}_n(X) |,\quad 0< q<1/4,
\end{equation} 
as $n$ becomes large (the upper bound on $q$ here ensures that the interval is not empty). See Figure~\ref{fig:different-qs-T0}, where we plot the data set $\ln(\bar{\mathcal{R}}^{[X]}_n(q))$ versus $\ln(n)$ to illustrate the convergence as $n$ grows for different values of $0< q<1/4$. Observe that for $q\in(0,1/4)$ fixed, linear regression predicts a power-law relationship between $\bar{\mathcal{R}}^{[X]}_n(q)$ and $n$ as $n\to\infty$. These results strongly suggest that the infinite-order rogue wave $\Psi^+(X,0)$ accurately approximates the rescaled finite-order rogue wave $\psi_{2n}(Xn^{-1},0)n^{-1}$ in the large-$n$ limit provided that $|X|\ll n^{5/4}$.

\begin{figure}[ht]
\begin{center}
\includegraphics{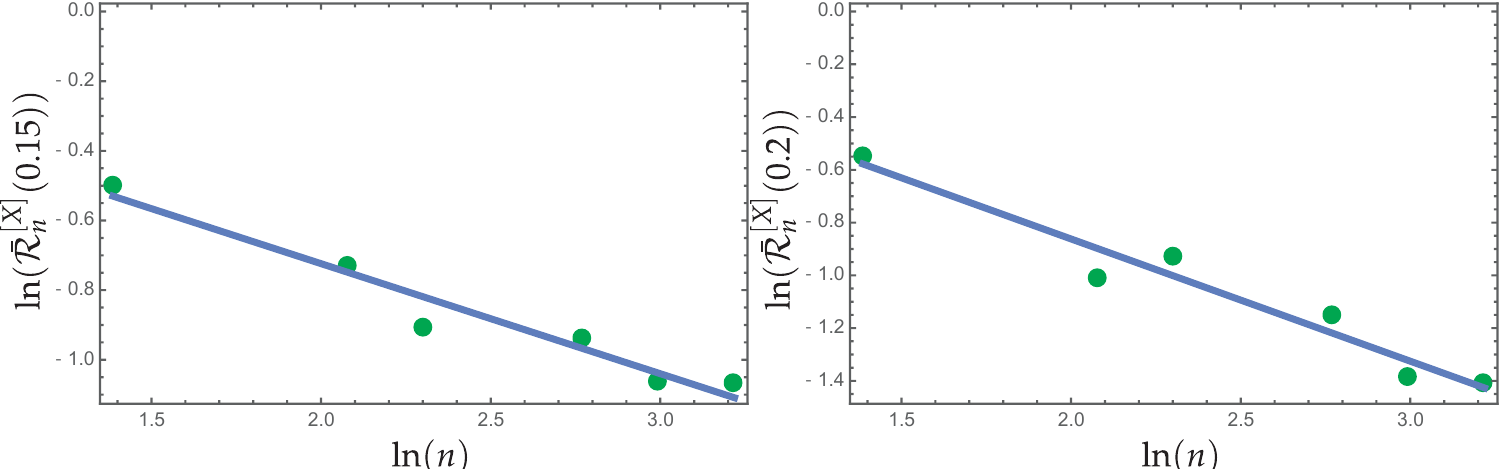}
\end{center}
\caption{Scatter plot of the data set $\ln(\bar{\mathcal{R}}^{[X]}_n(q))$ versus $\ln(n)$ for $n\in\{4,8,10,16,20,25\}$ (green dots) and the best-fit line (blue). Here $q=0.15$ (left panel) and $q=0.2$ (right panel). Note the difference in scales of the vertical axes in each panel.   
For $q=0.15$, the best-fit line is has intercept: $-0.0919092$, slope: $- 0.315997$, and $R^2=0.948347$. For $q=0.2$, the best-fit line has intercept: $0.064779$, slope: $-0.463413$, and $R^2= 0.946802$.}
\label{fig:different-qs-T0}
\end{figure}

\subsubsection{Restriction to $X=0$} We perform the same type of analysis when $X=0$, using the leading term $L^{[T]}(0,T)$ in the asymptotic formula \eqref{eq:Psi-large-T-X0} whose size is proportional to $T^{1/3}$ as $T\to+\infty$. Thus we consider the relative error
\begin{equation}
\mathcal{R}^{[T]}_n(T):= \left|L^{[T]}(0,T) - \psi_{2n}(0,T n^{-2})n^{-1} \right|T^{1/3},
\end{equation}
and plot $\mathcal{R}^{[T]}_n(X)$ over the domain $0\leq T \leq 2n^{11/5}$ for increasing values of $n$ chosen from $\{4,8,10,16,20,25\}$. As in the case on the $X$-axis, we compute $T_1(n):= \min \{T\colon \mathcal{R}^{[T]}_n(T)=1,\, 1\leq T \leq 2n^{11/5}\} $, the smallest value of $T$ at which $\mathcal{R}^{[T]}_n(T)=1$ in the interval $1\leq T \leq 2n^{11/5}$. See Figure~\ref{fig:relative-errors-X0}. 

\begin{figure}[ht]
\begin{center}
\includegraphics{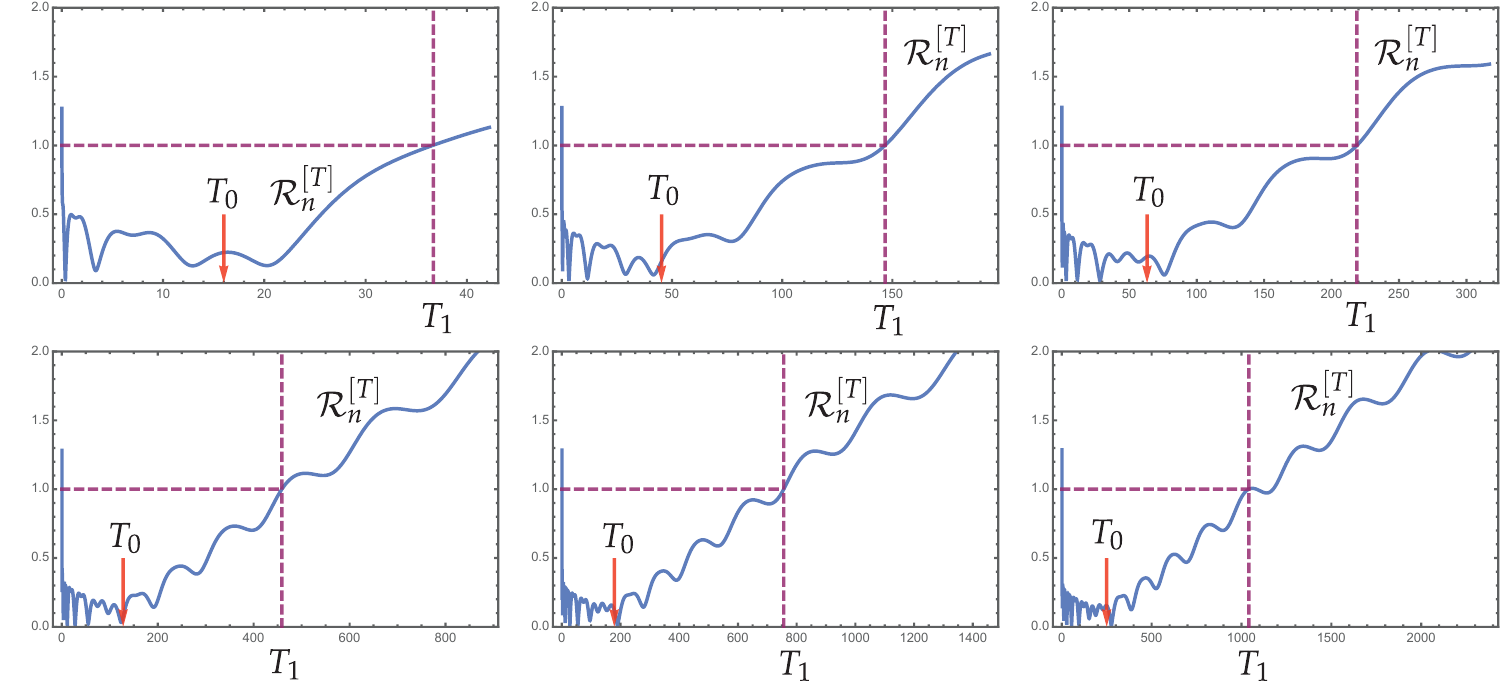}
\end{center}
\caption{Plots of $\mathcal{R}^{[T]}_n(T)$ versus $T$ for values of $n$ from left to right $n=4,\, 8,\,10$ in the first row, $n=16,\, 20,\,25$ in the second row. In each plot, the horizontal axis is the interval $0\leq T\leq 2n^{11/5}$. The values of $T_0(n):=2n^{3/2}$ and $T_1(n)$ 
are indicated on each plot. }
\label{fig:relative-errors-X0}
\end{figure}

Again by a linear regression performed on the data set $\ln(T_1(n))$ versus $\ln(n)$, we observe that $T_1(n)$ obeys a power law as $n$ grows. See Figure~\ref{fig:overlap-X0-regression} for the $\log$-$\log$ plot and the best-fit line predicted by the regression which is $\ln(T_1(n))=1.14431 + 1.81749 \ln(n)$. The slope of this line indicates that $T_1(n)$ grows proportional to roughly $n^{9/5}$ as $n$ becomes large. 

\begin{figure}[ht]
\begin{center}
\includegraphics{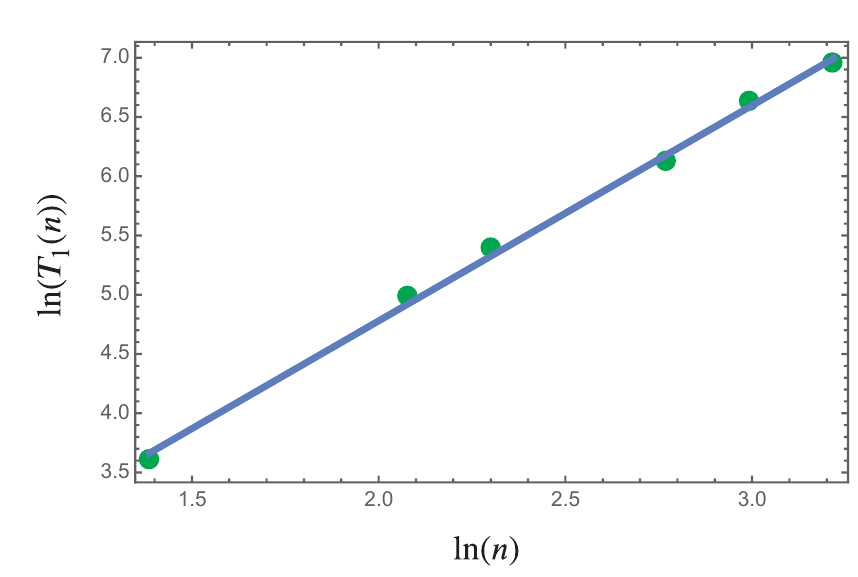}
\end{center}
\caption{Scatter plot of the data set $\ln(T_1(n))$ versus $\ln(n)$ for $n\in\{4, 8,10,16, 20,25\}$ (green dots) and the best-fit line $\ln(T_1(n)) = \ln(T_1(n))=1.14431 + 1.81749 \ln(n) $ (blue). }
\label{fig:overlap-X0-regression}
\end{figure}

Thus, we expect that to have convergence of the rescaled rogue wave to the rogue wave of infinite order, it is necessary to consider values of $T$ small compared to $T_1(n)\sim n^{9/5}$.  Consequently we study the relative error $\mathcal{R}^{[T]}_n(T)$ on intervals of $T$ with right endpoint given by
\begin{equation}
T^\text{R}_n(q):= C n^{9/5 - q},\quad q>0,
\end{equation}
where the constant $C\approx \ee^{1.14}$ is taken from the regression analysis so that $T^\text{R}_n(0) \approx T_1(n)$. Accordingly, we expect to see convergence of $\psi_{2n}(0,Tn^{-2})n^{-1}$ to $\Psi^{+}(0,T)$ on the interval $|T|\le C n^{9/5 - q}$ as $n\to\infty$, for any $q>0$. This is again a growing set in $n$ for $0<q<9/5$ rather than fixed as in the premise of Theorem~\ref{theorem:main}. To verify such convergence while still exploiting the simplicity of the large-$T$ asymptotic formula $L^{[T]}(0,T)$ in place of $\Psi^+(0,T)$, we need to consider $\mathcal{R}^{[T]}_n(T)$ for $T$ away from the region surrounding $T=0$ on which the leading term $L^{[T]}(0,T)$ given in \eqref{eq:Psi-large-T-X0} is a poor approximation of $\Psi^{+}(0,T)$. From the plots in Figure~\ref{fig:relative-errors-X0} we infer that the boundary of this region expands as $n^{3/2}$ as $n$ becomes large.  The points $T=T_0(n) := 2 n^{3/2}$, which roughly mark the boundary of this region are indicated with orange arrows in Figure~\ref{fig:relative-errors-X0}. Thus we numerically study
\begin{equation}
\bar{\mathcal{R}}^{[T]}_n(q) := \sup_{T_0(n)\leq T \leq T^\text{R}_n(q)}| \mathcal{R}^{[T]}_n(T) |,\quad  0< q<3/10,
\end{equation} 
for different values of $q$ in the range given above.  We indeed observe convergence as $n\to\infty$, illustrating that the rogue wave of infinite order is a good approximation for rescaled rogue waves provided $|T|\ll n^{9/5}$.

\begin{figure}[ht]
\begin{center}
\includegraphics{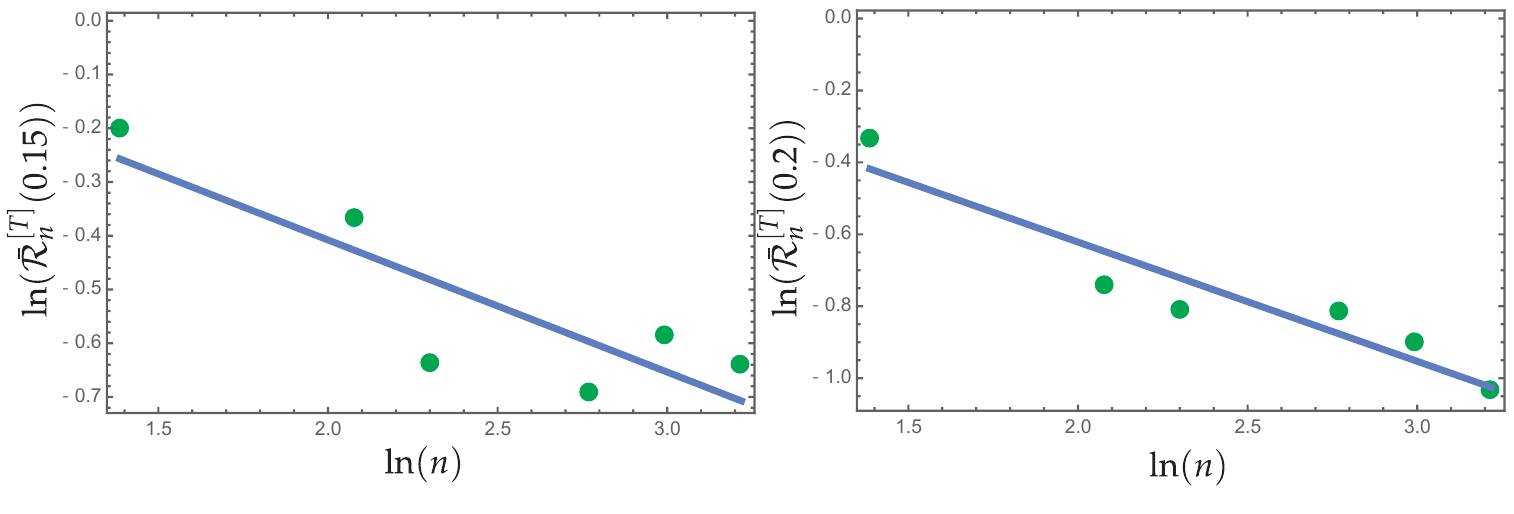}
\end{center}
\caption{Scatter plot of the data set $\ln(\bar{\mathcal{R}}^{[T]}_n(q))$ versus $\ln(n)$ for $n\in\{4,8,10,16,20,25\}$ (green dots) and the best-fit line (blue). Left:  $q=0.15$. Right:  $q=0.2$. Note the difference in scales of the vertical axes in each panel.%
For $q=0.15$, the best-fit line has intercept: $0.0838235$, slope: $-0.245828$, and $R^2=0.736517$. For $q=0.2$, the best-fit line has intercept: $0.0413682$, slope: $-0.331435$, and $R^2=0.890735$.}
\label{fig:different-qs-T0}
\end{figure}

\subsubsection{Conjecture formulation}
The preceding numerical results suggest the following, which is a generalization of Theorem~\ref{theorem:main}.
\begin{conjecture}[Generalized accuracy of the near-field limit]
Let $\{K_n\subset\mathbb{R}^2\}_{n=1}^\infty$ be a sequence of compact sets with the property that 
\begin{equation}
\lim_{n\to\infty}n^{-5/4}\sup_{(X,T)\in K_n}|X|=0\quad\text{and}\quad
\lim_{n\to\infty}n^{-9/5}\sup_{(X,T)\in K_n}|T|=0.
\end{equation}
Then 
\begin{equation}
\lim_{n\to\infty}\sup_{(X,T)\in K_n}\left|n^{-1}\psi_{2n}(n^{-1}X,n^{-2}T)-\Psi^+(X,T)\right|=0
\end{equation}
and
\begin{equation}
\lim_{n\to\infty}\sup_{(X,T)\in K_n}\left|n^{-1}\psi_{2n-1}(n^{-1}X,n^{-2}T)-\Psi^-(X,T)\right|=0.
\end{equation}
\label{conjecture:numerical}
\end{conjecture}
This conjecture therefore asserts uniform convergence of rescaled finite-order rogue waves to their infinite-order near-field limits on sets that can grow with $n$ in the $(X,T)$-plane but with a diameter small compared to $n^{5/4}$ in the $X$-direction and a diameter small compared to $n^{9/5}$ in the $T$-direction.  It is useful to express these scales in terms of the original independent variables $(x,t)$ of the finite-order rogue wave:  $X=o(n^{5/4})$ is equivalent to $x=o(n^{1/4})$ and $T=o(n^{9/5})$ is equivalent to $t=o(n^{-1/5})$.  Thus, the conjecture asserts that the near-field limit of high-order rogue waves is in fact valid on $x$-intervals whose diameter grows at a suitable rate with $n$, but $t$ is still required to be small compared to $n^{-1/5}$.

Proving this conjecture by generalizing the method of proof of Theorem~\ref{theorem:main} would require control of the solution of Riemann-Hilbert Problem~\ref{rhp:limit-simpler} for parameters $(X,T)$ in growing sets $K_n$.  The large-$X$, large-$T$, and transitional asymptotics described in Section~\ref{sec:Psi-asymptotic} may be exploited in this effort, but more work would be needed because the three asymptotic regimes themselves do not quite overlap.  Thus we leave any proof of Conjecture~\ref{conjecture:numerical} to the future.

\appendix
\section*{Appendix.  Equivalence of the Riemann-Hilbert and Determinantal Representations of Fundamental Rogue Waves}
\label{app:Equivalence}
In this appendix, we prove Proposition~\ref{prop:Equivalence}.  First, note that since Riemann-Hilbert Problems~\ref{rhp:rogue-wave} and \ref{rhp:renormalized} are equivalent due to \eqref{eq:N-M-relation} in which $\mathbf{M}^{(0)}(\lambda;x,t)$ is explicitly given by \eqref{eq:M-zero} and is a matrix with unit determinant, it suffices to deduce the determinantal representation \eqref{eq:psi-k-determinants} of $\psi_k(x,t)$ from Riemann-Hilbert Problem~\ref{rhp:renormalized}.  We will show that the latter problem can be solved by a suitable rational ansatz that is based on the theory of generalized Darboux transformations \cite{GuoLL12}.  Just as Riemann-Hilbert Problem~\ref{rhp:renormalized} splits into two cases depending on whether $k$ is even or odd, so does the choice of rational ansatz.  We shall consider the two cases separately.

\subsection*{The rational ansatz for $k=2n$ with $n\in\mathbb{Z}_{\ge 0}$}
We assume that the matrix $\mathbf{N}^{(k)}(\lambda;x,t)$ can be represented for $\lambda$ in the exterior of $\Sigma_\circ$ by an expression of the form
\begin{equation}
\mathbf{N}^{(k)}(\lambda;x,t)=\left(\frac{\lambda+\ii}{\lambda-\ii}\right)^n\left(\mathbb{I}-\mathbf{Y}(x,t)\mathbf{D}(\lambda)\mathbf{X}(x,t)^\top\right),\quad\text{$\lambda$ in the exterior of $\Sigma_\circ$,}
\label{eq:rational-ansatz}
\end{equation}
where $\mathbf{X}(x,t)$ and $\mathbf{Y}(x,t)$ are $2\times k$ matrices to be determined, and where $\mathbf{D}(\lambda)$ is the $k\times k$ matrix
\begin{equation}
\mathbf{D}(\lambda):=\begin{bmatrix}(\lambda+\ii)^{-1} &0 &0&0 & \cdots &0\\
(\lambda+\ii)^{-2} & (\lambda+\ii)^{-1} & 0 &0& \cdots & 0\\
(\lambda+\ii)^{-3}&(\lambda+\ii)^{-2} & (\lambda+\ii)^{-1} &0&\cdots & 0\\
\vdots & \ddots & \ddots &\ddots & \cdots & 0\\
(\lambda+\ii)^{-k} & (\lambda+\ii)^{-k+1}& (\lambda+\ii)^{-k+2} & (\lambda+\ii)^{-k+3} & \cdots &(\lambda+\ii)^{-1}\end{bmatrix}.
\label{eq:D-matrix-define}
\end{equation} 
In other words, $\mathbf{D}(\lambda)$ is the Toeplitz matrix of dimension $k\times k$ with symbol $d(z;\lambda):=(\lambda+\ii-z^{-1})^{-1}$, so that if $|\lambda+\ii|<1$, the Fourier coefficients $d_m(\lambda)$ are given by
\begin{equation}
d_m(\lambda):=\frac{1}{2\pi}\int_{-\pi}^\pi d(\ee^{\ii\theta};\lambda)\ee^{-\ii m\theta}\,\dd\theta = \frac{1}{2\pi\ii}\oint_{|z|=1}d(z;\lambda)z^{-m-1}\,\dd z = (\lambda+\ii)^{-m-1}\delta_{m\ge 0},
\end{equation}
and hence $D_{ij}(\lambda)=d_{i-j}(\lambda)$.  Without further conditions on $\mathbf{X}(x,t)$ and $\mathbf{Y}(x,t)$, we see that  $\mathbf{N}^{(k)}(\lambda;x,t)\to\mathbb{I}$ as $\lambda\to\infty$, and (since the circle $\Sigma_\circ$ encloses the points $\lambda=\pm\ii$) is analytic in the exterior domain.  Therefore, it only remains to define $\mathbf{N}^{(k)}(\lambda;x,t)$ as an analytic function in the \emph{interior} domain in such a way that the jump condition \eqref{eq:N-jump-even} relevant for  $k=2n$ even holds on $\Sigma_\circ$.  Since $\mathbf{U}(\lambda;x,t)$ defined by \eqref{eq:entire} is an entire function of $\lambda$ with unit determinant, the jump condition \eqref{eq:N-jump-even} or \eqref{eq:N-jump-odd} actually serves to define $\mathbf{N}^{(k)}(\lambda;x,t)$ in the interior domain by analytic continuation, except possibly for the points $\lambda=\pm\ii$ which are pole singularities of both the exterior domain ansatz \eqref{eq:rational-ansatz} and the jump matrix.  Therefore, the ansatz \eqref{eq:rational-ansatz} gives the (necessarily unique) solution of Riemann-Hilbert Problem~\ref{rhp:renormalized} in the exterior domain for $k=2n$, provided that $\mathbf{X}(x,t)$ and $\mathbf{Y}(x,t)$ are chosen so that
\begin{equation}
\left(\frac{\lambda+\ii}{\lambda-\ii}\right)^n\left(\mathbb{I}-\mathbf{Y}(x,t)\mathbf{D}(\lambda)\mathbf{X}(x,t)^\top\right)\mathbf{U}(\lambda;x,t)\mathbf{Q}\left(\frac{\lambda+\ii}{\lambda-\ii}\right)^{n\sigma_3}\mathbf{Q}^{-1}\mathbf{U}(\lambda;x,t)^{-1}\quad
\text{is analytic at $\lambda=\pm\ii$.}
\end{equation}
The final factors of $\mathbf{Q}^{-1}\mathbf{U}(\lambda;x,t)^{-1}$ are analytic and invertible near $\lambda=\pm\ii$ in particular, so they may be omitted.  The product of remaining terms in general has poles of order $n$ at each of the two points $\lambda=\pm\ii$.  Since the factor $((\lambda+\ii)/(\lambda-\ii))^{n\sigma_3}$ now appears as the right-most term in the product, we may easily formulate the conditions on $\mathbf{X}(x,t)$ and $\mathbf{Y}(x,t)$ necessary to remove the singularities as conditions on column vectors.  Writing $\mathbf{Q}=[\mathbf{q}^{(1)};\mathbf{q}^{(2)}]$, these conditions now read as follows:  
\begin{equation}
\left(\mathbb{I}-\mathbf{Y}(x,t)\mathbf{D}(\lambda)\mathbf{X}(x,t)^\top\right)\mathbf{U}(\lambda;x,t)\mathbf{q}^{(2)}\quad\text{should be analytic at $\lambda=-\ii$, and}
\label{eq:k-even-singularity}
\end{equation}
\begin{equation}
\left(\mathbb{I}-\mathbf{Y}(x,t)\mathbf{D}(\lambda)\mathbf{X}(x,t)^\top\right)\mathbf{U}(\lambda;x,t)\mathbf{q}^{(1)}\quad\text{should vanish to order $k$ at $\lambda=\ii$.}
\label{eq:k-even-kernel}
\end{equation}
In the literature on Darboux transformations, the conditions \eqref{eq:k-even-singularity} 
are frequently called \emph{residue conditions} while the conditions \eqref{eq:k-even-kernel} 
are frequently called \emph{kernel conditions}, which are imposed on a \emph{Darboux matrix} of the form $\mathbb{I}-\mathbf{Y}\mathbf{D}\mathbf{X}^\top$.

A key observation in the theory of Darboux transformations is that the residue conditions \eqref{eq:k-even-singularity}
 can be satisfied by making an explicit choice of the $2\times k$ matrix $\mathbf{X}(x,t)$.  Firstly, we observe that since $\mathbf{U}(\lambda;x,t)$ is analytic at $\lambda=-\ii$, any potential singularities come from the second term $-\mathbf{Y}(x,t)\mathbf{D}(\lambda)\mathbf{X}(x,t)^\top$, and for satisfying these conditions it is sufficient to drop the $\lambda$-independent prefactor of $\mathbf{Y}(x,t)$.  Thus, we will replace \eqref{eq:k-even-singularity} by the \emph{sufficient} condition
\begin{equation}
\mathbf{D}(\lambda)\mathbf{X}(x,t)^\top\mathbf{U}(\lambda;x,t)\mathbf{q}^{(2)}\quad\text{should be analytic at $\lambda=-\ii$}
\label{eq:sufficient-condition}
\end{equation}
which, taking separately into account the $k$ rows of $\mathbf{D}(\lambda)$ and introducing the notation 
\begin{equation}
\mathbf{X}(x,t)=[\mathbf{x}^{(1)}(x,t);\mathbf{x}^{(2)}(x,t);\cdots;\mathbf{x}^{(k)}(x,t)], 
\end{equation}
is equivalent to the $k$ conditions
\begin{equation}
\left(\sum_{\ell=0}^m\mathbf{x}^{(\ell+1)}(x,t)^\top(\lambda+\ii)^\ell\right)\mathbf{U}(\lambda;x,t)\mathbf{q}^{(2)}\quad\text{vanishes at $\lambda=-\ii$ to order $m+1$},\quad m=0,\dots,k-1.
\end{equation}

It is easy to check that all $k$ of these conditions are actually implied just by the final condition for $m=k-1$.  To obtain $\mathbf{x}^{(j)}(x,t)$ such that this final condition holds, we first introduce notation $\{\mathbf{v}^{(j)\pm}_\ell(x,t)\}_{\ell=0}^\infty$, $j=1,2$, for the Taylor coefficients of $\mathbf{U}(\lambda;x,t)\mathbf{q}^{(j)}$ about $\lambda=\pm\ii$:
\begin{equation}
\mathbf{U}(\lambda;x,t)\mathbf{q}^{(j)} = \sum_{\ell=0}^\infty\mathbf{v}^{(j)\pm}_\ell(x,t)(\lambda\mp\ii)^\ell, \quad j=1,2.
\label{eq:Uq-Taylor}
\end{equation}
Then we choose $\mathbf{x}^{(j)}(x,t):=\ii\sigma_2\mathbf{v}^{(2)-}_{j-1}(x,t)$.  With this choice,
\begin{multline}
\left(\sum_{\ell=0}^{k-1}x^{(\ell+1)}(x,t)^\top(\lambda+\ii)^\ell\right)\mathbf{U}(\lambda;x,t)\mathbf{q}^{(2)}\\
\begin{aligned}
&=
\left(\sum_{\ell=0}^{k-1}(\ii\sigma_2\mathbf{v}_\ell^{(2)-}(x,t))^\top(\lambda+\ii)^\ell\right)\sum_{\ell=0}^\infty\mathbf{v}_\ell^{(2)-}(\lambda+\ii)^\ell\\
&=-\left(\sum_{\ell=0}^{k-1}\mathbf{v}_\ell^{(2)-}(x,t)(\lambda+\ii)^\ell\right)^\top\ii\sigma_2\left(\sum_{\ell=0}^{k-1}\mathbf{v}_\ell^{(2)-}(x,t)(\lambda+\ii)^\ell +O((\lambda+\ii)^{k})\right)\\
&=O((\lambda+\ii)^{k}),\quad\lambda\to-\ii,
\end{aligned}
\end{multline}
because $\mathbf{x}^\top\ii\sigma_2\mathbf{x}=0$ for all vectors $\mathbf{x}\in\mathbb{C}^2$.  Therefore,
\begin{equation}
\mathbf{x}^{(j)}(x,t):=\ii\sigma_2\mathbf{v}_{j-1}^{(2)-}(x,t),\quad j=1,\dots,k\quad\implies\quad
\text{the residue conditions \eqref{eq:k-even-singularity} hold.}
\label{eq:x-vector-define-even}
\end{equation}

With the matrix $\mathbf{X}(x,t)$ so-determined, the kernel conditions \eqref{eq:k-even-kernel} 
imply a square linear system on the elements of the $2\times k$ matrix $\mathbf{Y}(x,t)$.  In fact, we shall obtain a system of size $k\times k$ separately for each of the two rows of $\mathbf{Y}(x,t)$.  This is important because only the first row of $\mathbf{Y}(x,t)$ is needed to construct the rogue wave $\psi_k(x,t)$ of order $k=2n$, since combining \eqref{eq:rogue-wave-recover-2} with \eqref{eq:rational-ansatz} yields the formula
\begin{equation}
\psi_k(x,t)=1-2\ii\sum_{j=1}^{k}Y_{1j}(x,t)X_{2j}(x,t) = 1+\mathbf{a}(x,t)^\top \mathbf{y}(x,t),
\label{eq:psi-k-vector-formula}
\end{equation}
where 
\begin{equation}
\begin{split}
\mathbf{a}(x,t)&:=[-2\ii X_{21}(x,t);-2\ii X_{22}(x,t);\cdots;-2\ii X_{2,k}(x,t)]^\top\\
\mathbf{y}(x,t)&:=[Y_{11}(x,t); Y_{12}(x,t); \cdots ; Y_{1,k}(x,t)]^\top.
\end{split}
\label{eq:y-column-vector}
\end{equation}

To express the kernel conditions \eqref{eq:k-even-kernel},
we need to find the Taylor expansion about $\lambda=\ii$ of the product $(\mathbb{I}-\mathbf{Y}(x,t)\mathbf{D}(\lambda)\mathbf{X}(x,t)^\top)\mathbf{U}(\lambda;x,t)\mathbf{q}^{(j)}$, which means combining \eqref{eq:Uq-Taylor} with the expansion
\begin{equation}
\mathbf{D}(\lambda)=\sum_{\ell=0}^\infty \mathbf{D}_\ell(\lambda-\ii)^\ell,
\end{equation}
where $\mathbf{D}_\ell$ is the $k\times k$ Toeplitz matrix with elements $D_{\ell,ij}:=\gamma_{\ell,i-j+1}\delta_{i\ge j}$, and where
\begin{equation}
\gamma_{\ell m}:=\frac{(-1)^\ell}{(2\ii)^{\ell+m}}\binom{\ell+m-1}{\ell},\quad \ell=0,1,2,3,\dots,\quad m=1,2,3,\dots.
\end{equation}
Then it is easy to see that
\begin{multline}
\left(\mathbb{I}-\mathbf{Y}(x,t)\mathbf{D}(\lambda)\mathbf{X}(x,t)^\top\right)\mathbf{U}(\lambda;x,t)\mathbf{q}^{(j)}=\\
\sum_{m=0}^{k-1}\left(\mathbf{v}^{(j)+}_m(x,t)-\mathbf{Y}(x,t)\sum_{\ell=0}^m\mathbf{D}_\ell\mathbf{X}(x,t)^\top\mathbf{v}_{m-\ell}^{(j)+}(x,t)\right)(\lambda-\ii)^m+O((\lambda-\ii)^{k}),\quad\lambda\to\ii,\quad j=1,2.
\end{multline}

Setting to zero the coefficients of $(\lambda-\ii)^m$ for $m=0,\dots,k-1$ to enforce the kernel conditions \eqref{eq:k-even-kernel} 
then yields a linear system on the column vector $\mathbf{y}(x,t)$ (cf., \eqref{eq:y-column-vector}) of the form
\begin{equation}
\widetilde{\mathbf{K}}(x,t)\mathbf{y}(x,t)=\mathbf{b}(x,t),
\label{eq:My=b}
\end{equation}
where
\begin{equation}
\widetilde{K}_{pq}(x,t):=\sum_{\ell=0}^{p-1}\sum_{m=1}^q\gamma_{\ell,q-m+1}\mathbf{x}^{(q)}(x,t)^\top\mathbf{v}^{(1)+}_{p-\ell}(x,t),\quad b_p(x,t):=v^{(1)+}_{p-1,1}(x,t),\quad 1\le p,q\le k,
\label{eq:M-b-even}
\end{equation}
in which $\mathbf{x}^{(q)}(x,t)$ is defined by \eqref{eq:x-vector-define-even}.  
Since $\mathbf{y}(x,t)=\widetilde{\mathbf{K}}(x,t)^{-1}\mathbf{b}(x,t)$, 
applying the matrix determinant lemma $\det(\widetilde{\mathbf{K}}+\mathbf{b}\mathbf{a}^\top)=(1+\mathbf{a}^\top\widetilde{\mathbf{K}}^{-1}\mathbf{b})\det(\widetilde{\mathbf{K}})$ to \eqref{eq:psi-k-vector-formula} yields the fundamental rogue wave of order $k=2n$ 
as a ratio of $k\times k$ determinants:
\begin{equation}
\psi_k(x,t)=
\frac{\det(\widetilde{\mathbf{K}}(x,t)+\mathbf{b}(x,t)\mathbf{a}(x,t)^\top)}{\det(\widetilde{\mathbf{K}}(x,t))},\quad k=2n.
\label{eq:psi-k-even-case}
\end{equation}

Recall now the explicit formula \eqref{eq:entire} for $\mathbf{U}(\lambda;x,t)$.  It is natural to introduce the Taylor coefficients of the entire functions $(x+\lambda t)\sin(\theta)/\theta$ and $\cos(\theta)$ appearing in \eqref{eq:entire} about $\lambda=\ii$ as follows (recall $\theta=\rho(\lambda)(x+\lambda t)$):
\begin{equation}
S(\lambda;x,t):=(x+\lambda t)\frac{\sin(\theta)}{\theta} = \sum_{\ell=0}^\infty S_\ell(x,t)(\lambda-\ii)^\ell\quad\text{and}\quad
C(\lambda;x,t):=\cos(\theta)=\sum_{\ell=0}^\infty C_\ell(x,t)(\lambda-\ii)^\ell.
\label{eq:SC-expand}
\end{equation}
Since $S(\lambda;x,t)$ and $C(\lambda;x,t)$ are both Schwarz-symmetric functions of $\lambda$ for $(x,t)\in\mathbb{R}^2$, it also holds that
\begin{equation}
S(\lambda;x,t) = \sum_{\ell=0}^\infty S_\ell(x,t)^*(\lambda+\ii)^\ell\quad\text{and}\quad
C(\lambda;x,t)=\sum_{\ell=0}^\infty C_\ell(x,t)^*(\lambda+\ii)^\ell.
\label{eq:SC-conjugate-expand}
\end{equation}
Since $\lambda=\pm\ii +(\lambda\mp \ii)$, combining \eqref{eq:Q-define} with \eqref{eq:entire} and \eqref{eq:Uq-Taylor} gives
\begin{equation}
\mathbf{v}^{(1)+}_\ell(x,t)=\frac{1}{\sqrt{2}}\left(-(2S_\ell(x,t)-\ii S_{\ell-1}(x,t))\begin{bmatrix}-1\\1\end{bmatrix}+C_\ell(x,t)\begin{bmatrix}1\\1\end{bmatrix}\right),\quad \ell=0,1,2,\dots,
\end{equation}
and
\begin{equation}
\mathbf{v}^{(2)-}_\ell(x,t)=\frac{1}{\sqrt{2}}\left((2S_\ell(x,t)-\ii S_{\ell-1}(x,t))^*\begin{bmatrix}1\\1\end{bmatrix}+C_\ell(x,t)^*\begin{bmatrix}-1\\1\end{bmatrix}\right),\quad \ell=0,1,2,\dots,
\end{equation}
where we have adopted the notational convention that $S_{-1}(x,t):=0$. Actually, one can notice that $\widetilde{S}_\ell(x,t):=2S_\ell(x,t)-\ii S_{\ell-1}(x,t)$ are the Taylor coefficients of the function $\widetilde{S}(\lambda;x,t):=(1-\ii\lambda)S(\lambda;x,t)$ about $\lambda=\ii$.  Further setting $\widetilde{S}_\ell(x,t)=(\tfrac{1}{2}\ii)^\ell F_\ell(x,t)$ and $C_\ell(x,t)=(\tfrac{1}{2}\ii)^\ell G_\ell(x,t)$ for $\ell=0,1,2,\dots$, using \eqref{eq:x-vector-define-even} in \eqref{eq:M-b-even} gives
\begin{equation}
\begin{split}
\widetilde{K}_{pq}(x,t)&=\sum_{\ell=0}^{p-1}\sum_{m=1}^q\gamma_{\ell,q-m+1}(-\tfrac{1}{2}\ii)^{m-1}(\tfrac{1}{2}\ii)^{p-\ell-1}\left(F_{m-1}(x,t)^*F_{p-\ell-1}(x,t)+G_{m-1}(x,t)^*G_{p-\ell-1}(x,t)\right)\\
&=
(\tfrac{1}{2}\ii)^{p-1}(-\tfrac{1}{2}\ii)^q\sum_{\mu=0}^{p-1}\sum_{\nu=0}^{q-1}\binom{\mu+\nu}{\mu}
\left(F_{q-\nu-1}(x,t)^*F_{p-\mu-1}(x,t)+G_{q-\nu-1}(x,t)^*G_{p-\mu-1}(x,t)\right),
\end{split}
\end{equation}
for $1\le p,q\le k$ and 
\begin{equation}
b_p(x,t)=\frac{1}{\sqrt{2}}(\tfrac{1}{2}\ii)^{p-1}(F_{p-1}(x,t)+G_{p-1}(x,t)),\quad p=1,\dots,k.
\end{equation}
Finally, we combine \eqref{eq:x-vector-define-even} with \eqref{eq:y-column-vector} and $X_{2j}(x,t)=x^{(j)}_2(x,t)$, $j=1,\dots,k$, to find
\begin{equation}
\begin{split}
a_p(x,t)&=-2\ii X_{2p}(x,t)\\&=-2\ii x^{(p)}_2(x,t)\\&=2\ii v^{(2)-}_{p-1,1}(x,t)\\&=\frac{2\ii}{\sqrt{2}}(-\tfrac{1}{2}\ii)^{p-1}(F_{p-1}(x,t)^*-G_{p-1}(x,t)^*),\quad p=1,\dots,k.
\end{split}
\end{equation}
Therefore,
\begin{equation}
(\mathbf{b}(x,t)\mathbf{a}(x,t)^\top)_{pq}=-2(\tfrac{1}{2}\ii)^{p-1}(-\tfrac{1}{2}\ii)^{q}(F_{p-1}(x,t)+G_{p-1}(x,t))(F_{q-1}(x,t)^*-G_{q-1}(x,t)^*),\quad 1\le p,q\le k.
\end{equation}
By factoring off the invertible diagonal multipliers $\mathrm{diag}((\tfrac{1}{2}\ii)^{p-1})_{p=1}^k$ and $\mathrm{diag}((-\tfrac{1}{2}\ii)^{q})_{q=1}^k$ on the left and right, respectively, from $\widetilde{\mathbf{K}}(x,t)$ and $\mathbf{b}(x,t)\mathbf{a}(x,t)^\top$, one then sees that the formula \eqref{eq:psi-k-even-case} coincides with \eqref{eq:psi-k-determinants} in the case $k=2n$.

\subsection*{The rational ansatz for $k=2n-1$ with $n\in\mathbb{Z}_{>0}$}
In this case, we modify the rational ansatz by assuming that
\begin{equation}
\mathbf{N}^{(k)}(\lambda;x,t)=\left(\frac{\lambda+\ii}{\lambda-\ii}\right)^n\left(\mathbb{I}-\sigma_3\mathbf{Y}(x,t)\mathbf{D}(\lambda)\mathbf{X}(x,t)^\top\sigma_3\right)\mathbf{T}(\lambda),\quad\text{$\lambda$ in the exterior of $\Sigma_\circ$,}
\label{eq:ansatz-odd}
\end{equation}
where again $\mathbf{X}(x,t)$ and $\mathbf{Y}(x,t)$ are $2\times k$ matrices to be determined and  the $k\times k$ Toeplitz matrix $\mathbf{D}(\lambda)$ is defined by \eqref{eq:D-matrix-define}, and where
\begin{equation}
\mathbf{T}(\lambda):=\mathbb{I}-\frac{2\ii}{\lambda+\ii}\mathbf{q}^{(2)}\mathbf{q}^{(2)\top}=\mathbb{I}-\frac{\ii}{\lambda+\ii}\begin{bmatrix}1 & -1\\-1 & 1\end{bmatrix}.
\label{eq:T-define}
\end{equation}
As before, this ansatz is analytic in the exterior domain and satisfies $\mathbf{N}^{(k)}(\lambda;x,t)\to\mathbb{I}$ as $\lambda\to\infty$ for any choices of the matrices $\mathbf{X}(x,t)$ and $\mathbf{Y}(x,t)$.  Insisting that the jump condition \eqref{eq:N-jump-odd} relevant for $k=2n-1$ holds on $\Sigma_\circ$, thereby connecting the ansatz \eqref{eq:ansatz-odd} having poles of order $n$ at $\lambda=\pm\ii$ with a matrix analytic in the interior domain, we arrive at the following conditions on $\mathbf{X}(x,t)$ and $\mathbf{Y}(x,t)$:   the residue conditions that
\begin{equation}
\left(\mathbb{I}-\sigma_3\mathbf{Y}(x,t)\mathbf{D}(\lambda)\mathbf{X}(x,t)^\top\sigma_3\right)\mathbf{T}(\lambda)\mathbf{U}(\lambda;x,t)\mathbf{q}^{(1)}\quad\text{should be analytic at $\lambda=-\ii$, and}
\label{eq:odd-residue}
\end{equation}
the kernel conditions that
\begin{equation}
\left(\mathbb{I}-\sigma_3\mathbf{Y}(x,t)\mathbf{D}(\lambda)\mathbf{X}(x,t)^\top\sigma_3\right)\mathbf{T}(\lambda)\mathbf{U}(\lambda;x,t)\mathbf{q}^{(2)}\quad\text{should vanish to order $2n$ at $\lambda=\ii$.}
\label{eq:odd-kernel}
\end{equation}
Assuming that $\mathbf{X}(x,t)$ and $\mathbf{Y}(x,t)$ are chosen so that these conditions hold, $\mathbf{N}^{(k)}(\lambda;x,t)$ is the solution of Riemann-Hilbert Problem~\ref{rhp:renormalized}, and the rogue wave of order $k=2n-1$ is given by combining \eqref{eq:rogue-wave-recover-2} with \eqref{eq:ansatz-odd}:
\begin{equation}
\psi_k(x,t)=-1+2\ii\sum_{j=1}^kY_{1j}(x,t)X_{2j}(x,t)=-(1+\mathbf{a}(x,t)^\top\mathbf{y}(x,t)),
\label{eq:psi-k-odd}
\end{equation}
where the $k\times 1$ column vectors $\mathbf{a}(x,t)$ and $\mathbf{y}(x,t)$ are given in terms of $\mathbf{X}(x,t)$ and $\mathbf{Y}(x,t)$ by \eqref{eq:y-column-vector}.

\begin{lemma}
The matrix $\mathbf{T}(\lambda)$ has the following properties:
\begin{itemize}
\item[(i)] $\sigma_3\mathbf{T}(\lambda)$ commutes with $\mathbf{U}(\lambda;x,t)$:
\begin{equation}
\sigma_3\mathbf{T}(\lambda)\mathbf{U}(\lambda;x,t)=\mathbf{U}(\lambda;x,t)\sigma_3\mathbf{T}(\lambda).
\end{equation}
\item[(ii)] $\mathbf{T}(\lambda)$ acts on the columns of $\mathbf{Q}$ as follows:
\begin{equation}
\mathbf{T}(\lambda)\mathbf{q}^{(1)}=\mathbf{q}^{(1)}\quad\text{and}\quad
\mathbf{T}(\lambda)\mathbf{q}^{(2)}=\frac{\lambda-\ii}{\lambda+\ii}\mathbf{q}^{(2)}.
\end{equation}
\end{itemize}
\label{lemma:T}
\end{lemma}
The proof is elementary, combining the definition \eqref{eq:T-define} of $\mathbf{T}(\lambda)$ with the formul\ae\ \eqref{eq:entire} and \eqref{eq:Q-define} for $\mathbf{U}(\lambda;x,t)$ and $\mathbf{Q}$ respectively.  Using Lemma~\ref{lemma:T} and the fact that 
\begin{equation}
\sigma_3\mathbf{q}^{(1)}=-\mathbf{q}^{(2)}\quad\text{and}\quad\sigma_3\mathbf{q}^{(2)}=-\mathbf{q}^{(1)},
\end{equation}
we rewrite the residue conditions \eqref{eq:odd-residue} as
\begin{multline}
\left(\mathbb{I}-\sigma_3\mathbf{Y}(x,t)\mathbf{D}(\lambda)\mathbf{X}(x,t)^\top\sigma_3\right)\mathbf{T}(\lambda)\mathbf{U}(\lambda;x,t)\mathbf{q}^{(1)}\\
\begin{aligned}
&=
\left(\mathbb{I}-\sigma_3\mathbf{Y}(x,t)\mathbf{D}(\lambda)\mathbf{X}(x,t)^\top\sigma_3\right)
\sigma_3\mathbf{U}(\lambda;x,t)\sigma_3\mathbf{T}(\lambda)\mathbf{q}^{(1)}\\
&=\left(\mathbb{I}-\sigma_3\mathbf{Y}(x,t)\mathbf{D}(\lambda)\mathbf{X}(x,t)^\top\sigma_3\right)
\sigma_3\mathbf{U}(\lambda;x,t)\sigma_3\mathbf{q}^{(1)}\\&=
-\left(\mathbb{I}-\sigma_3\mathbf{Y}(x,t)\mathbf{D}(\lambda)\mathbf{X}(x,t)^\top\sigma_3\right)
\sigma_3\mathbf{U}(\lambda;x,t)\mathbf{q}^{(2)}\quad\text{should be analytic at $\lambda=-\ii$,}
\end{aligned}
\label{eq:odd-residue-rewrite}
\end{multline}
or, since $\mathbf{U}(\lambda;x,t)$ is entire and $\sigma_3^2=\mathbb{I}$, equivalently,
\begin{equation}
\mathbf{Y}(x,t)\mathbf{D}(\lambda)\mathbf{X}(x,t)^\top\mathbf{U}(\lambda;x,t)\mathbf{q}^{(2)}\quad
\text{should be analytic at $\lambda=-\ii$.}
\label{eq:odd-residue-rewrite-again}
\end{equation}
It is therefore \emph{sufficient} to enforce the condition \eqref{eq:sufficient-condition}.
Similarly, we rewrite the kernel conditions \eqref{eq:odd-kernel} as
\begin{multline}
\left(\mathbb{I}-\sigma_3\mathbf{Y}(x,t)\mathbf{D}(\lambda)\mathbf{X}(x,t)^\top\sigma_3\right)\mathbf{T}(\lambda)\mathbf{U}(\lambda;x,t)\mathbf{q}^{(2)}\\
\begin{aligned}
&=\left(\mathbb{I}-\sigma_3\mathbf{Y}(x,t)\mathbf{D}(\lambda)\mathbf{X}(x,t)^\top\sigma_3\right)\sigma_3\mathbf{U}(\lambda;x,t)\sigma_3\mathbf{T}(\lambda)\mathbf{q}^{(2)}\\
&=\frac{\lambda-\ii}{\lambda+\ii}\left(\mathbb{I}-\sigma_3\mathbf{Y}(x,t)\mathbf{D}(\lambda)\mathbf{X}(x,t)^\top\sigma_3\right)\sigma_3\mathbf{U}(\lambda;x,t)\sigma_3\mathbf{q}^{(2)}\\
&=-\frac{\lambda-\ii}{\lambda+\ii}\left(\mathbb{I}-\sigma_3\mathbf{Y}(x,t)\mathbf{D}(\lambda)\mathbf{X}(x,t)^\top\sigma_3\right)\sigma_3\mathbf{U}(\lambda;x,t)\mathbf{q}^{(1)}\quad\text{should vanish to order $2n$ at $\lambda=\ii$,}
\end{aligned}
\label{eq:odd-kernel-rewrite}
\end{multline}
or, equivalently, since $k=2n-1$ and $\sigma_3^2=\mathbb{I}$, the condition \eqref{eq:k-even-kernel} is required (although now with the interpretation that $k=2n-1$ is odd).  Since the residue and kernel conditions \eqref{eq:odd-residue}--\eqref{eq:odd-kernel} have thus been reduced to exactly the same conditions \eqref{eq:k-even-kernel}--\eqref{eq:sufficient-condition} as were used in the case $k=2n$ to determine the matrices $\mathbf{X}(x,t)$ and $\mathbf{Y}(x,t)$ (although again now taking $k$ to be odd), we may again define the columns of $\mathbf{X}(x,t)$ by \eqref{eq:x-vector-define-even} and obtain the first row of $\mathbf{Y}(x,t)$ rewritten as a column vector $\mathbf{y}(x,t)$ by the solution of the linear system \eqref{eq:My=b}--\eqref{eq:M-b-even}.  As in the case $k=2n$, we then obtain
\begin{equation}
\psi_k(x,t)=-\frac{\det(\widetilde{\mathbf{K}}(x,t)+\mathbf{b}(x,t)\mathbf{a}(x,t)^\top)}{\det(\widetilde{\mathbf{K}}(x,t))},\quad k=2n-1.
\end{equation}
Once again, by factoring off suitable diagonal multipliers, this formula is equivalent to \eqref{eq:psi-k-determinants} when $k$ is odd.

\subsection*{Solvability of the linear system \eqref{eq:My=b}--\eqref{eq:M-b-even}}
Note that the matrix $\widetilde{\mathbf{K}}(x,t)$ can be written in the form
\begin{multline}
\widetilde{\mathbf{K}}(x,t)=\mathrm{diag}((\tfrac{1}{2}\ii)^0,(\tfrac{1}{2}\ii)^1,\dots,(\tfrac{1}{2}\ii)^{k-1})\\{}\cdot
\left(\mathbf{F}(x,t)\mathbf{S}\mathbf{F}(x,t)^\dagger + \mathbf{G}(x,t)\mathbf{S}\mathbf{G}(x,t)^\dagger\right)
\mathrm{diag}((-\tfrac{1}{2}\ii)^1,(-\tfrac{1}{2}\ii)^2,\dots,(-\tfrac{1}{2}\ii)^{k}),
\end{multline}
where $\mathbf{F}(x,t)$ and $\mathbf{G}(x,t)$ are Toeplitz matrices:
\begin{equation}
\mathbf{F}(x,t):=\begin{bmatrix}F_0(x,t) & 0&\cdots & \cdots & 0\\
F_1(x,t) & F_0(x,t) & 0 &\cdots  &0\\
\vdots & \ddots & \ddots & \ddots & \vdots \\
F_{k-2}(x,t) & F_{k-3}(x,t) &\cdots & F_0(x,t) & 0\\
F_{k-1}(x,t) & F_{k-2}(x,t) & \cdots & F_1(x,t) & F_0(x,t)\end{bmatrix},
\end{equation}
\begin{equation}
\mathbf{G}(x,t):=\begin{bmatrix}G_0(x,t) & 0&\cdots & \cdots & 0\\
G_1(x,t) & G_0(x,t) & 0 &\cdots &0\\
\vdots & \ddots & \ddots &\ddots&\vdots\\
G_{k-2}(x,t) & G_{k-3}(x,t) & \cdots & G_0(x,t) & 0\\
G_{k-1}(x,t) & G_{k-2}(x,t) & \cdots & G_1(x,t) & G_0(x,t)\end{bmatrix},
\end{equation}
and $\mathbf{S}$ is the symmetric Pascal matrix of binomial coefficients:
\begin{equation}
\mathbf{S}:=\begin{bmatrix}\binom{0}{0} & \binom{1}{0} &\binom{2}{0}&\cdots & \binom{k-1}{0}\\
\binom{1}{1} & \binom{2}{1} & \binom{3}{1} &\cdots& \binom{k}{1}\\
\binom{2}{2} &\binom{3}{2} & \binom{4}{2} & \cdots &\binom{k+1}{2}\\
\vdots & \vdots &\vdots &&\vdots\\
\binom{k-1}{k-1} & \binom{k}{k-1} & \binom{k+1}{k-1} &\cdots &\binom{2k-2}{k-1}
\end{bmatrix}.
\end{equation}
Since $F_0(x,t)=\widetilde{S}_0(x,t)=2S_0(x,t)=2(x+\ii t)$ and $G_0(x,t)=C_0(x,t)=1$, $\mathbf{F}(x,t)$ is invertible except when $x=t=0$ and $\mathbf{G}(x,t)$ is invertible for all $(x,t)\in\mathbb{R}^2$.  Furthermore, $\mathbf{S}$ is positive definite because $\det((S_{jk})_{j,k=1,\dots,p})=1$ for all $p=1,\dots,k$.  Therefore, $\mathbf{F}(x,t)\mathbf{S}\mathbf{F}(x,t)^\dagger$ is positive semidefinite while $\mathbf{G}(x,t)\mathbf{S}\mathbf{G}(x,t)^\dagger$ is positive definite.  It follows that $\det(\widetilde{\mathbf{K}}(x,t))\neq 0$ for all $(x,t)\in\mathbb{R}^2$.  Therefore $\mathbf{Y}(x,t)$ exists for all $(x,t)\in\mathbb{R}^2$ and so the relevant rational ansatz (\eqref{eq:rational-ansatz} for $k=2n$ or \eqref{eq:ansatz-odd} for $k=2n-1$) furnishes a solution of Riemann-Hilbert Problem~\ref{rhp:renormalized}.  It is standard that the solution is unique if it exists.  This completes the proof of Proposition~\ref{prop:Equivalence}.


\begin{thebibliography}{99}
\bibitem{BilmanM17} D. Bilman and P. D. Miller, ``A robust inverse scattering transform for the focusing nonlinear Schr\"odinger equation,'' \texttt{arXiv:1710.06568}, 2017.
\bibitem{BilmanLMT18} D. Bilman, L. Ling, P. D. Miller, and A. Tovbis, ``High-order fundamental rogue waves in the far-field limit,'' in preparation, 2018.
\bibitem{BilmanT17} D. Bilman and T. Trogdon, ``Numerical Inverse Scattering for the Toda Lattice,'' \textit{Commun.\@ Math.\@ Phys.\@} \textbf{352}, 805--879, 2017.
\bibitem{ChesterFU57} C. Chester, B. Friedman, and F. Ursell, ``An extension of the method of steepest descents,'' \textit{Proc.\@ Cambridge Philos.\@ Soc.\@} \textbf{53}, 599--611, 1957.
\bibitem{DeiftZ93} P. Deift and X. Zhou, ``A steepest descent method for oscillatory Riemann-Hilbert problems.  Asymptotics for the mKdV equation,'' \textit{Ann.\@ Math.\@}, \textbf{137}, 295--368, 1993.
\bibitem{FokasIKN06} A. S. Fokas, A. R. Its, A. A. Kapaev, and V. Yu.\@ Novokshenov, \textit{Painlev\'e Transcendents.  The Riemann-Hilbert Approach}, Mathematical Surveys and Monographs, 128, American Mathematical Society, Providence RI, 2006.
\bibitem{GuoLL12} B. Guo, L. Ling, and Q. P. Liu, ``Nonlinear Schr\"odinger equation:  generalized Darboux transformation and rogue wave solutions,'' \textit{Phys.\@ Rev.\@ E} \textbf{85}, p. 026607, 2012.
\bibitem{JimboM81} M. Jimbo and T. Miwa, ``Monodromy preserving deformation of linear ordinary differential equations with rational coefficients. II,'' \textit{Physica D} \textbf{2}, 407--448, 1981.
\bibitem{KharifP03} C. Kharif and E. Pelinovsky, ``Physical mechanisms of the rogue wave phenomenon,'' \textit{Eur.\@ J. Mech.\@ B Fluids} \textbf{22}, 603--634, 2003.
\bibitem{Miller18} P. D. Miller, ``On the increasing tritronqu\'ee solutions of the Painlev\'e-II equation,'' \texttt{arXiv:1804.03173}, 2018.
\bibitem{DLMF} F. W. J. Olver, A. B. Olde Daalhuis, D. W. Lozier, B. I. Schneider, R. F. Boisvert, C. W. Clark, B. R. Miller, and B. V. Saunders, eds., NIST Digital Library of Mathematical Functions, \texttt{http://dlmf.nist.gov/}, Release 1.0.17, 2017.
\bibitem{Olver12} S. Olver, ``A general framework for solving Riemann-Hilbert problems numerically," \textit{Numer.\@ Math.\@} \textbf{122}, 305--340, 2012.
\bibitem{RHPackage} S. Olver, RHPackage, \texttt{http://www.maths.usyd.edu.au/u/olver/projects/RHPackage.html}, 2011.
\bibitem{Peregrine83} D. H. Peregrine, ``Water waves, nonlinear Schr\"odinger equations and their solutions,'' \textit{J. Aust.\@ Math.\@ Soc.\@ Ser.\@ B} \textbf{25}, 16--43, 1983.
\bibitem{Sakka09} A. H. Sakka, ``Linear problems and hierarchies of Painlev\'e equations,'' \textit{J. Phys.\@ A:  Math.\@ Theor.\@} \textbf{42}, 025210 (19 pp.\@), 2009.
\bibitem{ISTPackage} T. Trogdon and D. Bilman, ISTPackage, \texttt{https://bitbucket.org/trogdon/istpackage}, 2014.
\bibitem{TrogdonO13} T. Trogdon and S. Olver, ``Nonlinear steepest descent and numerical solution of Riemann-Hilbert problems," \textit{Comm.\@ Pure Appl.\@ Math.\@} \textbf{67} 8, 1353--1389, 2014.
\bibitem{TrogdonO16-book} T. Trogdon and S. Olver, \textit{Riemann-Hilbert problems, their numerical solution, and the computation of nonlinear special functions}, SIAM, Philadelphia, PA, 2016. 
\bibitem{TrogdonOD14} T. Trogdon, S. Olver, and B. Deconinck, ``Numerical inverse scattering for the Korteweg-de Vries and modified Korteweg-de Vries equations,'' \textit{Phys.\@ D} \textbf{241}(11), 1003--1025, 2014.
\bibitem{WangYWH17} L. Wang, C. Yang, J. Wang, and J. He, ``The height of an $n$-th-order fundamental rogue wave for the nonlinear Schr\"odinger equation,'' \textit{Phys.\@ Lett.\@ A} \textbf{381}, 1714--1718, 2017.
\bibitem{Zhou89a} X. Zhou, ``Direct and inverse scattering transforms with arbitrary spectral singularities,'' \textit{Comm.\@ Pure Appl.\@ Math.\@} \textbf{42}, 895--938, 1989.
\bibitem{Zhou89} X. Zhou, ``The Riemann-Hilbert problem and inverse scattering,'' \textit{SIAM J.\@ Math.\@ Anal.\@} \textbf{20}, 966--986, 1989.
\end{thebibliography}
\end{document}